\documentclass[12pt]{article}

\usepackage{amsmath}
\usepackage{theorem}
\usepackage{amssymb}
\usepackage{amsfonts}
\usepackage{calc}
\usepackage{picinpar}
\usepackage{ulem}
\usepackage{float}
\usepackage[dvips]{graphicx}
\newcommand{\dlt}[1]{\ensuremath{\Delta\Delta u}}
\usepackage[colorlinks=true,linkcolor=blue,hypertex,latex2html]{hyperref}

\newcounter{hours}\newcounter{minutes}
\newcommand{\printtime}{%
    \setcounter{hours}{\time/60}
    \setcounter{minutes}{\time-\value{hours}*60}
    \thehours h:\theminutes m}

{\theoremstyle{break} \theorembodyfont{\itshape}
\newtheorem{thm}{Theorem}[section]}

{\theoremstyle{break} \theorembodyfont{\itshape}
}

{\theoremstyle{plain} \theorembodyfont{\itshape}
\newtheorem{cor}[thm]{Corollary}}

{\theoremstyle{plain} \theorembodyfont{\itshape}
}

{\theoremstyle{plain} \theorembodyfont{\itshape}
}

{\theoremstyle{break} \theorembodyfont{\rmfamily}
}

{\theoremstyle{break} \theorembodyfont{\itshape}
\newtheorem{com}{Comment}[section]}

{\theoremstyle{break} \theorembodyfont{\itshape}
\newtheorem{maxwell}{\hyperlink{maxker}{Maxwell}}[section]}

{\theoremstyle{break} \theorembodyfont{\itshape}
\newtheorem{define}{Definition}[section]}

{\theoremstyle{break} \theorembodyfont{\itshape}
\newtheorem{lem}{Lemma}[section]}

{\theoremstyle{break} \theorembodyfont{\itshape}
\newtheorem{hyp}{Hypothesis}[section]}

\newenvironment{proof}[1][Proof]{\textbf{#1.}~}{\rule{0.5em}{0.5em}}

\makeatletter
\def\@footer{\hfil{\footnotesize \today\ \ - \printtime}\hfil}
\def\ps@headings{%
    \let\@oddfoot\@footer
    \def\@oddhead{{\slshape\rightmark}\hfil\thepage}%
    \let\@mkboth\markboth
    \def\sectionmark##1{%
      \markright {\MakeUppercase{%
        \ifnum \c@secnumdepth >\m@ne
          \thesection\quad
        \fi
        ##1}}}}
\makeatother

\newcommand{\bs}[1]{\ensuremath{\boldsymbol{#1}}}

\newcommand{\defn}{{\ensuremath{:=}}}

\begin{document}
\pagestyle{headings} \thispagestyle{empty} \markboth{\protect\thepart}{\protect\thepart}
\renewcommand{\thepart}{\Roman{part}}
\numberwithin{equation}{section}
\normalem

\title{On the Depletion Effect in Colloids}
\author{
P. Kotelenez\footnote{Mathematics, Case Western Reserve University,
Cleveland, OH~44106, USA, \texttt{peter.kotelenez@case.edu}, 216 368 4838 (analog), 216 368 5163 (fax).}\\
M. J. Leitman\footnote{Mathematics, Case Western Reserve University, Cleveland, OH~44106, USA,
\texttt{marshall.leitman@case.edu}, 216 368 2890 (analog), 216 368 5163 (fax).}\\and\\
J. A. Mann\footnote{Chemical Engineeing, Case Western Reserve University, Cleveland, OH~44106, USA, \texttt{j.mann@case.edu}, 216 368 4122 (analog), 216 368 3016 (fax).} }
\date{}
\maketitle
\begin{center}\fbox{\date{VERSION: \today,\printtime}}\end{center}

\begin{abstract}Our object is to formulate and analyze a physically plausible and mathematically sound
model to better understand the phenomenon of clustering in colloids. The term \emph{depletion force} refers to the force (in the Newtonian sense) which is associated with the clustering. Our model
is stochastic but derived from a deterministic setup in a Newtonian setting. A mathematical transition from the deterministic dynamics of several large particles and infinitely many small particles
to a kinetic description of the stochastic motion of the large particles is available. Assuming that the empirical velocity distribution of the small particles is governed by a probability density,
the mean-field force on the large particles can be represented as the negative gradient of a scaled version of that density. The stochastic motion of the large particles can then be described by a
system of correlated Brownian motions. The scaling in the transition preserves a small parameter, the correlation length. From the limiting kinetic stochastic equations we compute the probability
flux rates for the difference in position between two large particles. We show that, for short times, two particles sufficiently close together tend to be attracted to each other. This agrees with
the \emph{depletion} phenomena observed in colloids. To quantify this effect, we extend the notion of van~Kampen's one-dimensional probability flux rate in an appropriate way to account for higher
dimensional effects.
\end{abstract}
\newpage
\tableofcontents \listoffigures
\newpage

\section{Introduction}\label{sec-intro}A kinetic model for Brownian motion was introduced by A. Einstein~\cite{EI} and by
M. von Smoluchowski~\cite{SM};  a corresponding dynamic model was analyzed by Uhlenbeck and Ornstein~\cite{UH}. The mathematical model, first formulated for a single large particle suspended in a
stationary liquid, can be generalized easily to models of several large particles under the assumption that their positions, as well as their velocities, are spatially independent. The requirement
of independence is a central assumption in the theory of stochastic differential equations, mathematical particle systems and their macroscopic limits. However, this assumption appears to be in
conflict with the fact that the ``suspended particles all float in the \emph{same} fluid."\footnote{Spohn~\cite{SP}, Part II, Section 7.2. Emphasis added.} This begs the question of how to
incorporate correlations for physically relevant models. We address this question by considering systems of two types of interacting particles --- large particles and small particles;  the terms
\emph{large} and \emph{small} here refer to their different masses. After a suitable transition from micro-scales to meso-scales, the positions, as well as the velocities, of the large particles
will perform Brownian motions.

If there were only a single large particle, the fluid around it would appear homogeneous and isotropic, which leads to a relatively simple statistical description of the displacement of that
particle as a result of the collisions with the small particles.  Under the assumption that the collisions of the large particle with the small particles are elastic, a Brownian motion as an
approximation to the position and velocity of the particle has been obtained by several authors (See the references in Kotelenez~\cite{KO3}). Under the same assumptions, if there are two large
particles, sufficiently far apart that the fluid around each may be considered homogeneous and isotropic, each would be expected to experience a Brownian motion as an approximation to its position
and velocity; moreover, these Brownian motions would be independent. (See Figure~(\ref{FigPartAB}, A) below.)

On the other hand, if the two large particles are very close together, the fluid around each will no longer appear homogeneous and isotropic. (See Figure~(\ref{FigPartAB}, B) below.) In fact, the
fluid between the two large particles will get \emph{depleted} in the sense that fewer small particles per unit volume will be found between them (See Asakura, S and Oosawa, F.~\cite{AS},
G\"{o}etzelmann, B., Evans, R. and Dietrich, S.~\cite{GO} and the references therein.)  Thus, if the two large particles are sufficiently close together, the collisions of the many small particles
with them cause their motions to become statistically correlated. The length scale at which this occurs is measured by the \textit{correlation length}, which we denote by~\(\sqrt{\varepsilon}\).

To capture this effect faithfully,  we analyze a model in which the elastic collisions between large and small particles are replaced by a mean-field interaction. The mean-field force is derived
from the probability density of the velocity field of the small particles. The scaling preserves the correlation length. The positions and velocities of \(N,\,N\geq 1,\) large particles and
infinitely many small particles are given by a deterministic system of coupled nonlinear equations with independent random initial conditions. The equations are coupled through the rescaled
mean-field force. Using coarse graining in space and time, Kotelenez~\cite{KO3} obtained \(N\) correlated Brownian motions in a scaling limit as an approximation to the positions of the \(N\) large
particles, where the spatial correlations are computed from the probability density of the velocity field of the small particles.
We consider here the correlated limiting diffusions, outlining the main features of our analysis but focusing mainly on the interaction between just two large particles~(\(N=2\)).

The rest of this paper is organized as follows.  In Section~\ref{sec-jam} we first describe the excluded volume model of Vrij and de~Hek~\cite{VR1}, which is equivalent to the model of Asakura and
Oosawa~\cite{AS}. We continue with a brief discussion of three experimental studies of the \textit{depletion force,} each more stochastically oriented than its predecessor. The final, and most
stochastically oriented, gives results that agree nicely with the hard sphere model of Vrij and de~Hek~\cite{VR2,VR1}. In Section~\ref{sec-pk} we outline the underlying interacting particle model
and the stochastic limit referred to above. The notation in Section~\ref{sec-jam} is consistent with the literature referred to in that section, but generally differs from that which we use in the
rest of this paper. In Section~\ref{sec-prelim} we set out the notational conventions used in the other sections. In this section we also provide a statement of the fundamental \textit{Principle of
Material Frame Indifference} together with a useful lemma. In Section~\ref{sec-partsys} our stochastic model is formulated in detail for the case of two large particles. This section culminates with
a derivation and description of the diffusion matrix that governs the stochastic process for the separation between the two large particles. In Section~\ref{sec-sep}, we continue with an extensive
discussion of this stochastic process. The depletion effect is revealed and quantified in Section~\ref{sec-clump} as a consequence of our model through our \(d-\)dimensional version of van~Kampen's
\textit{probability flux rate.}
Finally, in Section~\ref{sec-sum}, we summarize our results and lay out avenues of further research.

\section{Interacting Particles --- Depletion
Effect}\label{sec-jam} Consider a mixture of small, spherical particles of diameter $a_{s}$ and big, spherical particles of diameter $a_{b}$ so that $a_{s}\ll a_{b}.$ Assume that the number
fractions obey, $x_{b}\ll x_{a}$; the reference state is the very dilute solution with respect to the big particles. In this state, the probable separation between any two big particles is at least
$10a_{s}.$This separation is sufficient to ensure that the distribution of small particles around the each big particle is uniform with spherical symmetry. The configuration is
shown in Figure (\ref{FigPartAB}, A).%
\begin{figure}
[h]
\begin{center}
\includegraphics[
natheight=2.253200in, natwidth=6.306300in, height=1.3383in, width=3.7169in
]%
{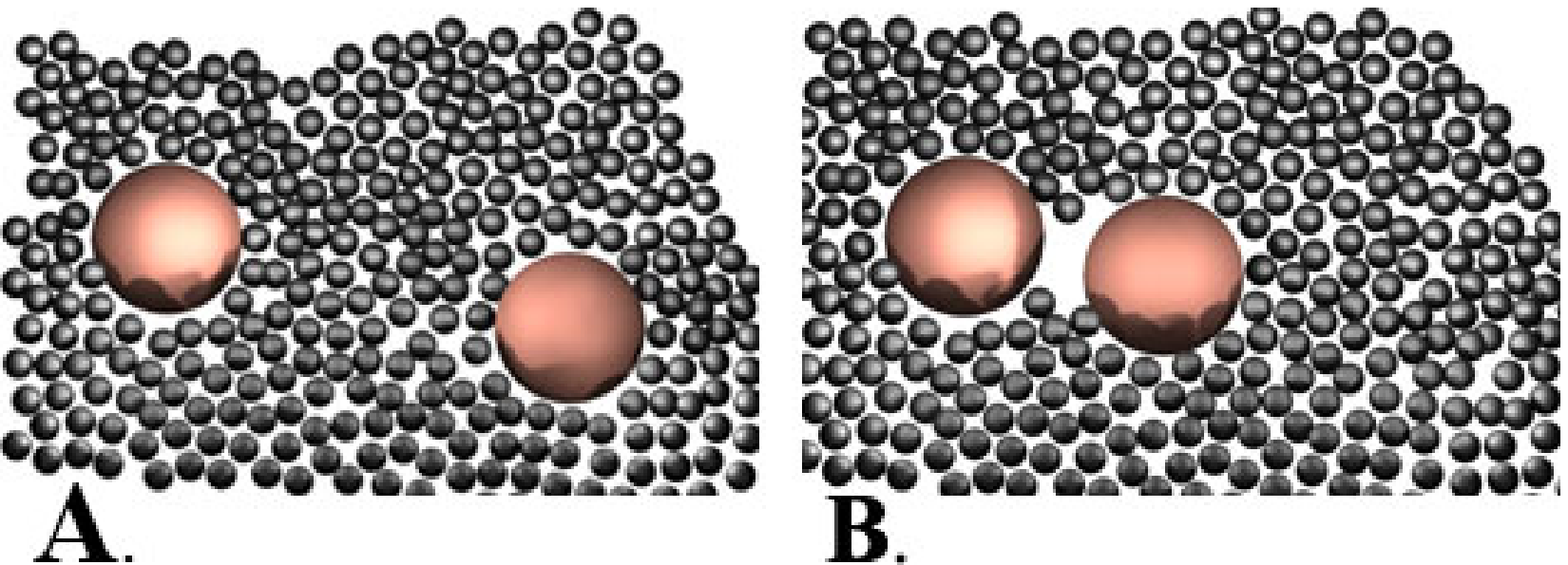}%
\caption[Two large and many small particles]{\newline\textbf{A. }Two large particles are separated sufficiently so that the packing around each particle is isotropic on average. \newline\textbf{B.
}If the two large particles are closer than about $2a_{s}$ the distribution of small
particles is depleted in the region between them.}%
\label{FigPartAB}%
\end{center}
\end{figure}

When two big particles move together the configuration space of the small particles changes when the separation is less than about $2a_{s}$ as represented in Figure (\ref{FigPartAB}, B). The
reduction of small particle density in the space between large particles that are sufficiently close is referred to as \textit{depletion.}

Experiments early in the last century showed that colloidal systems composed of dispersions of particles and macromolecules show interesting flocculation or agglomeration phenomena, which was not
understood until the paper of Asakura and Oosawa \cite{AS}. They describe the case of parallel plates immersed in a solution of hard, macromolecular particles. Then they assert that as the two
plates are brought close together an attractive osmotic force develops as a result of depletion of the macromolecules from the volume between the two plates.

Vrij and de Hek~\cite{VR2,VR1} independently rediscovered the depletion effect and carried out experimental and theoretical studies that involved coated silica spheres, the large particles, of about
50 nm radius in dispersions containing polystyrene, the small particles, mixed in cyclohexane. The marker for the depletion effect is the phase separation, which was observed with the variation of
polystyrene concentration. 

\subsection{The excluded volume model of Vrij}

The hard sphere model of Vrij~\cite{VR2} and Vrij and de~Hek~\cite{VR1} is equivalent to
that of Asakura and Oosawa \cite{AS} and is shown in Figure (\ref{VrijModel}%
). In many experiments the small sphere is a coiled polymer, which is approximated by assuming the polymer component behaves as a small hard-sphere when interacting with a large hard-sphere but
penetrable when interacting with another small ``hard-sphere." This is the \textit{penetrable hard sphere~(PHS)} model.
shown in Figure (\ref{VrijModel}).%

\begin{figure}
[h]
\begin{center}
\includegraphics[
natheight=7.106600in, natwidth=10.666600in, height=2.0506in, width=3.0751in
]%
{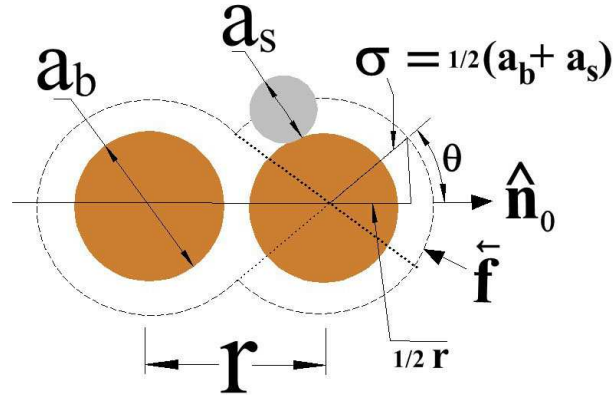}%
\caption[Vrij model]{The Vrij model \cite{VR1}, \cite{VR2}, while crude, provides
physical insight and fits dilute systems, see Figure (\ref{LaserTweezerLayout}%
). }%
\label{VrijModel}%
\end{center}
\end{figure}

Briefly, two large particles of diameter $a_{b}$ are moving toward collision and assume that the large particles are hard so that the force is unbounded when $r<a_{b}.$ The number density of the
small particles, of diameter $a_{s},$ is assumed sufficiently dilute that the osmotic pressure, $P_{os},$ is ideal, $P_{os}=c_{s}k_{B}T.$ The excluded volume produced by the hard sphere interaction
between the small and large particles is traced by the dashed curve of
radius $\sigma=%
\frac12
(a_{b}+a_{s})$ in Figure~\ref{VrijModel}. Notice that the normal, $\hat{n},$ points into the surrounding phase whereas the osmotic force, $\vec{f}=-P_{os}\hat{n}$, points inward. At a particular
separation, $r,$ the two excluded volumes overlap and $c_{s}$ is diminished there. The osmotic force, projected along the center to center axis, $\hat{n}_{0},$ is attractive, formula (\ref{Intg1},
\ref{K}).
Following Vrij, the depletion force, $K$ is estimated as%
\begin{equation}
K=\int_{A}\vec{f}\cdot\hat{n}_{0}dA.\label{Intg1}%
\end{equation}
Here, $\vec{f}\cdot\hat{n}_{0}=-P_{os}\cos\theta$ and from Figure
(\ref{VrijModel}) $\cos\theta=%
\frac12
r/\sigma$ so that following Vrij%
\begin{equation}
K(T,\sigma,c_{s},r)=\left\{
\begin{array}
[c]{c}%
\infty\quad\text{for }r<a_{b}\mathstrut\smallskip\\
-\pi\sigma^{2}(1-%
\frac14
\left(  \frac{r}{\sigma}\right)  ^{2})c_{s}k_{B}T\quad\text{for }a_{b}\leq
r\leq2\sigma\smallskip\\
0\quad\text{for }r>2\sigma
\end{array}
\right.  .\label{K}%
\end{equation}
The potential is
\begin{equation}
\begin{split}
V(T,\sigma,c_{s},r)&=\int_{r}^{2\sigma}K(T,\sigma,c_{s},\bar{r})\,d\bar{r}\\&=\left\{
\begin{array}
[c]{c}%
\infty\quad\text{for }r<a_{b}\smallskip\\
-\frac{4}{3}\pi\sigma^{3}(1-%
\frac34
(\frac{r}{\sigma})+\frac{1}{16}(\frac{r}{\sigma})^{3})c_{s}k_{B}%
T\quad\text{for }a_{b}\leq r\leq2\sigma\smallskip\\
0\quad\text{for }r>2\sigma
\end{array}
\right.  .\label{V}%
\end{split}
\end{equation}
Vrij and his collaborators certainly recognized that this model is crude, however, it suggested experiments that sharpen the physical chemistry of the depletion effect. The hard sphere physics can
be approximated by dressing small particles in layers that reduce the Van der Waals attractive force and the coulombic interactions that obtain in many colloid systems. See, for example, de Hek and
Vrij \cite{VR1}.

The papers by Henderson \textit{et al.}~\cite{HE}, G\"{o}tzelmann, \textit{et al.}~\cite{GO}, Zhou~\cite{ZH}, Biben, \textit{et al.}~\cite{BI2}, and Kinoshita~\cite{KI} represent the literature
wherein a fundamental approach is developed to investigated a multicomponent hard-sphere dispersion. Also, see the review by Tuinier, \textit{et al. }\cite{TU1}.

We did not find any articles that addressed the effect of depletion on the correlated, stochastic motion of hard-sphere, multicomponent dispersions. However, there are experimental techniques and
results that are pertinent to the theory developed herein.

\subsection{Some Experimental Studies of Depletion Forces}

As developed in the early work of Vrij and de~Hek~ \cite{VR2,VR1} and reviewed by Tuinier, \textit{et al.}~\cite{TU1}, the details of the phase behavior of polymer, colloid mixtures yield
information about depletion effects. For example, silica particles coated with stearyl alcohol dispersed in cyclohexane will provide close to index matching conditions so that the Van der Waals
force between the coated particles is ignorably small. Moreover, the effect of any electrical double layer is negligible. Such a system is very close to a dispersion of hard spheres. The interaction
of two large particles allows the determination of the depletion effect induced by the number density of the polymer component.

\subsubsection{Scattering techniques and colloid properties that show
depletion effects}

Static and dynamic light scattering measurements are capable of measuring depletion effects as outlined by Tuinier \textit{et al.}~\cite{TU2} who describe a small angle neutron scattering (SANS)
study of casein micelles and an exocellular polysaccharide system that behaves as a PHS system. They also report dynamic light scattering and turbidity results. Briefly, in the case of SANS, the
Rayleigh ratio $R(q)$ is determined experimentally as a function of the wave number,
\begin{equation}
q=2\left(  \frac{2\pi}{\lambda_{0}}\right)  \sin(\theta/2),\label{qDef}%
\end{equation}
where $\lambda_{0}$ is the vacuum wavelength of either photons or neutrons and $\theta$ is the scattering angle. See Berne~\cite{BE4}, for example. The Rayleigh ratio is related to the structure
factor, $S(q)$ and the particle form factor, $f(q)$ through the following formula
\begin{equation}
R(q)=Kmcf(q)S(q)\label{R(q)def}%
\end{equation}
where $K$ is a material coefficient, $m$ the mass of the particle, $c$ the concentration. The form factor, $f,$ depends on the particle radius as well as $q$ and can be computed theoretically (See
Berne~\cite{BE4}). The form factor may also depend on the energy of the incident photons (or neutrons). The structure factor is is the Fourier transform of the radial distribution function of the
particles, $g(r)$, that are scattering:%
\begin{equation}
S(q)=1+4\pi\rho\int_{0}^{\infty}(g(r)-1)r^{2}\frac{\sin(qr)}{qr}%
dr.\label{S(q)}%
\end{equation}
The function $g(r)$ is the probability density of finding a pair of particles separated by a distance $r.$ Therefore expect that $g(r)$ will depend on the number density of large particles and on
their interactions, including the depletion effect. See Tuinier \textit{et al.} \cite{TU1}, \cite{TU2}.

\begin{com}
Scattering functions such as formula (\ref{S(q)}) are functions of "reciprocal" space, which is defined through the Fourier transform, for
example%
\begin{equation}
\hat{f}(\vec{q})=\int f(\vec{x})e^{-i\vec{q}\cdot\vec{x}}d\vec{x}%
\label{RCPspace}%
\end{equation}
defines the "reciprocal" space representation of the function $f$ with domain
of definition in "direct" space, $%
\mathbb{R}
^{d}.$ The reciprocal space representation expresses the physics of the scattering phenomenon as a result, for example, of solving Maxwell's equations with appropriate boundary conditions. However,
it is often useful to Fourier transform the data so as to search for symmetry elements. In scattering experiments, formula (\ref{qDef}) is an expression of the vector momentum,
$\vec{p}=\hbar\vec{k},$ balance, for example, of a photon scattered by
fluctuations is%
\begin{align}
\hbar\vec{k}_{inc}  & =\hbar\vec{k}_{sct}\pm\hbar\vec{q}\label{MomentumBal}\\
q  & =k_{0}\left\vert \vec{e}_{inc}-\vec{e}_{sct}\right\vert \label{qMomBal}%
\end{align}
where $\vec{e}_{inc}$ is the direction of the incident beam and $\vec{e}%
_{sct}$ the direction of the scattered beam. The wavenumber $k_{0}$ is defined as $k_{0}=2\pi/\lambda_{0},$ where $\lambda_{0}$ is the wavelength of the incident beam.
\end{com}

When the scattering is kinematic, which obtains for single scattering events, Maxwell's equations teach that the field at a detector distant from the illuminated volume, $V_{I}$, which includes the
particles $(1,\cdots N),$ is
\begin{equation}
\varphi(q)=\sum_{n=1}^{(1,\cdots N)\subset V_{I}}f_{n}(q,E)e^{-i\vec{q}%
\cdot\vec{x}_{n}(t)}.\label{PhiDefn}%
\end{equation}
Here, $\vec{x}_{n}(t)$ is the trajectory of the $n$th particle in the volume $V_{I}$. If the system is monodispersed, all large spheres are of the same radius, $f_{n}(q,E)=f(q,E)$ independent of
$n.$ In that case
\begin{equation}
\varphi(q)=f(q,E)\sum_{n=1}^{(1,\cdots N)\subset V_{I}}e^{-i\vec{q}\cdot
\vec{x}_{n}(t)}\label{PhiMonodispersed}%
\end{equation}
and the sum, the intermediate scattering function, is the Fourier transform of
\begin{equation}
\chi(\vec{x})=\sum_{n=1}^{(1,\cdots N)\subset V_{I}}\delta(\vec{x}-\vec{x}%
_{n}(t)).\label{Chi(x)}%
\end{equation}
Note that $\{\vec{x}_{n}(t)\}$ is available from molecular dynamics simulation.

Dynamic light scattering is observable as a correlation function of the photocurrent, which amounts to an average of the square of the field at the
detector:%
\begin{align}
& <\varphi(q,t)\varphi^{\ast}(q^{\prime},t^{\prime})>=N<f(q,E)f^{\ast
}(q,E)>\nonumber\\
+  & <f(q,E)f^{\ast}(q,E)>\sum_{\substack{n,\acute{n}=1 \\n\neq\acute{n} }}^{(1,\cdots N)\subset V_{I}}<e^{-i\vec{q}\cdot(\vec{x}_{n}(t)-\vec
{x}_{\acute{n}}(t^{\prime}))}>.\label{xtcorrelation1}%
\end{align}
Experimentally, the self-beat method of dynamic light scattering provides a correlation function in time and wavenumber of the particles and is accurately
portrayed in very dilute solutions by $<\left\vert \vec{x}_{n}(t)-\vec{x}%
_{n}(0)\right\vert ^{2}>=6D_{0}t$ where the Stokes-Einstein diffusion
coefficient is $D_{0}=(k_{B}T/6\pi\eta)\times(%
\frac12
a).$ Here, $\eta$ is the viscosity of the medium, $a$ is the hydrodynamic diameter of the particle. In addition, the system is so dilute that each particle moves as an independent Brownian particle.
Indeed, experimental evidence supports the assumption that the dynamics are represented by the set
of independent Langevin equations:%
\begin{align}
d\vec{v}_{n}  & =-\beta\vec{v}_{n}dt+\vec{a}_{n}(dt)\label{Langevin1}\\
d\vec{x}_{n}  & =\vec{v}_{n}dt\nonumber
\end{align}
here $\beta=6\pi\eta\times(%
\frac12
a)/m$ \ and $m$ is the mass of the large particle. The random accelerations are connected to $\beta$ through a correlation function and a fluctuation dissipation theorem.\footnote{See van
Kampen~\cite{KA}} The correlation function is
\begin{equation}
<\vec{a}_{n}(t)\cdot\vec{a}_{\acute{n}}(t^{\prime})>=b\delta_{n\acute{n}%
}\delta(t-t^{\prime})\label{accelCorr}%
\end{equation}
and a fluctuation-dissipation theorem for this case gives $b=2\beta\frac {k_{B}T}{m}=2\beta^{2}D$.

Interparticle forces change the picture. Now the Langevin equation reads
\begin{align}
d\vec{v}_{n}  & =\vec{F}_{n}dt-\beta\vec{v}_{n}dt+\vec{a}_{n}%
(dt)\label{Langevin2}\\
d\vec{x}_{n}  & =\vec{v}_{n}dt\nonumber
\end{align}
where
\begin{equation}
\vec{F}_{n}(x_{n}(t))=\sum_{\substack{n^{\prime}=1 \\n\neq n^{\prime}}%
}^{All}\vec{F}_{n^{\prime}}(\vec{x}_{n}(t),\vec{x}_{n^{\prime}}(t))+\cdots
.\label{Forces}%
\end{equation}
It is often assumed that the pair forces dominate to the exclusion of triplet and higher contributions considered ignorable. A further simplification is to assume that $\vec{F}_{n^{\prime}}$ depends
only on the separation $\left\vert \vec{x}_{n}(t)-\vec{x}_{n^{\prime}}(t)\right\vert .$ When the particles satisfy the PHS model, it is clear that $\vec{F}_{n}$ will affect the correlation function
of dynamic light scattering as will also the necessarily different relationship between $\beta$ and $b.$

The focus of this paper is the form of a stochastic term analogous to that in formula (\ref{Langevin1}), which in the hard sphere case is shown to involve a depletion effect. The more general case
wherein intermolecular forces are other than hard sphere must still reflect the depletion effect derived herein.

\newpage

\subsubsection{Measurement of forces between two large particles and between a
large particle and a plate.}

It is now possible to measure the forces between a micron to submicron size particle directly against a flat plate or a second particle. The basic idea of
the measurement is shown in Figure (\ref{ForceExpModel}).%

\begin{figure}
[h]
\begin{center}
\includegraphics[
natheight=8.000000in, natwidth=10.666600in, height=2.088in, width=2.7812in
]%
{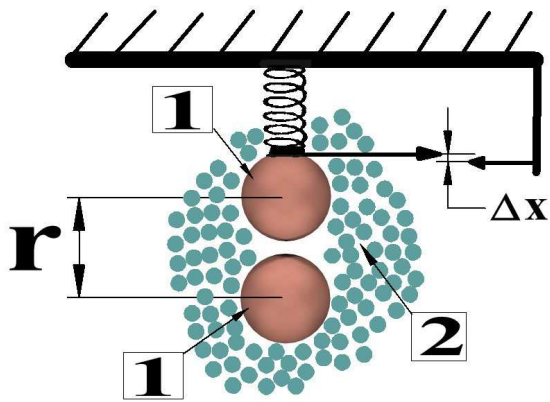}%
\caption[Conceptualization of the force measurement]{Conceptualization of the force measurement as the two particles are
forced together. }%
\label{ForceExpModel}%
\end{center}
\end{figure}
The force measurement as the two particles come together is conceptualized by an atomic force microscopy (AFM) measurement where the cantilever is replaced by a spring attached to the upper
particle. As the particles are brought to small separation, $h=r-a_{b}$, a change in $\Delta x$ is observed and transformed into a force since the spring constant is known. Here \(a_b\) is the
diameter of the large particles. The force experienced by the large particle is proportional to \(\Delta x\) scaled by the spring constant in the figure. The AFM experiment is usually done with a
flat plate instead of the lower particle shown here. The function $\Delta x$ will fluctuate in time as a result of impacts of the small particles on the large particles.

\paragraph{Atomic force microscopy (AFM)}

AFM experiments have produced interesting depletion data, see for example Knoben \textit{et al. }\cite{KN}, Clark, \textit{et al.}~\cite{CL}, and Tulpar, \textit{et al. }\cite{TU3}. However, AFM
experiments suffer from two related problems. When the attractive force is too large, the upper particle will snap to the lower particle or plate, therefore, the full force curve is not generally
available. Providing the optimum spring constant is difficult. It is also difficult to define accurately the point of contact between the sphere and substrate and therefore the actual separation
between the particle mounted on the cantilever and the substrate is poorly defined. Progress to resolve this difficulty is discussed in references \cite{CL} and \cite{MC}. Biggs, \textit{et
al.}~\cite{BI} point out that while TIRM (total internal reflection microscopy) is more sensitive, the technique cannot measure strong forces; the two techniques AFM and TIRM are complementary.

\paragraph{Total internal reflection microscopy}

(TIRM) has the advantage of using an evanescent wave to interrogate the position of a large particle with respect to an optical surface, see Tuinier \textit{et al.} \cite{TU1} for a brief
description of the technique and see references \cite{BE2}, \cite{RU}, \cite{OC} and \cite{BI} for reports of depletion effects observed. Also, see Figure
(\ref{TIRM}).%

\begin{figure}
[th]
\begin{center}
\includegraphics[
natheight=8.000000in, natwidth=8.173500in, height=1.7659in, width=1.8032in
]%
{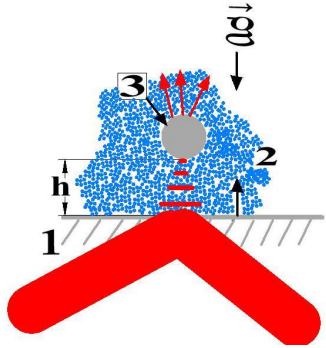}%
\caption[Laser beam]{The laser beam (in red) is incident on the flat interface at an angle consistent with total internal reflection. The evanescent wave, heavy (red) line, attenuates exponentially
into the solution phase. The intensity, I(h), of the
light scattered by the particle, (red) arrows is detected.}%
\label{TIRM}%
\end{center}
\end{figure}
In Figure (\ref{TIRM}), the intensity, I(h), of the light scattered by the particle (red arrows) is detected and follows the rule $I(h)=I(0)\exp(-\Lambda h)$ where $1/\Lambda$ is the $1/e$
penetration depth of the electric field. The coefficient $\Lambda$ depends on the refractive index numbers $n_{1}$ and $n_{2}$. The scattering intensity depends on $n_{2}$ and $n_{3}.$ In addition
to the depletion forces (and colloid forces in general), the particles sediment, $\vec{g}$ is the acceleration of gravity.

Consider Figure (\ref{TIRMfluct1}), taken from Bechinger \textit{et al.}~\cite{RU}, which shows the variation of fluctuation amplitudes with
concentration of the small particles.%

\begin{figure}
[h]
\begin{center}
\includegraphics[
natheight=3.808200in, natwidth=5.388900in, height=2.2225in, width=3.1382in,keepaspectratio=true
]%
{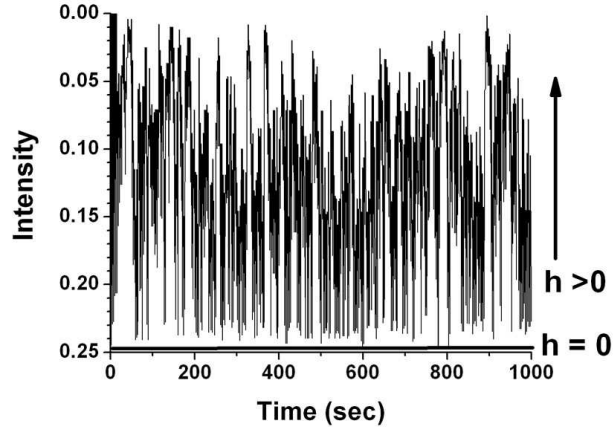}%
\caption[TIRM fuluctuations]{TIRM fuluctuations in intensity observed when small spheres are
absent.}%
\label{TIRMfluct1}%
\end{center}
\end{figure}

The authors used a dilute solution of polystyrene spheres of 1.5$\mu m$ radius dispersed in water. Polyethylene oxide with molecular weight of $2\times 10^{-6}$ Dalton was added, which had a radius
of gyration of $0.01\mu m$ and behaved as a small hard spheres. No polymer was added for the time-series shown in Figure (\ref{TIRMfluct1}), so that the potential constructed from $p(h)$, the
probability density that $h$ will occur, generated from this time series, gave the function shown in Figure (\ref{TIRM2}) for $\varphi_{S}=0.0$. The raw data for $\varphi_{S}=0.032$ were probably
similar to that of Figure (\ref{TIRMfluct1}) but was not reported in Bechinger \textit{et al.} \cite{RU}. Note that at $h=0,$ the intensity of the scattered light will reach a maximum. The bar just
above the $0.25$ tick indicates the intensity for $h=0.$ Figure (\ref{TIRMfluct1}) shows the change in the amplitude of particle fluctuations along the normal to the glass substrate. They convert
the data of Figure (\ref{TIRMfluct1}) to potentials through the determination of $p(h),$ then, the potential, $V(h),$essentially a free energy, is computed through the canonical ensemble (Helmholtz)
distribution density taken as
\begin{equation}
p(h)=ce^{-V(h)/k_{B}T}.\label{Helmholtz}%
\end{equation}
Unfortunately, these and other authors did not determine the space-time correlation functions of these systems.

\begin{com}
Since the Hamiltonian of the B-particle observed in the evanescent optical field is of the form $H=\frac{1}{2m}\vec{p}\cdot\vec{p}+V(h),$ the joint probability density factors so that formula
(\ref{Helmholtz}) obtains where $c$ is the inverse partition function, $\frac{1}{c}=\int_{0}^{\infty }e^{-V(h)/k_{B}T}dh.$ Note that the particle is constrained to $h\geq0,$see Figure (\ref{TIRM}).
Formula (\ref{Helmholtz}) is considered exact, see Henderson et al. \cite{HE}.
\end{com}

A second paper, by Bechinger \textit{et al.} \cite{RU} showed an interesting pattern; the AO - Vrij model works reasonably at low concentrations of polymer but a maximum is found at higher
concentrations. See
Figure(\ref{TIRM2}).%

\begin{figure}
[h]
\begin{center}
\includegraphics[
natheight=3.043600in, natwidth=7.577400in, height=2.0208in, width=5.0062in
]%
{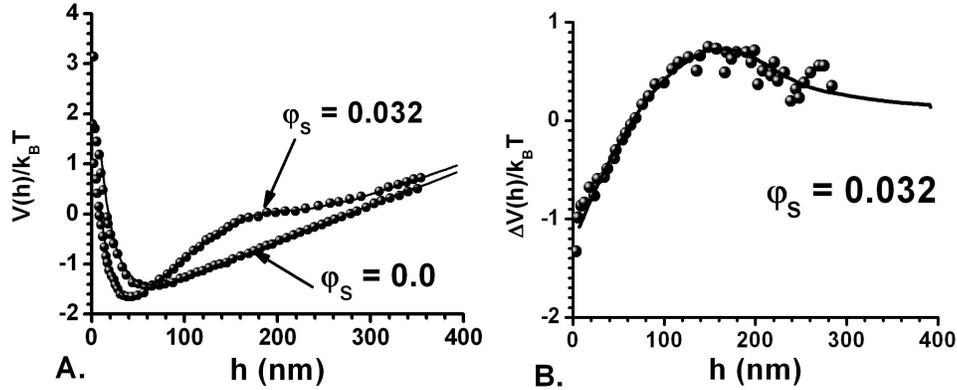}%
\caption[Potential measurements]{\textbf{A.} The total potential is shown and is considered to be a sum of double layer, sedimentation, Van der Waals and depletion effects. The variation with
separation between the large particle and the plate is $h$. The polymer concentrations (number density) are $n=0,$ or a volume fraction of $\varphi_{S}=0.0$ and $n=6.3\mu m^{-3}$ or
$\varphi_{S}=0.032.$ \textbf{B.} Shows the difference between the data of the curves $\varphi_{S}=0.032$ and $0.0$, which represents the depletion component of the potential. (Redrawn
from Figure 1 of Bechinger \cite{RU}.)}%
\label{TIRM2}%
\end{center}
\end{figure}

That a potential barrier develops at higher concentrations of polymer is further supported by experiments using laser tweezer technology. Also see Zhou \cite{ZH}.

Oetama and Walz \cite{OC} report the result of a study of short-time dynamics with a focus on determining the diffusion coefficient of the large particles in the presence of the small particles.

\paragraph{Laser tweezer results}

Particles can be trapped in the focal volume of a laser if the refractive index difference between the particle and liquid is of sufficient magnitude. Once trapped, the force that the particle
experiences with respect to external sources can be measured down to nanometer length scales and sub-picoNewton forces, for example, between two$\ B$-particles. See Figure
(\ref{LaserTweezer}).%

\begin{figure}
[h]
\begin{center}
\includegraphics[
natheight=2.840200in, natwidth=5.333300in, height=1.6073in, width=3.0062in
]%
{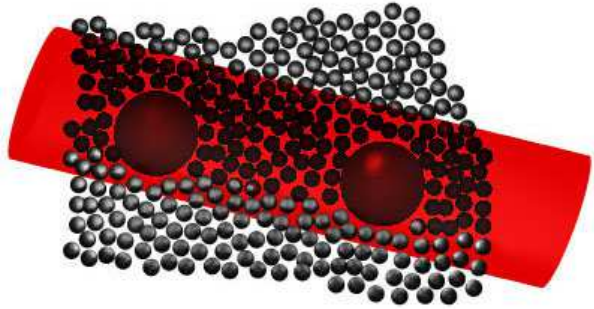}%
\caption[Laser tweezer concept]{The focus of a laser beam is rastered at a sufficient rate that the particle is trapped to move in a tube of length $L$. $L$ is sufficiently large that the depletion
force is ignorably small for that maximum separation. The Brownian motion of the particles is tracked in time. Reference \cite{FA} and papers cited therein provide details of the measurement
technique.}%
\label{LaserTweezer}%
\end{center}
\end{figure}

Crocker, \textit{et al.}~\cite{CR} used a tweezer system that scans a roughly $10\mu m$ trapping line in a thin sample cell. In a typical run, a pair of large spheres $(a_{b}=1.100$ $\pm0.015\mu m$
diameter$)$ of PMMA (polymethylmethacrylate) was trapped in the tweezer and one-dimensional Brownian motion was monitored. Various volume fractions of $a_{s}=0.083\mu m$ diameter polystyrene spheres
(ps) provided the background. Here too, the pair potential is computed from the observed distribution of separations of a pair of particles through the formula
\begin{equation}
p(|\vec{x}_{1}-\vec{x}_{2}|)=ce^{-V(|\vec{x}_{1}-\vec{x}_{2}|)/k_{B}%
T}.\label{Pdensity(r)}%
\end{equation}

\begin{com}
In this case, the Hamiltonian is $H=\frac{1}{2m}\vec{p}_{1}\cdot\vec{p}%
_{1}+\frac{1}{2m}\vec{p}_{2}\cdot\vec{p}_{2}+V(\vec{x}_{1},\vec{x}_{2}).$ Assume that $V(\vec{x}_{1},\vec{x}_{2})=V(|\vec{x}_{1}-\vec{x}_{2}|),$ then the maximum entropy principal provides formula
(\ref{Pdensity(r)}). This formula is considered exact given the form of $V,$ \cite{HE}. Triplet depletion forces may be important enough to be considered, see Melchionna and Hansen \cite{ME}.
\end{com}

At low volume fractions of ps, they confirmed the model of AO-Vrij as shown in
Figure (\ref{LaserTweezerLayout} A.).%

\begin{figure}
[h]
\begin{center}
\includegraphics[
 natheight=3.3217in, natwidth=7.5392in,height=1.8431in, width=4.166in,keepaspectratio=true
]%
{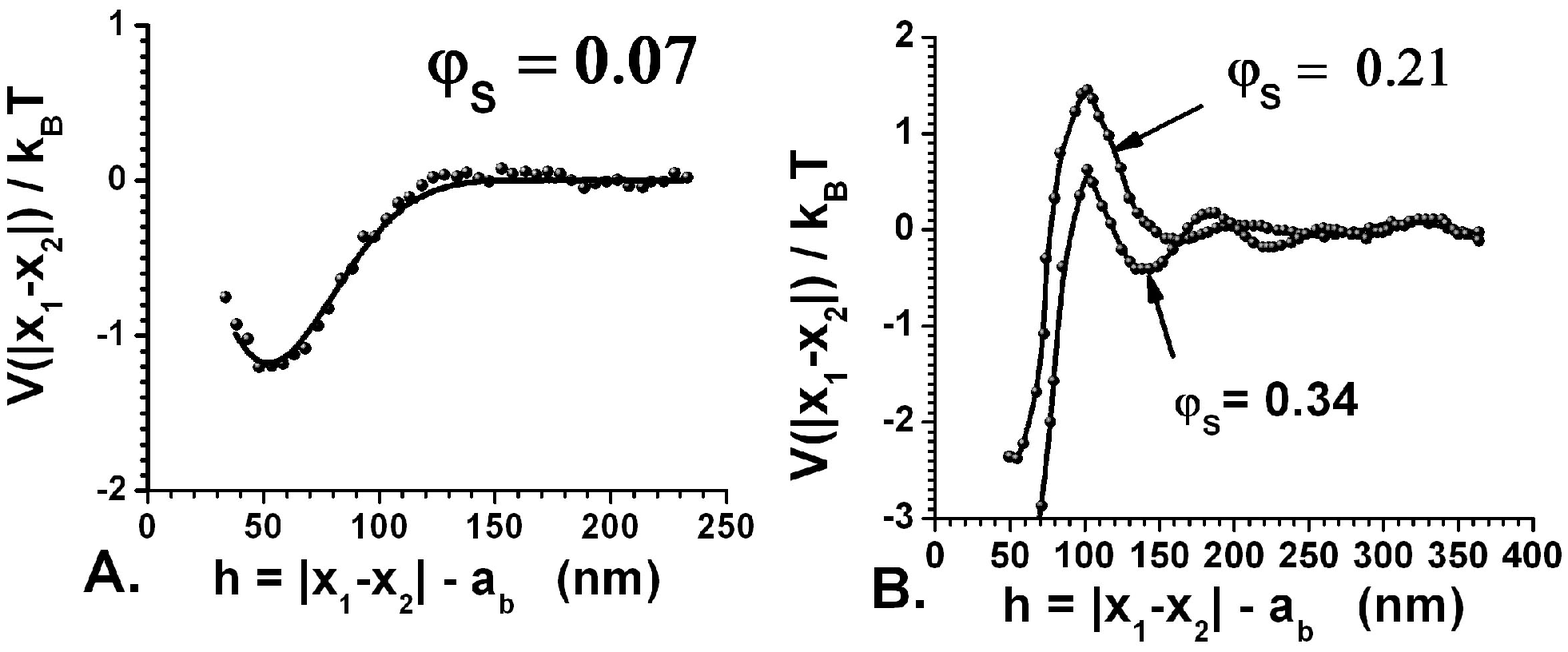}%
\caption[Laser tweezer measurements]{\newline\textbf{A.} The laser tweezer method was applied to very dilute solutions of $\symbol{126}1\mu m$ diameter large hard-spheres. Small volume fractions,
$\varphi_{s},$ of $\symbol{126}0.08\mu m$ diameter small hard-spheres were added. The pair interaction energy was fit well by the AO-Vrij theory. \newline\textbf{B. }The particle sizes are the same
as noted for Figure \textbf{A}. A range of volume fractions was reported by the authors as shown. Note the structure that is evident at $\varphi_{s}>0.07.$ Note that $r=|\vec{x}_{1}-\vec{x}_{2}|$
and that $h=r-a_{b}.$ These Figures were redrawn
from Figures 1 and 3 of reference~\cite{CR}.}%
\label{LaserTweezerLayout}%
\end{center}
\end{figure}

The fit to the AO-Vrij theory is convincing that the physics was captured quite well by a very simple PHS theory. However, Crocker \textit{et al.}~\cite{CR} report a significant deviation from the
AO-Vrij theory at larger volume fractions, $\varphi_{s}>0.07,$ as shown in Figure (\ref{LaserTweezerLayout} B.). Also see Zhou \cite{ZH} for a theoretical viewpoint. The authors show that there is a
substantial depletion repulsion at separations that correspond to approximately one small sphere diameter when $\varphi_{s}>1$.

The time-space correlation functions for the laser tweezer experiment as well as the TIRM experiment are in principle observable.

\paragraph{The molecular dynamic equivalent of the laser tweezer; steered MD.}

For example, see Park and Schulten~\cite{PAS} on calculating potentials of mean force from steered molecular dynamics simulation. The molecular dynamic equivalent of the laser tweezer experiment is
formulated as follows. The equations of motion of the small particles are for each particle, $s=1,2,\ldots,\#_{s}$,
\begin{equation}\label{NewtonSmall}
\begin{split}
m_{s}\frac{d\vec{v}_{s}}{dt}  & =\vec{F}_{s|s}+\vec{F}_{s|b}%
\\
\frac{d\vec{x}_{s}}{dt}  & =\vec{v}_{s}
\end{split}
\end{equation}
where $m_{s}$ is the mass of the small particle, $\vec{v}_{s}$ is its velocity, $\vec{F}_{s|s}=\sum_{\alpha=1}^{\#_{s}}\vec{F}_{s\alpha}$ is the sum of forces between the small particle $s$ and the
other small particles $\alpha,$ $\vec{F}_{s|b}=\sum_{n=1}^{2}\vec{F}_{sn}$ is the sum of forces between the small particle, $s$, and all of the $B$-particles. The small particles are not otherwise
constrained (other than for the usual cyclic boundary conditions).

Analogous to the laser tweezer experiment, a constraint to the motion of the large particles can be imposed through adding to Equation~(\ref{NewtonSmall}) and external force field \(\vec{E}\). The
equations of motion for the large particles (\(B-\)particles) are
\begin{equation}\label{NewtonBig}
\begin{split}
m_{B}\frac{d\vec{v}_{b}}{dt}  & =\vec{F}_{b|b}+\vec{F}_{b|s}+\vec
{E},\\
\frac{d\vec{x}_{b}}{dt}  & =\vec{v}_{b},
\end{split}
\end{equation}
where $m_{B}$ is the mass of the large particle and $\vec{v}_{b}$ its velocity. The forces are computed between the two large particles, $\vec{F}_{b|b}$ and
$B$-particles interacing with small, $\vec{F}_{b|s}=\sum_{\alpha=1}^{\#s}%
\vec{F}_{b\alpha}.$ In addition, the large particles are constrained to move on a line (without friction), but allowed by a harmonic force function to experience small fluctuations away from the
line:
\begin{equation}
\vec{E}(x,y,z)=-K_{B}(x\hat{e}_{x}+y\hat{e}_{y}).\label{HarmonicForce}%
\end{equation}
where $K_{B}$ is the spring constant and $\hat{e}_{x},\hat{e}_{y}$ are unit vectors perpendicular to the $z$ axis taken along the line of constraint.

The various forces $\vec{F}_{s\alpha},\vec{F}_{b\alpha},$ and $\vec{F}_{b|b}$ may be hard sphere, PHS or follow some simple but convenient force law such as the gradient of L-J 12-6 potentials. The
specifications of these force laws allow the accurate separation of the depletion effect from the total force of interaction.

A protocol for the simulation involves the initial positioning of the particles and assigning initial velocities from a Maxwell distribution to the small particles. Initially, the $B$-particles are
fixed along the $z$-axis separated by about $10$ small particle diameters and the small particles will redistribute following the dynamical Equations (\ref{NewtonSmall}) until equilibrium is reached
as suggested by Figure (\ref{FigPartAB}, A). Then the $B$-particles are released and given an initial speed drawn from a Maxwell distribution and allowed to move toward collision following the
dynamics of Equation (\ref{NewtonBig}).

It is clear that the characteristic time constant can be quite different for the two classes of particles and that makes the simulation of a binary solution with many $B$-particles awkward if not
impractical. Indeed if the set of $B$-particles are free to move but their number density is small, expect collisions to be so rare that the depletion effect will be rare and therefore the
statistics of the depletion potential will be very poorly represented. However, the harmonic force, Equation~(\ref{HarmonicForce}), steers two particles sufficiently that they will collide
relatively often and thereby enhance the determination of the distribution function from which the pair potential for the B-particles is estimated. Also, since the acceleration, velocity and
position of all particles are available it is possible to compute averages such as $<\vec{a}_{n}(t)\cdot\vec{a}_{\acute{n}}(t^{\prime})>$ and compare these "experimental" numbers with theory.

\section{Underlying Model and Stochastic Limit\label{sec-pk}}
In constructing our model we speak of particles instead of atoms or molecules. The term \textit{solute} refers to the large particles and the term \textit{solvent} to the medium of small particles.
In 1905 Einstein~\cite{EI} developed a model of Brownian motion to describe the motion of the large particles as a result of their interaction with the small particles. He assumed that the motions
of the large particles are statistically independent provided the system is very dilute; that is, the large particles are far apart from one another. (See Figure~(\ref{FigPartAB}, A).) In rigorous
mathematical treatments of several Brownian motions (Wiener processes), the assumption of independence, regardless of separation distance, has become widely accepted. However, we saw in
Section~\ref{sec-jam} that when the large particles are close together the depletion effect induces a force that attracts the particles to one another. In particular, the motions of the large
particles become statistically correlated when they are close, as measured by the correlation length~\(\sqrt{\varepsilon}\).

We therefore seek a model for correlated Brownian motions of the large particles that satisfies the following four
\newpage\textbf{Desiderata:}
\begin{enumerate}
    \item\label{desid1} The \emph{marginal} motion of any single  particle is Brownian (Wiener process).\footnote{Since we are dealing here with the joint motion of several large particles and each motion is, in a sense,
    an infinite dimensional random variable, the term \textit{marginal}  must properly be defined in this context. We do this for pair-motions (two large particles) in Appendix~\ref{app-marg}.}
    \item\label{desid2} If the particles are widely separated (dilute system), they perform approximately independent Brownian motions.
    \item\label{desid3} If the separation between particles is small, their motions are correlated. Moreover, the correlation is such that if the separation is sufficiently small, as measured by the
    the  correlation length, they tend statistically to approach one another further.
    \item\label{desid3b}As the correlation length tends to zero, the particles become \(\delta-\)~correlated in space and time.\footnote{That is, in the limit the particles become
    uncorrelated unless they collide.}
\end{enumerate}

Kotelenez~\cite{KO3} obtained a class of correlated Brownian motions as a scaling limit for the positions of several large particles immersed in a medium of infinitely many smaller
particles.\footnote{Kotelenez~\cite{KO5,KO6}(1995) introduced correlated Brownian motions as a driving term in stochastic ordinary differential equations (SODEs) and stochastic partial differential
equations (SPDEs). (See, for example, Equation~(\ref{eq3.2}.))} Here is a brief sketch of this work.

We need to consider two levels of description of the particle system. On the \emph{microscopic} level, we suppose Newtonian mechanics governs the equations of motion of the individual atoms
or molecules. 
These equations are cast in the form of a system of deterministic coupled nonlinear equations. The next level is called \emph{mesoscopic}. On this level the motion of the large particles is
stochastic; the randomness of their motions is determined by the surrounding medium.
Here, spatially extended particles are replaced by point particles; large and small particles are distinguished by their large and small masses.\footnote{For a rare gas a mean-field force can be a
result of coarse graining in space and time, where on a finer scale the interaction is governed by collisions. See also Comment~\ref{com3.2}.} Furthermore, the interaction between small particles is
assumed to be negligible and interactions between large particles can (temporarily) be neglected.\footnote{As the interaction between large particles occurs on a much slower time scale than the
interaction between large and small particles, it can be included after the scaling limit employing fractional steps (See Goncharuk and Kotelenez~\cite{GN}).}

We suppose that the interaction between large and small particles is governed by a scalar-valued potential of the form \(\varphi(|\bs{r}-\bs{q}|^2)\), where \bs{r} denotes the position of a large
particle and \bs{q} denotes the position of a small particle. Thus we assume that the potential does not depend on the locations of the two particles but only on their vector difference, \(\bs{r}-
\bs{q}\); it is a \textit{homogeneous} or \textit{shift invariant} function of \bs{r} and~\bs{q}. In fact, we assume that the potential depends only on the \emph{magnitude of the difference} or the
\emph{separation}, \(|\bs{r}-\bs{q}|\); it is an \textit{isotropic} function of the difference.\footnote{ This assumption is entirely consistent with the Principle of Material Frame Indifference,
discussed in Subsection~\ref{subsec-frame}. In fact, for a homogeneous scalar function of two vector arguments, \hyperlink{isotropic}{isotropy} is equivalent to isotropy in the difference. Here
\(|\bs{r}-\bs{q}|\) denotes the Euclidean distance in the state space \(\mathbb{R}^d\) between the two particles. } This case will be studied in detail in Section~\ref{sec-sep}. Recall that Brownian
motion is interpreted to be the result of collisions between  many, fast moving, small particles and a few, slowly moving, large particles. Often it is assumed that these collisions are elastic.

Specifically, we suppose that the force \(\hat{\bs{g}}_{\varepsilon,\mu}(\bs{r}-\bs{q})\) on a large particle at~\bs{r} due to a small particle at~\bs{q} is derivable from the potential
\(\hat{\varphi}_{\varepsilon,\mu}\):
\begin{equation}\label{pot}
\hat{\bs{g}}_{\varepsilon,\mu}(\bs{r}-\bs{q}):=-\bs{\nabla}\hat{\varphi}_{\varepsilon,\mu}(|\bs{r}-\bs{q}|^2)=-2(\bs{r}-\bs{q})\hat{\varphi}_{\varepsilon,\mu}^{\prime}(|\bs{r}-\bs{q}|^2).
\end{equation}
The potential function \(\hat{\varphi}_{\varepsilon,\mu}\) and, hence, the force \(\hat{\bs{g}}_{\varepsilon,\mu}\) depend on two parameters: the correlation length \(\sqrt{\varepsilon}\) and a
time-scale parameter \(\mu\).\footnote{The parameter \(\mu\) has the units of reciprocal time (\([=]\frac{1}{T}\)),  $\bs{\nabla}$ denotes the spatial gradient in \(\mathbb{R}^d\), and the prime
(\(^{\prime}\)) denotes differentiation with respect to the scalar argument.}

In the classical model of a dilute system of large (Brownian) particles, mentioned above, the fluctuation forces \(\bs{f}(\bs{r}^1,t_1)\) and \(\bs{f}(\bs{r}^2,t_2)\) on the positions of two large
particles, located at \(\bs{r}^1\) and \(\bs{r}^2\) at times \(t_1\) and \(t_2\), are assumed to be $\delta-$correlated in space and time: \(\langle\bs{f}(\bs{r}^1,t_1)\bs{f}^T(\bs{r}^2,t_2)\rangle
\varpropto \bs{\delta} (\bs{r}^1-\bs{r}^2)\otimes\delta(|t_1-t_2|)\), where $\bs{\delta} \text{ and }\delta$ denote Dirac's $\delta-$function in \(\mathbb{R}^d\) and \(\mathbb{R}\).\footnote{ In
terms of stochastic analysis this means that the ``noise'' (replacing the solute of small particles) is \emph{white} in space and time. In particular, this implies the independent increments in
Brownian motions.} In contrast, our model can capture the qualitative behavior of both dilute and non-dilute systems. 
Since we consider $\delta-$correlated
noise to be an approximation to the more realistic spatially correlated noise, we require in Desideratum~\ref{desid3b} that the fluctuation forces associated with our model be $\delta-$correlated in
the limit as the correlation length $\sqrt{\varepsilon}\downarrow 0$.

Suppose there are $N$ large particles and infinitely many small particles distributed in the Euclidean state space \({\mathbb{R}^d}\).\footnote{Infinitely many small particles are needed to generate
independent increments in the limiting Brownian motion. See our Comment~\ref{com3.1} and Kotelenez~\cite{KO3}.} At time \(t\), the position of the $\alpha$th large particle is denoted by
$\bs{r}^{\alpha}(t)$ and its velocity is denoted by $\bs{v}^{\alpha}(t)$. The  position and velocity of the $\lambda$th small particle are denoted by $\bs{q}^{\lambda}(t)$ and $\bs{w}^{\lambda}(t)$.
The mass of a large particle is $m_l$,  the mass of a small particle is $m_s$, and we suppose that \(m_s\ll m_l\). The empirical mass distributions of large and small particles are (formally) given
by
\begin{equation}\label{empdist}
\mathcal{X}(d\bs{r},t)=m_l\sum_{{\alpha}=1}^N\delta_{\bs{r}^{\alpha}(t)}(d\bs{r})\qquad\text{and}\qquad \mathcal{Y}(d\bs{q},t)=m_s\sum_{\lambda\in\mathbb{N}}\delta_{\bs{q}^{\lambda}(t)}(d\bs{q}).
\end{equation}
We assume, at the microscopic level, that the interaction between small and large particles can be described by the following infinite system of coupled nonlinear dynamic equations:\footnote{Note
that Equation~(\ref{pot}) implies \(\hat{\bs{g}}_{\varepsilon,\mu}(\bs{r}-\bs{q})=-\hat{\bs{g}}_{\varepsilon,\mu}(\bs{q}-\bs{r})\); that is, the force on \bs{r} due to \bs{q} is necessarily equal
and opposite to the force on \bs{q} due to \bs{r}. Strictly speaking \(\hat{\bs{g}}_{\varepsilon,\mu}\) plays the r\^{o}le of a  \emph{force density} or \emph{force per small particle} in
Equations~(\ref{eq3.1a}) and therefore has units of force per unit volume (\( [=]\frac{M\cdot L}{T^2}\frac{1}{L^d}\)).}
\begin{equation}\label{eq3.1a}
    \begin{split}
    \frac{d}{dt}\bs{r}^{\alpha}(t)&=\bs{v}^{\alpha}(t),\\
    \frac{d}{dt}\bs{v}^{\alpha}(t)&=-\mu\bs{v}^{\alpha}(t)+\frac{1}{m_l m_s}\int_{\mathbb{R}^d}\hat{\bs{g}}_{\varepsilon,\mu}(\bs{r}^{\alpha}(t)-\bs{q})\mathcal{Y}(d\bs{q},t),\\
    \frac{d}{dt}\bs{q}^{\lambda}(t)&=\bs{w}^{\lambda}(t),\\
    \frac{d}{dt}\bs{w}^{\lambda}(t)&=-\frac{1}{m_l m_s}\int_{\mathbb{R}^d}\hat{\bs{g}}_{\varepsilon,\mu}(\bs{r}-\bs{q}^{\lambda}(t))\mathcal{X}(d\bs{r},t),
    \end{split}
\end{equation}
where \(\alpha=1,2,\cdots,N\) and \(\lambda\in\mathbb{N}\). The positive time-scale parameter, \(\mu\), is also introduced here as a Stokes-friction parameter associated with the large particles. To
specify an evolutionary system we append the random initial conditions:
\begin{equation}\label{eq3.1b}
\bs{r}^{\alpha}(0)=\bs{r}^{\alpha}_0,\quad\bs{v}^{\alpha}(0)=\bs{v}^{\alpha}_0,\quad\bs{q}^{\lambda}(0)=\bs{q}^{\lambda}_0,\quad\bs{w}^{\lambda}(0)=\bs{w}^{\lambda}_0.
\end{equation}

The resulting mesoscopic model of correlated Brownian motions can be defined as follows: Let $w(d\bs{q},ds)$ denote standard Gaussian white noise on $\mathbb{R}^d\times\mathbb{R}^+$, which is a
space-time generalization of the time increments of a standard scalar-valued Brownian motion. (See Appendix~\ref{app-white} for details of this generalization.) The white noise and the random
initial data are defined on the same probability space and are assumed independent.

In the stochastic limit, the positions of the large particles are shown to be the solutions of the \(N\) kinematic stochastic integral equations
\begin{equation}\label{eq3.2}
\bs{r}^{\alpha}(t) = \bs{r}^{\alpha}_0 + \int_0^t\int_{\mathbb{R}^d} \bs{g}_{\varepsilon}(\bs{r}^{\alpha}(s)-\bs{q})w(d\bs{q},ds),\quad \alpha = 1,...,N.
\end{equation}
The kernel,~\(\bs{g}_{\varepsilon}\), in Equations~(\ref{eq3.2}) is induced by the force density field \(\hat{\bs{g}}_{\varepsilon,\mu}\) of Equations~(\ref{eq3.1a}) through the transition from a
second order system in time (dynamic description) to a first order system in time (kinematic description).
This transition requires an assumption that the force density field \( \hat{\bs{g}}_{\varepsilon,\mu} \) has specific asymptotic behavior as the time-scale/friction parameter \(\mu\) gets large;
specifically, there is a function \(\bs{r}\mapsto\bs{g}_{\varepsilon}(\bs{r})\) such that\footnote{Observe that large values of \(\mu\) correspond to short time-scales.}
\begin{equation}\label{timescale}
\hat{\bs{g}}_{\varepsilon,\mu}(\bs{r})\approx\mu\bs{g}_{\varepsilon}(\bs{r})\qquad\text{as }\mu\rightarrow\infty.
\end{equation}
In the sequel, we refer to the function \(\bs{g}_{\varepsilon}\) as the \hyperlink{kernel}{forcing kernel}.\footnote{If we suppose that the space-time white noise, \(w(d\bs{q},dt)\), has the units
of volume times time (\([=]L^d T\)), the \hyperlink{kernel}{forcing kernel}, \(\bs{g}_{\varepsilon}\), cannot have the units of  a force density.}

\begin{com}\label{com3.2}
Suppose that small particles move with different velocities. If most of the small particles moving in the direction of a large particle can avoid collisions with other small particles (as in a rare
gas or in the PHS model of Section~\ref{sec-jam}), fast small particles coming from ``far away'' can collide with a given large particle at approximately the same time as slow small particles that
were close to the large particle before the collision. If, in repeated microscopic time steps, collisions of a given small particle with the \emph{same} large particle are negligible, then, in a
mesoscopic
time unit, 
the collision dynamics can be replaced by long-range mean field dynamics. Dealing with a wide range of velocities and working with discrete time steps, a long range force is generated.

If, for instance, we assume that the empirical velocity distribution of the small particles is approximately Maxwellian, the aforementioned force density, \(\hat{\bs{g}}_{\varepsilon,\mu}\), can be
given by an expression in the form of Equation~(\ref{timescale}): (See Kotelenez~\cite{KO3}, Equation~(1.2))\footnote{Clearly, the right-hand side of Equation~(\ref{eq3.3}) is of the form of
Equation~(\ref{pot}); that is, it is the negative gradient of a potential. }

\begin{equation}\label{eq3.3}
\hat{\bs{g}}_{\varepsilon,\mu}(\bs{r}-\bs{q})\approx\mu \kappa_{\varepsilon,d}(\bs{r}-\bs{q})\, e^{-\frac{|\bs{r}-\bs{q} |^2}{2\varepsilon}},
\end{equation}
where \(\mu\) is the friction/time-scale parameter associated with the large particles and \(\kappa_{\varepsilon,d}\) is a normalizing constant chosen so that the particles become
\(\delta-\)correlated in a standard way as \(\varepsilon\downarrow 0\).

In other words, if the above assumptions hold in a first approximation, the interactions between large and small particles are governed by a velocity field for which the variance of the distribution
is the correlation length. Obviously, this example can be generalized to an arbitrary velocity field of the small particles. (A more realistic model might involve some, possibly nonlinear,
transformation of the velocity field, taking into account collisions between small particles, etc.) For the purpose of our work here, it suffices to consider a general \hyperlink{kernel}{forcing
kernel}, as in Equations~(\ref{eq3.2}), and show that for certain kernels the right hand side behaves according to the requirements stated at the beginning of this section.
\end{com}

\begin{maxwell}\label{com-delcor}
We must choose the constant \(\kappa_{\varepsilon,d}\) of Equation~(\ref{eq3.3}) in order to satisfy  Desideratum~\ref{desid3b}. Later we will see that for the \hyperlink{maxker}{Maxwell kernel},
\(\kappa_{\varepsilon,d}\) should be chosen so that if \(\bs{g}_{\varepsilon}(\bs{r}):=\kappa_{\varepsilon,d}\bs{r}e^{-\frac{|\bs{r}|^2}{2\varepsilon}}\) in the kinetic stochastic
Equation~(\ref{eq3.2}), then \(\bs{r}\mapsto\frac{1}{d}| \bs{g}_{\varepsilon}(\bs{r}) |^2\) must approximate the \(\delta-\)function as \(\varepsilon\downarrow 0\).  This requirement implies
\(\kappa^2_{\varepsilon,d}\frac{\varepsilon }{2} (\pi\varepsilon)^{d/2} =1\). This normalization implies that \(\bs{r}\mapsto\frac{1}{d}| \bs{g}_{\varepsilon}(\bs{r}) |^2\) is a probability density
on \(\mathbb{R}^d\). We record for later use, that the variance of the associated probability distribution is \(\frac{(d+2)\varepsilon}{2}\).

Of course, Desideratum~\ref{desid3b} requires that for \emph{any} forcing function,  \(\bs{r}\mapsto\frac{1}{d}| \bs{g}_{\varepsilon}(\bs{r}) |^2\) must approximate the \(\delta-\)function as
\(\varepsilon\downarrow 0\). Later we will see that this is equivalent to requiring that, in the limit as \(\varepsilon\downarrow 0\), the infinitessimal generator associated with the diffusion is
one-half the Laplacian. (See Comment~\ref{com-vkflux}.) This, in turn, is equivalent to requiring that, in the limit as \(\varepsilon\downarrow 0\), each particle experiences a \emph{standard}
Brownian motion and that the motions are independent.
\end{maxwell}

\begin{com}\label{com-forcing}
The transition from the microscopic (dynamic) to the mesoscopic (kinematic) description is quite technical. See Kotelenez~\cite{KO3} for the details. Heuristically, however, an examination of the
resulting kinematic Equations~(\ref{eq3.2}) reveals that one consequence of the procedure is to render negligible the inertial effect on \emph{each} large particle due to its interactions with
\emph{all} the small particles. In the sequel our model will imply that the effect on each large particle by all the small particles will be, at most, fluctuations in the position of the large
particle. An analysis of the mutual spatial correlations of these motions is the object of our work here. To describe these correlations and compare them with the depletion phenomenon we have here
excluded the interactions between the large Brownian particles.\footnote{There is a large literature on \emph{interacting} Brownian motions. In the context of coagulating Brownian particles  refer
to a recent paper by Hammond and Rezakhanlou~\cite{HA} and the references therein. }
\end{com}

It is important to note that the system of Equations~(\ref{eq3.2}) is coupled only through the Gaussian space-time white noise $w(d\bs{q},ds)$.\footnote{The space-time white noise \(w(d\bs{q},ds)\)
is obtained in a scaling limit from the number of small particles, in the small volume \(d\bs{q}\) and the small time interval \(ds\), with a given large particle. See Kotelenez~\cite{KO3} and the
brief discussion here of the scaling limit.} Regularity assumptions are assumed sufficient to guarantee that the integrals (It\^{o} integrals) define continuous square integrable martingales and
that, for the large class of kernels $\bs{g}_{\varepsilon}$ we consider, each of the Equations~(\ref{eq3.2}) has a unique solution, which is itself a Brownian motion. In Section~\ref{sec-partsys} we
give a detailed analysis of the correlations between these motions when \(N=2\); that is, when there are two large particles. It will be shown that the \emph{joint} motion of the pair is not
Brownian.

The transition from the microscopic description of Equations~(\ref{eq3.1a}, \ref{eq3.1b}) to the  mesoscopic (stochastic) description of Equation~(\ref{eq3.2})  is accomplished through the following
steps:

\begin{itemize}

\item  Form small clusters (ensembles) of particles, if their initial positions and velocities are similar (coarse graining in
space).

\item Replace the time derivative in Equations~(\ref{eq3.1a}) by an Euler scheme (coarse graining in time).

\item Randomize the initial distribution of clusters, where the probability distribution is determined by the
relative sizes of the clusters, assuming statistical independence of the initial distributions of different clusters.

\item Assume that the initial average velocity, $\langle\bs{w}_0\rangle$, of the small particles and the time-scale/friction
coefficient $\mu$ associated with the large particles both tend to infinity in such a way that $\sqrt{\varepsilon}\mu \ll \langle\bs{w}_0\rangle$.

\item Allow the small particles to escape to infinity after interacting with the large particles for a macroscopically small
time.\footnote{This hypothesis seems to be acceptable if, for spatially extended particles, the interparticle distance is considerably greater than the diameter of a typical particle. The assumption
holds for a gas (See Lifshits and Pitayevskii~\cite{LI}, Ch.1, \S 3 ), but not for a liquid, like water. For a liquid, we refer to the PHS model, introduced in Section~\ref{sec-jam}.}
\end{itemize}

Carrying out these steps, Kotelenez~\cite{KO3}  shows that the positions of the large particles in a sequence of coarse grained versions of Equations~(\ref{eq3.1a},~\ref{eq3.1b})  tend to the
solutions of Equations~(\ref{eq3.2}) weakly in an appropriate space of functions with values in $\mathbb{R}^{d N}$, $d \geq 2$.\footnote{More precisely, in the Skorohod space of \textsc{cadlag}
functions with values in $\mathbb{R}^{d N}$, $d \geq 2$.} We begin the substance of our work here with the motions governed by the system of Equations~(\ref{eq3.2}).

\begin{com}\label{com3.1}
The escape to infinity after a short period of interaction with the large particles is necessary to generate independent increments in the limit. This can be seen as follows: The small time induces
a partition of the time axis into small time intervals. In each of the small time intervals the large particles are being displaced by the interaction with clusters of small particles. Note that the
vast majority of small particles had previously not interacted with the large particles and that they disappear toward infinity after that time step. Since clusters have started independently, this
implies almost independence of the displacements of the large particles in different time intervals. In a scaling limit, as the initial velocities of the small particles and the friction coefficient
for the large particles tend to infinity, the motions of the large particles have independent increments. The requirements of an infinite number of small particles and their eventual escape to
infinity are both needed to obtain independent increments in time. A similar result is obtained if no friction term is introduced in the dynamic equations for the large particles. In this case,
however, the limit is an Ornstein-Uhlenbeck model (described by Langevin equations), where the velocities, rather than the positions, perform correlated Brownian motions.
\end{com}

\begin{com}\label{gurarie}
Klyatskin and Gurarie~\cite{GU} show that a kinetic model of the form \(\frac{d}{dt}\bs{r}(t)=\bs{f}(\bs{r}(t),t)\), where \(\bs{f}\) is a field consisting of a deterministic term plus a random
term, can exhibit clustering in simulations when the vector field \bs{f} is compressive; that is, derivable from a potential. The vector field in our kinetic model is compressive.
\end{com}

\section{Preliminaries\label{sec-prelim}}
\subsection{Notation}\label{subsec-not}
In this section we set out the basic notational conventions used throughout the rest of our paper.

The real numbers are denoted by \(\mathbb{R}\) and the non-negative real numbers by \(\mathbb{R}^+\). Scalars are always lightface symbols. In particular, \(t\text{ and }s\) always denote times.
Vectors and square matrices are denoted by boldface symbols. Generally, we use Latin minuscules for vectors and Latin majuscules for square matrices in the \(d\)-dimensional real Euclidean vector
space \(\mathbb{R}^d,\,d\geq1\). Unless otherwise stated, all matrices are \emph{vectors} or \emph{square} matrices; that is, \((d\times 1)\)-matrices,  \((1\times d)\)-matrices, or  \((d\times
d)\)-matrices.  If \(\bs{M}\) denotes a matrix and \bs{x} denotes a vector, \(\bs{M}^T\) and \(\bs{x}^T\) denote their transposes. We say \textit{vector} when we mean \textit{column vector;} that
is, \(\bs{x}\in\mathbb{R}^{d\times 1}\) and, hence, \(\bs{x}^T\in\mathbb{R}^{1\times d}\) is a \textit{row vector.} For a matrix, we say \textit{transformation} when we mean to emphasize its
r\^{o}le as a linear transformation. For two vectors, \(\bs{x},\bs{y}\), their inner (scalar) product is denoted by \(\bs{x}\bullet\bs{y}\) and their outer (tensor) product is denoted by
\(\bs{x}\otimes\bs{y}\). We have \(\bs{x}^T\bs{y}\equiv\bs{x}\bullet\bs{y}\), a scalar, and \(\bs{x}\bs{y}^T\equiv\bs{x}\otimes\bs{y}\), a square matrix. The norm, or length, of a vector \bs{x} is
\(|\bs{x}|:=\sqrt{\bs{x}\bullet\bs{x}}\).  (We also write \(|\gamma|\) for the modulus of a scalar.) The inner (scalar) product of two matrices, \bs{M} and \bs{N}, is the scalar given by
\(\bs{M}\bullet\bs{N}:=trace(\bs{M}\bs{N}^T)\).\footnote{Note that \(trace(\bs{x}\bs{y}^T)=\bs{x}^T\bs{y}\) or, equivalently, \(trace(\bs{x}\otimes\bs{y})=\bs{x}\bullet\bs{y}\).} The norm of a
matrix induced by this inner product is the Hilbert-Schmidt norm \(|\bs{M}|:=\sqrt{\bs{M}\bullet\bs{M}}\). \footnote{It is topologically equivalent to the operator norm, \(\| \bs{M} \|\), of the
matrix as a linear operator on \(\mathbb{R}^d\).}

For \(\bs{y}\not=\bs{0}\), we write \(\bs{u}(\bs{y})\) for the \textit{unit outward radial vector (at~\bs{y}):}
\begin{equation}\label{normvec}
\bs{u}(\bs{y}):=\frac{\bs{y}}{|\bs{y}|}.
\end{equation}
Any identity matrix is denoted by \bs{1} and any matrix or vector whose entries are all zeros is denoted by \bs{0}. For \(\bs{y}\not =\bs{0}\) we write \(\bs{P}(\bs{y})\) and
\(\bs{P}^{\perp}(\bs{y})\)  for the complementary orthogonal projections onto the \bs{y} direction, the \(1-\)dimensional subspace~\(\{\bs{y}\}\), and the hyperplane orthogonal to \bs{y}, the
\((d-1)-\)dimensional subspace \(\{\bs{y}\}^{\perp}\), given by\footnote{If necessary, we extend these definitions to include \(\bs{y}=\bs{0}\) by requiring \(\bs{P}(\bs{0}):=\bs{0}\), the
projection on \(\{\bs{0}\}\) and, hence, \( \bs{P}^{\perp}(\bs{0}):=\bs{1} \), the projection on \(\{\bs{0}\}^{\perp}\).}

\begin{equation}\label{proj}
\bs{P}(\bs{y}):=\frac{\bs{y}\bs{y}^T}{|\bs{y}|^2}\qquad\text{and}\qquad\bs{P}^{\perp}(\bs{y}):=\bs{1}-\bs{P}(\bs{y})=\bs{1}-\frac{\bs{y}\bs{y}^T}{|\bs{y}|^2}.
\end{equation}
Note that, for \(\bs{y}\not =\bs{0}\),
\begin{equation}\label{proj1}
\bs{P}(\bs{y})\bullet\bs{P}(\bs{y})=1,\quad\bs{P}^{\perp}(\bs{y})\bullet\bs{P}^{\perp}(\bs{y})=d-1,\quad\text{and}\quad\bs{P}(\bs{y})\bullet\bs{P}^{\perp}(\bs{y})=0.
\end{equation}

The position vector of a large particle is denoted by~\(\bs{r}\). In case there are more than one, the position vector of the \(\alpha\)th large particle is denoted by~\(\bs{r}^{\alpha}\). In
general, large particles are indexed by Greek minuscules, so \(\alpha,\beta=1,2,\ldots ,N.\) In  \(\mathbb{R}^d\), the standard coordinates of a vector \bs{r} are indexed by Latin minuscules; thus,
\(r^{\alpha}_j\) denotes the \(j\)th standard coordinate of \(\bs{r}^{\alpha}\).\footnote{Coordinates are scalars, so the corresponding symbol is lightface. \textit{Standard} means with respect to
some prescribed orthonormal basis.}

The underlying probability space for all random variables is \((\Omega,\mathcal{F},\mathcal{P})\). If \(Z\) denotes a \emph{random variable,} \(Z=Z(\omega),\,\omega\in\Omega\), we suppress the
dependence on~\(\omega\), unless confusion is likely or emphasis is needed. We write \(E\big[Z\big]\) to denote its expectation: \( E\big[Z\big]:=\int_{\Omega}Z(\omega)\mathcal{P}(d\omega) \).

The norm of a function, \(f\), (scalar- or vector-valued) in an appropriate \(L^2\) setting is denoted by \(\|f\|\). If \(\zeta\mapsto f(\zeta)\) is a  function of a scalar argument (scalar-,
vector-, or tensor-valued) differentiable at \(\zeta=\hat{\zeta}\), we write \(f^{\prime}(\hat{\zeta})\) for \(\frac{d}{d\zeta} f(\zeta)\big|_{\zeta=\hat{\zeta}}\).

We adopt the following special convention:  If \((\bs{x}^{1},\bs{x}^{2})\) is an (ordered) pair of vectors in \(\mathbb{R}^d\equiv\mathbb{R}^{d\times 1}\)
then \(\hat{\bs{x}}\)  is the vector in \(\mathbb{R}^{2d}\equiv\mathbb{R}^{2d\times 1}\) given by \(\hat{\bs{x}}=\begin{pmatrix}\bs{x}^{1}\\
\bs{x}^{2}\end{pmatrix}\).

\subsection{Material Frame-Indifference}\label{subsec-frame}
 The material descriptions that underlie our model are assumed to be invariant under changes in external observer. This fundamental requirement, known since the time of Stokes,
is now called the \hypertarget{frame-inv}{\textit{Principle of Material Frame-Indifference.}} As formulated by Noll~\cite{NO}, it asserts:
\begin{quote}
The constitutive laws governing the internal interactions between the parts of a system should not depend on whatever external frame of reference is used to describe them.
\end{quote}
In our context, this requirement imposes invariance relations on certain materially significant scalar-, vector-, and matrix-valued functions of position. In the context of constitutive laws,
functions restricted by this principle are often said to be \textit{frame-indifferent;} mathematically, they are \hypertarget{isotropic}{\textit{isotropic.}}\footnote{There is an extensive
literature on such functions in the context of mechanics. See, for example, the comprehensive article in the \textit{Encyclopedia of Physics - The Non-Linear Field Theories of Mechanics} by
Truesdell and Noll~\cite{TR}. In particular, they state and outline a proof of the fundamental Representation Theorem of Cauchy~\cite{CA} on which Lemma~\ref{lem0} is based. (For details, see
Truesdell and Noll~\cite{TR}, \S B,II,11.)}

\begin{define}[Isotropic Functions]\label{def1}
Let \(\bs{r}\rightarrow \psi(\bs{r}),\,\bs{f}(\bs{r}),\,\text{and }\bs{F}(\bs{r})\) denote scalar-, vector-, and matrix-valued functions of position  on \(\mathbb{R}^d\). They are \emph{isotropic}
functions whenever they satisfy:
\begin{align}
\psi(\bs{Q}\bs{r})      &=\psi(\bs{r}),\label{frinv1}\\
\bs{f}(\bs{Q}\bs{r})    &=\bs{Q}\bs{f}(\bs{r}),\label{frinv2}\\ \bs{F}(\bs{Q}\bs{r})    &=\bs{Q}\bs{F}(\bs{r})\bs{Q}^T,\label{frinv3}
\end{align}
for all orthogonal transformations (\((d\times d)\)-orthogonal matrices) \bs{Q} on \(\mathbb{R}^d\).\footnote{The collection of \((d\times d)\)-orthogonal matrices, \bs{Q}, is the \textit{orthogonal
group}, \(\mathcal{O}^d\). The \textit{proper orthogonal group}, \(\mathcal{O}^d_+ \), is the sub-group that preserves orientation; these are the \textit{rotations.} In particular, if
\(\bs{Q}\in\mathcal{O}^d\) then \(\text{det}\bs{Q}=\pm1\); while if \(\bs{Q}\in\mathcal{O}^d_+\) then \(\text{det}\bs{Q}=1\).}
\end{define}
Isotropic functions must have very special forms.

\begin{lem}[Representations for \hyperlink{isotropic}{Isotropic} Functions]\label{lem0}
\begin{enumerate}
\item\label{lem01} A scalar-valued function of position, \(\psi:\mathbb{R}^d\rightarrow\mathbb{R}\), is \hyperlink{isotropic}{isotropic} if and only if there is a scalar-valued function \(\xi\mapsto\alpha(\xi):\mathbb{R}^+\rightarrow\mathbb{R}\)
such that
\begin{equation}\label{frinv4}
\psi(\bs{r})=\alpha(|\bs{r}|^2);
\end{equation}
\item\label{lem02} A vector-valued function of position, \(\bs{f}:\mathbb{R}^d\rightarrow\mathbb{R}^d\), is \hyperlink{isotropic}{isotropic} if and only if there is a scalar-valued function
\(\xi\mapsto\beta(\xi):\mathbb{R}^+\rightarrow\mathbb{R}\) such that\footnote{\hyperlink{isotropic}{Isotropy} here requires that \(\bs{f}(\bs{0})=\bs{0}\).}
\begin{equation}\label{frinv5}
\bs{f}(\bs{r})=\beta(|\bs{r}|^2)\bs{r};\text{ and}
\end{equation}
\item\label{lem03} A matrix-valued function of position, \(\bs{F}:\mathbb{R}^d\rightarrow\mathbb{R}^{d\times d}\), is \hyperlink{isotropic}{isotropic} if and only if there are two
scalar-valued functions \(\xi\mapsto\gamma(\xi),\eta(\xi):\mathbb{R}^+\rightarrow\mathbb{R}\) such that\footnote{\hyperlink{isotropic}{Isotropy} here requires that \(\bs{F}(\bs{0})=\kappa\bs{1}\),
for some constant \(\kappa\). This is a special case of the well-known result that a matrix commutes with all orthogonal matrices if, and only if, it is a multiple of the identity. See
Theorem~\ref{full} and the discussion in Appendix~\ref{app-a}. This is implicit in Equations~(\ref{frinv6}) and~(\ref{frinv7}). Note also that the conclusion of Part~\ref{lem03} implies that a
\hyperlink{isotropic}{isotropic} matrix function is symmetric.}
\begin{equation}\label{frinv6}
\bs{F}(\bs{r})=\gamma(|\bs{r}|^2)\bs{1}+\eta(|\bs{r}|^2)\bs{r}\bs{r}^T.
\end{equation}
\end{enumerate}
\end{lem}
A proof of this Lemma is provided in Appendix~\ref{app-a}.\footnote{If the underlying vector space has dimension \(d=1\), the invariance of statements in Equations~\((\ref{frinv1}-\ref{frinv3})\)
are not all the same; the statements of Equations~(\ref{frinv1}) and~(\ref{frinv3}) are the same and imply that the function in question is even while that of Equation~(\ref{frinv2}) implies that
the function in question is odd.} We will also use the following alternate version  Lemma~(\ref{lem0}), Part~(\ref{lem03}), employing the spectral decomposition and the projections of
Equation~(\ref{proj}):
\begin{quote}\begin{description}\label{lem03alt}
\item[\indent\normalfont{\textit{3. (Alternate)}}]
\textit{A matrix-valued function of position, \(\bs{F}:\mathbb{R}^d\rightarrow\mathbb{R}^{d\times d}\), is \hyperlink{isotropic}{isotropic} if and only if there are two scalar-valued functions
\(\xi\mapsto\lambda(\xi),\lambda_{\perp}(\xi):\mathbb{R}^+\rightarrow\mathbb{R}\) such that}\footnote{Observe that the complementary orthogonal projections \(\bs{P}(\bs{r})\) and
\(\bs{P}^{\perp}(\bs{r})\) are themselves \hyperlink{isotropic}{isotropic}.}
\begin{equation}\label{frinv7}
\bs{F}(\bs{r})=\lambda(|\bs{r}|^2)\bs{P}(\bs{r})+\lambda_{\perp}(|\bs{r}|^2)\bs{P}^{\perp}(\bs{r}).
\end{equation}
\end{description}
\end{quote}
In this version \(\lambda(|\bs{r}|^2)\) and \(\lambda_{\perp}(|\bs{r}|^2)\) are the eigenvalues of \(\bs{F}(\bs{r})\) with corresponding eigenspaces \(\{\bs{r}\}\) and
\(\{\bs{r}\}^{\perp}\).\footnote{Here, \(\{\bs{r}\}\) denotes the subspace spanned by \bs{r} and \(\{\bs{r}\}^{\perp}\) denotes its orthogonal compliment.} If these eigenvalues are distinct, then
\(\lambda(|\bs{r}|^2)\) is a simple eigenvalue; if not, \(\bs{F}(\bs{r})=\lambda(|\bs{r}|^2)\bs{1}=\lambda_{\perp}(|\bs{r}|^2)\bs{1}\). The latter occurs if and only if in addition to being
isotropic, \bs{F} satisfies: \(\bs{F}(\bs{r})=\bs{F}(\bs{Q}\bs{r})\), for all orthogonal \bs{Q}.\footnote{In this case \(\bs{F}(\bs{r})\) commutes with all orthogonal transformations, so the result
follows independently from Theorem~\ref{full} in Appendix~\ref{app-a}.}



\section{One- and Two-Particle Systems\label{sec-partsys}} Henceforth we will track a single particle, with position vector \(\bs{r}\), or an ordered pair of
particles, \((\bs{r}^{\alpha},\bs{r}^{\beta})\).  Following our convention, the pair-position vector \(\hat{\bs{r}}\) is the vector in \(\mathbb{R}^{2d}\) with the block form
\begin{equation}\label{2dvec}
    \hat{\bs{r}}=\begin{pmatrix}\bs{r}^{\alpha}\\\bs{r}^{\beta}  \end{pmatrix}.
\end{equation}
In this context, the indices will always take the values \(\alpha,\beta=1,2\). We begin with some basics and the
fundamental one-particle system.

\subsection{The Fundamental Kinetic Equation}
Suppose the random position vector of a large particle, \bs{r}, depends on time, \(t\mapsto\bs{r}(t)\). From the stochastic limit, described in Section~\ref{sec-pk}, we begin with the assumption
that the random position vector of a \emph{single} large particle obeys a kinematic equation of the form:
\begin{equation}\label{ito1}
    d\bs{r}(t)=\int_{\mathbb{R}^d}\bs{g}(\bs{r}(t)-\bs{q})w(d\bs{q},dt).
\end{equation}
This is a \textit{stochastic integral equation} of It\^{o} type. The integrator, \(w(d\bs{q},dt)\), is standard space-time Gaussian white noise,\footnote{See Walsh~\cite{WA}} which represents the
influence of the medium of small particles upon the large particles. The nature of the forcing on the large particles is determined through the \hypertarget{kernel}{\textit{forcing kernel}} \bs{g}.

The \hyperlink{kernel}{forcing kernel} \(\bs{g}:\mathbb{R}^d\rightarrow\mathbb{R}^d;\bs{r}\mapsto\bs{g}(\bs{r})\), is a vector-valued function of position derived from the underlying physics
governing the interactions between large and small particles. In effect, \bs{g} can be taken proportional to the negative spatial gradient of the velocity distribution of the small
particles.\footnote{If \(\psi(\bs{r})\) is a differentiable scalar-valued isotropic function of position \bs{r}, then its gradient \(\bs{\nabla}_{\bs{r}}\psi(\bs{r})\) is a isotropic vector-valued
function of position. In particular, \(\bs{\nabla}_{\bs{r}}\gamma(|\bs{r}|^2)=2\gamma'(|\bs{r}|^2)\bs{r}\). The converse is also true.} It plays the r\^{o}le of a constitutive function in
Equation~(\ref{ito1}) and in the sequel. As required by the \hyperlink{isotropic}{Principle of Material Frame-Indifference,} the distributions implicit in the derivation of Equation~(\ref{ito1})
must be independent under changes in external observer. Consequently, we start with the following.\footnote{For convenience, we suppress the subscript~\(\varepsilon\) of \bs{g} in
Section~\ref{sec-pk} except when a specific form of the \hyperlink{kernel}{forcing kernel} is used.}
\begin{hyp}[\hyperlink{isotropic}{Isotropy} of the \hyperlink{kernel}{Forcing Kernel}]\label{g-1}
The \hyperlink{kernel}{forcing kernel},~\bs{g}, is isotropic:
\begin{equation}\label{g-2}
\bs{g}(\bs{Qr})=\bs{Q}\bs{g}(\bs{r}),
\end{equation} for all orthogonal transformations  \bs{Q} in \(\mathbb{R}^d\).
\end{hyp}

By this hypothesis and Lemma~(\ref{lem0}), there is a scalar function, called the \hypertarget{forcingfunction}{\textit{forcing function,}}
\(\phi:\mathbb{R}^+\rightarrow\mathbb{R};\xi\,\mapsto\phi(\xi)\) such that\footnote{If \(\bs{r}\mapsto\varphi(|\bs{r}|^2)\) is the distribution function for the velocities of the small particles,
then \(\phi(|\bs{r}|^2)\) is proportional to \(-2\varphi'(|\bs{r}|^2)\). }
\begin{equation}\label{frinv11}
\bs{g}(\bs{r})=\phi(|\bs{r}|^2)\bs{r}.
\end{equation}

\begin{hyp}[Regularity of the \hyperlink{kernel}{Forcing Kernel}]\label{g0}
The scalar-valued forcing function \(\phi\) in Equation~(\ref{frinv11}) is positive, decreasing, and sufficiently regular so that the \hyperlink{kernel}{forcing kernel}, \bs{g}, satisfies:
\begin{enumerate}
\item \bs{g} is twice continuously differentiable,
\item all partial derivatives up through order~\(2\) of all components of \bs{g} are square integrable (over \(\mathbb{R}^d\)),
\item  \(|\bs{g}|^n,\,1\leq n\leq 4\) is integrable (over \(\mathbb{R}^d\)), and
\item \(\lim_{|\bs{r}|\rightarrow\infty}|\bs{g}(\bs{r})|=0\).
\end{enumerate}
\end{hyp}

For a distinct pair of large particles indexed by \(\alpha\) and \(\beta\), \((\alpha\neq\beta)\), we posit the following two-particle kinematic system
\begin{equation}\label{ito2}
    \begin{pmatrix}d\bs{r}^{\alpha}(t)\\d\bs{r}^{\beta}(t) \end{pmatrix}
    =\begin{pmatrix}\int_{\mathbb{R}^d}\bs{g}(\bs{r}^{\alpha}(t)-\bs{q})w(d\bs{q},dt)\\
    \int_{\mathbb{R}^d}\bs{g}(\bs{r}^{\beta}(t)-\bs{q})w(d\bs{q},dt)\end{pmatrix}.
\end{equation}
The system of Equations~(\ref{ito2}) is presented as two copies of equation~(\ref{ito1}) and, as such, does not appear to be coupled. Nevertheless, the two equations are \emph{stochastically
coupled} through the common white noise integrator, \(w(d\bs{q},dt)\),  which represents the medium of small particles in which the two large particles move.\footnote{The nature of this coupling
will be made explicit in Section~\ref{sec-stochlimit}.} If we adjoin a (random) initial state
\begin{equation}\label{ito3a}
    \begin{pmatrix}
    \bs{r}^{\alpha}(0)\\
    \bs{r}^{\beta}(0)
    \end{pmatrix}
    =\begin{pmatrix}
    \bs{r}^{\alpha}_0\\
    \bs{r}^{\beta}_0
    \end{pmatrix}
\end{equation}
to Equations~(\ref{ito2}), the corresponding initial-value problem can be shown to have a unique solution \(t\mapsto\hat{\bs{r}}(t),\,t \geq 0\), which is a Markov process in
\(\mathbb{R}^{2d}\).\footnote{Kotelenez~\cite{KO5, KO4}} We express the initial-value problem~(\ref{ito2}, \ref{ito3a}) succinctly by
\begin{equation}\label{ito3}
    \begin{split}
    d\hat{\bs{r}}(t)
    &=\int_{\mathbb{R}^d}\hat{\bs{g}}(\hat{\bs{r}}(t)-\hat{\bs{q}})w(d\bs{q},dt),\,t\geq 0,\\
    \hat{\bs{r}}(0)&=\hat{\bs{r}}_0,
    \end{split}
\end{equation}
where

    \begin{equation}\label{ito5}
    \hat{\bs{g}}(\hat{\bs{y}})=\hat{\bs{g}}\bigg(\begin{pmatrix}\bs{y}^{\alpha}\\\bs{y}^{\beta}\end{pmatrix}\bigg):=\begin{pmatrix}\bs{g}(\bs{y}^{\alpha})\\
    \bs{g}(\bs{y}^{\beta})
    \end{pmatrix}.
    \end{equation}
In Equation~(\ref{ito3}), and below in Equation~(\ref{ito4}), \(\hat{\bs{q}}\defn\begin{pmatrix}\bs{q}\\\bs{q}\end{pmatrix}\), which emphasizes that the white noise integrator is the same for both
large particles.

\begin{com}
We  systematically consider a solvent containing just two large particles. But the structure we propose can be extended in a natural way to include any finite number of large particles.
\end{com}

For a single large particle, indexed by \(\alpha\), Equation~(\ref{ito1}), together with the initial condition \(\bs{r}^{\alpha}_0\), is the same as the stochastic integral
equation\footnote{Provided the driving noise, \(w(d\bs{q},dt)\), and the initial data, \(\bs{r}^{\alpha}_0\), are independent.}
\begin{equation}\label{ito1a}
    \bs{r}^{\alpha}(t)=\bs{r}^{\alpha}_0+\int_0^t\int_{\mathbb{R}^d}\bigg(\bs{g}(\bs{r}^{\alpha}(s)-\bs{q})\bigg)w(d\bs{q},ds).
\end{equation}
For the pair of distinct large particles, the initial value problem, given by the system of Equations~(\ref{ito3}), is equivalent to the single stochastic integral equation
\begin{equation}\label{ito4}
    \hat{\bs{r}}(t)=\hat{\bs{r}}_0+\int_0^t\int_{\mathbb{R}^d}\bigg(\hat{\bs{g}}(\hat{\bs{r}}(s)-\hat{\bs{q}})\bigg)w(d\bs{q},ds).
\end{equation}
For convenience, we write \(\bs{m}^{\alpha}(t)\) for the stochastic integral involving the process
 \(\bs{r}^{\alpha}(\cdot)\) on the right-hand
side of Equation~(\ref{ito1a}); thus\footnote{Our assumptions imply that the integrals defining the process \(\bs{m}^{\alpha}(\cdot)\) in~(\ref{ito1b}) are continuous, square-integrable, martingales
whenever the processes \(\bs{r}^{\alpha}(\cdot)\) are adapted (by their histories). The letter \bs{m} signifies \textit{\bs{m}artingale}.}
\begin{equation}\label{ito1b}
    \bs{m}^{\alpha}(t):=\int_0^t\int_{\mathbb{R}^d}\bigg(\bs{g}(\bs{r}^{\alpha}(s)-\bs{q})\bigg)w(d\bs{q},ds).
\end{equation}
If the process \(\bs{r}^{\alpha}(\cdot)\) in Equation~(\ref{ito1b}) is a solution to the evolutionary system Equation~(\ref{ito1a}) with random initial condition \(\bs{r}_0^{\alpha}\), we replace
\(\bs{r}^{\alpha}(s)\) in Equation~(\ref{ito1b}) with \(\bs{r}^{\alpha}(s,\bs{r}_0^{\alpha})\) in which case \(\bs{m}^{\alpha}(t)\) should properly be replaced by
\(\bs{m}^{\alpha}(t,\bs{r}_0^{\alpha})\), etc. Generally, we will suppress this dependence on the random initial data.

\begin{com}
It is important to bear in mind that in our model the space-time white noise,~\(w(d\bs{q},dt)\), is obtained as a scaling limit of small particle velocities acting in a short time scale in a small
box on the velocities of the large particles.\footnote{Kotelenez~\cite{KO3}} The significant difference between our model and the traditional one is that in the traditional model each large particle
is driven by its own independent Brownian motion, whereas in our model the large particles are driven by a  \emph{Brownian medium,} which is the \emph{same medium for all large particles.} The
marginal distributions associated with Brownian medium\footnote{These are obtained by convolution of \(\bs{g}(\cdot)\) with \(w(d(\bs{\cdot}),dt)\) and the initial conditions. See also our
discussion in Appendix~\ref{app-marg}.} for each large particle is a traditional Brownian motion, but the joint distribution of two or more large particles is \emph{not} Gaussian and, hence,
\emph{not} Brownian.
\end{com}

To make some computations more specific, we will use a \hyperlink{kernel}{forcing kernel} of the following specific form:
\begin{equation}\label{kernel1}
\bs{g}_{\varepsilon}(\bs{r})=\phi_{\varepsilon}(|\bs{r}|^2)\bs{r}=\kappa_{\varepsilon,d}e^{-\frac{|\bs{r}|^2}{2\varepsilon}}\bs{r},
\end{equation}
where \(\sqrt{\varepsilon}\) is a correlation length and \(\kappa_{\varepsilon,d}\) is a constant that depends on~\(\varepsilon\) and the physical dimension \(d\).\footnote{See
Comment~\ref{com-delcor}.} This special form of \bs{g} is induced by assuming a Maxwell distribution for the velocities of the small particles, which seems physically plausible. Note that the kernel
\(\bs{g}_{\varepsilon}\) is of the form given in Equation~(\ref{frinv11}) (or Equation~(\ref{eq3.3})) whose scalar part, the forcing function, is
\(\xi\mapsto\phi_{\varepsilon}(\xi)=\kappa_{\varepsilon,d}e^{-\frac{\xi^2}{2\varepsilon}}\), so it is \hyperlink{isotropic}{isotropic.}\footnote{The \hyperlink{isotropic}{isotropy} is expected,
since this formulation is based upon the basic physical laws and material assumptions governing the interaction of the small and large particles. Similar structures should obtain for any reasonable
(unimodal) distribution.} We call \(\bs{g}_{\varepsilon}\) the \hypertarget{maxker}{\textit{Maxwell kernel.}} In the sequel we continue to use the symbol \(\bs{g}\), reserving
\(\bs{g}_{\varepsilon}\) specifically for the Maxwell kernel of Equation~(\ref{kernel1}).


Return to the unique smooth forward flow induced by the initial value problem of Equations~(\ref{ito2}, \ref{ito3a}) (or the equivalent Equation~(\ref{ito4})). Since the two equations in
Equation~(\ref{ito2}) are the same, namely copies of Equation~(\ref{ito1}), they each generate the same flow in \(\mathbb{R}^d\) distinguished only through the initial condition for the particle
\bs{r} initially located at \( \bs{r}^{\alpha}_0\):
\begin{equation}\label{flow1}
    t\mapsto\bs{r}^{\alpha}(t)\defn\bs{r}(t,\bs{r}^{\alpha}_0).
\end{equation}
For a pair of large particles, we have the induced flow in \(\mathbb{R}^{2d}\) of Equation~(\ref{ito4}):
\begin{equation}\label{flow2}
    t\mapsto\hat{\bs{r}}(t)\defn\hat{\bs{r}}(t,\hat{\bs{r}}_0).
\end{equation}
Whenever we need to emphasize the role of the driving noise \(w(d\bs{q},dt)\), we write \(\bs{r}(t,\bs{r}^{\alpha}_0;w)\) for \(  \bs{r}(t,\bs{r}^{\alpha}_0)\);
 similarly for the pair-process, we write \( \hat{\bs{r}}(t,\hat{\bs{r}}_0;w) \) for \( \hat{\bs{r}}(t,\hat{\bs{r}}_0)
\),
where the driving noise \(w(d\bs{q},dt)\) is the same for both particles.

\begin{com}
The significance of the conclusion above is its connection with the motion of a large particle in the interacting particle system of many small and some large particles. The forward flow described
above for large particles is a stochastic limit of the motion of the large particles as the number of small particles becomes infinite; the evolution of the positions of the large particles follows
an Einstein-Smoluchowski model.\footnote{The limit is distributional in the sense that, for each \(\alpha\), the distribution associated with the motion of the large particle,
\(\bs{r}_n^{\alpha}(\cdot)  \), in the presence of the background of \(n\) small particles, \(\bs{r}_n^{\alpha}(\cdot)\), (weakly) approaches those associated with  \(\bs{r}^{\alpha}(\cdot)\) as
\(n\rightarrow\infty\). (See Kotelenez~\cite{KO3}.)}
\end{com}

We proceed to describe the properties of this stochastic limit. In particular, we will characterize the joint probability density in \(\mathbb{R}^d\) for a pair of distinct large particles.

\subsection{Properties of the Stochastic Limit}\label{sec-stochlimit}As emphasized earlier, the motions of each large particle (given by an
appropriate \textit{marginal} of the solutions to the system in Equation~(\ref{ito4})\footnote{These are described in Appendix~\ref{app-marg}.}) is Brownian, provided the initial state is
deterministic; however, we will see that the joint motion of the pair is not Brownian. 

First we observe that a pair of large particles that are initially distinct will almost never coincide. More precisely, an argument due to Dawson shows that if\footnote{Private communication.}
\begin{equation}\label{meet1}
    E[|\bs{r}^{\alpha}(0)-\bs{r}^{\beta}(0)|^{-2}]>0,
\end{equation}
then, for any \(T>0\),
\begin{equation}\label{meet2}
    \mathcal{P}[\{\omega\in\Omega:\exists t\in[0,T]\text{ such that }\bs{r}^{\alpha}(t,\omega)=\bs{r}^{\beta}(t,\omega)\}]=0:
\end{equation}
\textit{The probability is zero that two particles, initially distinct, ever coincide in a finite time interval.}

  Define the \((d\times d)-\)matrix
\(\bs{C}\) by
\begin{equation}\label{coeff1}
     \bs{C}:=\int_{\mathbb{R}^d}\bs{g}(\bs{q})\bs{g}^T(\bs{q})\,d\bs{q}.
\end{equation}
Now consider the process \(\bs{m}^{\alpha}(\cdot)\), defined in Equation~(\ref{ito1b}). By the translation-invariance of the integrals, it follows that, for each \(\alpha\), \(\bs{m}^{\alpha}(t)\)
is a \(d\)-dimensional Brownian motion with incremental covariance matrix \(\bs{C}\).\footnote{This follows from a \(d\)-dimensional version of Paul L\'{e}vy's Theorem (see Ethier and
Kurtz~\cite{ET}, Chapter 7, Theorem 1.1 or Theorem~\ref{L-I} in Appendix~\ref{app-qv}.)} By its construction, \bs{C} is non-negative definite and symmetric. It will follow, from Lemma~\ref{posdef},
that the \hyperlink{isotropic}{isotropy} of \bs{g} guarantees that \bs{C} is positive definite (non-degenerate).


For a process consisting of a pair of particles, \((\bs{r}^{\alpha}(\cdot),\bs{r}^{\beta}(\cdot))\), the \textit{cross quadratic variation of the processes} \(\bs{m}^{\alpha}(\cdot)\) with
\(\bs{m}^{\beta}(\cdot)\), denoted by \(\langle\negthinspace\langle\bs{m}^{\alpha},\bs{m}^{\beta}\rangle\negthinspace\rangle(\cdot)\), is a well-defined \((d\times d)\)-matrix valued process.
Definitions of quadratic variation and  cross quadratic variation are given in Appendix~\ref{app-qv}.

The properties of the white noise, \(w(d\bs{q},ds)\), imply
\begin{equation}\label{qvqr1}
    \begin{split}
    \langle\negthinspace\langle\bs{m}^{\alpha},\bs{m}^{\beta}\rangle\negthinspace\rangle(t)
    &=\int_0^t\int_{\mathbb{R}^d}\bs{g}(\bs{r}^{\alpha}(s)-\bs{q})\bs{g}^T(\bs{r}^{\beta}(s)-\bs{q})\,d\bs{q}\,ds,\\
    \end{split}
\end{equation}
which, using shift-invariance, is the same as
\begin{equation}\label{qvqr1a}
\langle\negthinspace\langle\bs{m}^{\alpha},\bs{m}^{\beta}\rangle\negthinspace\rangle(t)=\int_0^t
    \int_{\mathbb{R}^d}\bs{g}((\bs{r}^{\alpha}(s)-\bs{r}^{\beta}(s)-\bs{q})\bs{g}^T(-\bs{q})\,d\bs{q}\,ds.
\end{equation}
The essence of the computation leading to Equation~(\ref{qvqr1}) together with a discussion of some of the consequences that follow are also supplied in Appendix~\ref{app-qv}. Thus, assuming that
the two large particles were initially distinct in the sense that assumption~(\ref{meet1}) holds, the quadratic variation
\(\langle\negthinspace\langle\bs{m}^{\alpha},\bs{m}^{\beta}\rangle\negthinspace\rangle(\cdot)\) depends only on the process that is the \emph{difference} between the positions of the two large
particles, namely \( \bs{r}^{\alpha}(\cdot)-\bs{r}^{\beta}(\cdot) \). More important, \(\langle\negthinspace\langle\bs{m}^{\alpha},\bs{m}^{\beta}\rangle\negthinspace\rangle(\cdot)\) cannot vanish on
any interval, which implies that the processes  \(\bs{m}^{\alpha}(\cdot)\) and \(\bs{m}^{\beta}(\cdot)\) are correlated on every interval.\footnote{The fact that the matrix function
\(\langle\negthinspace\langle\bs{m}^{\alpha},\bs{m}^{\beta}\rangle\negthinspace\rangle(\cdot)\) cannot vanish on any interval follows from the characterization of \( \bs{D}^{\alpha\beta} \) in
Appendix~\ref{pairdiffsec} and the formula in Equation~(\ref{qvqrd3})below. In particular, the lateral component of \( \bs{D}^{\alpha\beta} \) can never vanish.} If the particles coincide, so
\(\alpha=\beta\), then \(\langle\negthinspace\langle\bs{m}^{\alpha},\bs{m}^{\beta}\rangle\negthinspace\rangle(t)=\langle\negthinspace\langle\bs{m}^{\alpha}\rangle\negthinspace\rangle(t)
=\langle\negthinspace\langle\bs{m}^{\beta}\rangle\negthinspace\rangle(t)=t\bs{C}=t\,c\bs{1}\), where \( \langle\negthinspace\langle\bs{m}\rangle\negthinspace\rangle \) denotes the quadratic
variation of the process \bs{m}.\footnote{See Appendix~\ref{app-qv}, especially Theorem~\ref{L-I}. We show below that \(\bs{C}=c\bs{1}\), for some \(c>0\).}

Equation~(\ref{qvqr1}) further implies that, in general, the \emph{joint} motion of two, initially distinct, large particles cannot be Gaussian and, \textit{a fortiori,} cannot be Brownian. However,
the motion of each large particle in the joint motion is Brownian if viewed as a \(d\)-dimensional \textit{marginal} process with deterministic initial conditions.\footnote{See the discussions in
Appendices~\ref{app-marg} and~\ref{app-qv}.} This becomes transparent when we examine the covariance process for the joint motion.

If \(\hat{\bs{r}}=\begin{pmatrix}\bs{r}^{\alpha}\\\bs{r}^{\beta}
\end{pmatrix}\) is a pair in \(\mathbb{R}^{2d}\equiv\mathbb{R}^{d}\times\mathbb{R}^{d}\) corresponding to two large particles, define \(\bs{D}^{\alpha\beta}(\hat{\bs{r}})\)
to be the \((d\times d)\)-matrix function on \(\mathbb{R}^{2d}\equiv\mathbb{R}^{d}\times\mathbb{R}^{d}\)
\begin{equation}\label{qvqrd}
\bs{D}^{\alpha\beta}(\hat{\bs{r}}):=\int_{\mathbb{R}^d}\bs{g}(\bs{r}^{\alpha}-\bs{q})\bs{g}^T(\bs{r}^{\beta}-\bs{q})\,d\bs{q}.
\end{equation}
Using shift invariance, this is the same as
\begin{equation}\label{qvqrd1}
\bs{D}^{\alpha\beta}(\hat{\bs{r}})=\int_{\mathbb{R}^d}\bs{g}((\bs{r}^{\alpha}-\bs{r}^{\beta})-\bs{q})\bs{g}^T(-\bs{q})\,d\bs{q};
\end{equation}
so
\begin{equation}\label{qvqrd2}
\bs{D}^{\alpha\beta}\bigg(\begin{pmatrix}\bs{r}^{\alpha}\\\bs{r}^{\beta}
\end{pmatrix}\bigg)=\bs{D}^{\alpha\beta}\bigg(\begin{pmatrix}0\\\bs{r}^{\alpha}-\bs{r}^{\beta}\end{pmatrix}\bigg)
\end{equation}
and \(\bs{D}^{\alpha\beta}\) is a function only of the difference \( (\bs{r}^{\alpha}-\bs{r}^{\beta})\in\mathbb{R}^d \). Observe that, from Equation~(\ref{qvqr1}), we have
\begin{equation}\label{qvqrd3}
\langle\negthinspace\langle\bs{m}^{\alpha},\bs{m}^{\beta}\rangle\negthinspace\rangle(t,\hat{\bs{r}}_0)=\int_0^t \bs{D}^{\alpha\beta}(\hat{\bs{r}}(s,\hat{\bs{r}}_0))\,ds,
\end{equation}
provided \(\hat{\bs{r}}(\cdot,\hat{\bs{r}}_0)\) is the pair-process determined by the stochastic evolutionary system in Equation~(\ref{ito4}).

Using Equation~(\ref{qvqrd2}) we see that if the two particles coincide, \(\bs{r}^{\alpha}=\bs{r}^{\beta}\), we have \(\bs{D}^{\alpha\alpha}(\hat{\bs{r}})=\bs{D}^{\beta\beta}(\hat{\bs{r}})=\bs{C}\),
where \bs{C} is the \emph{constant,} symmetric, positive definite covariance matrix associated with a single large particle given in Equation~(\ref{coeff1}). The \hyperlink{isotropic}{isotropy} of
\bs{g} implies that \bs{C} satisfies: \(\bs{Q}\bs{C}\bs{Q}^T=\bs{C}\) for all orthogonal transformations \(\bs{Q}\) on \(\mathbb{R}^d\); that is, \bs{C} commutes with all orthogonal transformations.
It follows from the Theorem~(\ref{full}) in Appendix~\ref{app-a} that \bs{C} is a constant, positive multiple of the identity matrix; that is, \(\bs{C}=c\bs{1},\,c>0\).\footnote{This gives another
argument using \hyperlink{isotropic}{isotropy} that \bs{C} is positive definite; for if \(c=0\) then \(\bs{C}=\bs{0}\), which is impossible.} Again, from Equation~(\ref{qvqrd2}), if the two
particles are distinct, \(\bs{r}^{\alpha}\not=\bs{r}^{\beta}\), the matrix \(\bs{D}^{\alpha\beta}(\hat{\bs{r}})\) is generally \emph{not constant} in~\(\hat{\bs{r}}\), for it depends specifically on
the difference, \( \bs{r}^{\alpha}-\bs{r}^{\beta} \), between the positions of the two large particles.

From its definition, Equation~(\ref{qvqrd}), we see that structurally \((\bs{D}^{\alpha\beta})^T=\bs{D}^{\beta\alpha}\). We have already observed that when \(\alpha=\beta\),
\(\bs{D}^{\alpha\alpha}=\bs{D}^{\beta\beta}=\bs{C} \) is symmetric. Less obvious is the fact that \hyperlink{isotropic}{isotropy} implies each \( \bs{D}^{\alpha\beta} \) is symmetric when
\(\alpha\not=\beta\). Thus
\begin{equation}\label{qvqrd3b}
(\bs{D}^{\alpha\beta})^T=\bs{D}^{\beta\alpha}=\bs{D}^{\alpha\beta}.
 \end{equation}
To reveal this symmetry, use shift invariance in the definition Equation~(\ref{qvqrd}), or any of the equivalent versions that follow, to obtain
\begin{equation}\label{qvqrd4}
\bs{D}^{\alpha\beta}(\hat{\bs{r}})=\int_{\mathbb{R}^d}\bs{g}\Big(\frac{1}{2}(\bs{r}^{\alpha}-\bs{r}^{\beta})-\bs{q}\Big)\bs{g}^T\Big(-\frac{1}{2}(\bs{r}^{\alpha}-\bs{r}^{\beta})-\bs{q}\Big)\,d\bs{q}.
\end{equation}
Now \hyperlink{isotropic}{isotropy} implies that \(\bs{g}(-\bs{y})=-\bs{g}(\bs{y})\), so Equation~(\ref{qvqrd4}) becomes
\begin{equation}\label{qvqrd5}
\bs{D}^{\alpha\beta}(\hat{\bs{r}})=\int_{\mathbb{R}^d}\bs{g}\Big(\bs{q}-\frac{1}{2}(\bs{r}^{\alpha}-\bs{r}^{\beta})\Big)\bs{g}^T\Big(\bs{q}+\frac{1}{2}(\bs{r}^{\alpha}-\bs{r}^{\beta})\Big)\,d\bs{q}.
\end{equation}
Structurally, its transpose is
\begin{equation}\label{qvqrd6}
(\bs{D}^{\alpha\beta})^T(\hat{\bs{r}})=\int_{\mathbb{R}^d}\bs{g}\Big(\bs{q}+\frac{1}{2}(\bs{r}^{\alpha}-\bs{r}^{\beta})\Big)\bs{g}^T\Big(\bs{q}-\frac{1}{2}(\bs{r}^{\alpha}-\bs{r}^{\beta})\Big)\,d\bs{q}.
\end{equation}
Finally, by changing the integration variable \bs{q} in Equation~(\ref{qvqrd6}) to \(-\bs{q}\) we recover \(\bs{D}^{\alpha\beta}(\hat{\bs{r}}) \) in Equation~(\ref{qvqrd4}).

For two  particles (\(\bs{r}^{\alpha},\,\bs{r}^{\alpha}\)) we define the corresponding \((2d\times 2d)\)-covariance matrix \(\hat{\bs{D}}(\hat{\bs{r}})\), given in block form by\footnote{Note that
   \(\hat{\bs{D}}(\hat{\bs{r}})
    =\int_{\mathbb{R}^d}\hat{\bs{g}}(\hat{\bs{r}}-\hat{\bs{q}})\hat{\bs{g}}^T(\hat{\bs{r}}-\hat{\bs{q}})\,d\bs{q},\)
where \(\hat{\bs{q}}\defn\begin{pmatrix}\bs{q}\\\bs{q}
\end{pmatrix}\).}

\begin{equation}\label{qyqr2}
   \begin{split}\hat{\bs{D}}(\hat{\bs{r}})
    &:=
    \begin{pmatrix}\bs{D}^{\alpha\alpha}(\hat{\bs{r}})&\bs{D}^{\alpha\beta}(\hat{\bs{r}})\\
    \bs{D}^{\beta\alpha}(\hat{\bs{r}})&\bs{D}^{\beta\beta}(\hat{\bs{r}})  \end{pmatrix}\\
    &=\begin{pmatrix}\bs{C}&\bs{D}^{\alpha\beta}(\hat{\bs{r}})\\
    \bs{D}^{\alpha\beta}(\hat{\bs{r}})&\bs{C}  \end{pmatrix}\\
    &=\begin{pmatrix}c\bs{1}&\bs{D}^{\alpha\beta}(\hat{\bs{r}})\\
    \bs{D}^{\alpha\beta}(\hat{\bs{r}})&c\bs{1}  \end{pmatrix}.
    \end{split}
\end{equation}
Structurally, the matrix \(\hat{\bs{D}}(\hat{\bs{r}})\) must be symmetric and non-negative definite; in addition, each block is symmetric. In Lemma~\ref{posdef}, we show that whenever the two
particles are distinct, \(\bs{r}^{\alpha}\not =\bs{r}^{\beta}\) the matrix \(\hat{\bs{D}}(\hat{\bs{r}})\) is positive definite (non-degenerate). Moreover, under our assumptions,
\(\bs{D}^{\alpha\beta}(\hat{\bs{r}})\rightarrow\bs{0}\) or, equivalently, \(\hat{\bs{D}}(\hat{\bs{r}})\rightarrow c\bs{1}  \) as \(|\bs{r}^{\alpha}-\bs{r}^{\beta}|\rightarrow\infty\). Observe that
the diagonal blocks of \(\hat{\bs{D}}(\hat{\bs{r}})\) are \emph{constant} but the off-diagonal blocks (cross terms) are \emph{non-constant} and \emph{non-zero} in \(\hat{\bs{r}}\);
they depend specifically on the difference \( \bs{r}^{\alpha}-\bs{r}^{\beta}\).

\begin{com}
For \emph{any} kernel, under the normalization of Comment~\ref{com-delcor}, \(c=\|\bs{C}\|=1\), the operator norm of \bs{C}. The Hilbert-Schmidt matrix norm of \(\bs{C}\) is then
\(|\bs{C}|=\sqrt{\bs{C}\bullet\bs{C}}=c\sqrt{d}=\sqrt{d}\). The covariance matrix then has the form
\begin{equation}\label{qyqr2a}
 \hat{\bs{D}}(\hat{\bs{r}})=\begin{pmatrix}\bs{1}&\bs{D}^{\alpha\beta}(\hat{\bs{r}})\\
    \bs{D}^{\alpha\beta}(\hat{\bs{r}})&\bs{1}  \end{pmatrix}.
\end{equation}
And, in the limit as \(\sqrt{\varepsilon}\downarrow 0\), \( \bs{D}^{\alpha\beta}(\hat{\bs{r}})\rightarrow\bs{0}\) or, equivalently, \( \hat{\bs{D}}(\hat{\bs{r}})\rightarrow\bs{1}\).
\end{com}

\subsection{The Generator for the Difference Process}\label{sec-gen}Recall that the pair-process  \(t\mapsto\hat{\bs{r}}(t,\hat{\bs{r}}_0)\) is the solution to the stochastic evolutionary
system in Equation~(\ref{ito4}); it is a homogeneous Markov process such that \(\hat{\bs{r}}(0,\hat{\bs{r}}_0)= \hat{\bs{r}}_0 \). This process is associated with a semigroup of linear operators
\(\{\hat{T}_t:t\geq 0\}\) through its \textit{transition probability function \(\hat{P}\):}
\begin{equation}\label{tp0}
    (t,\hat{\bs{x}},B)\mapsto
\hat{P}(t,\hat{\bs{x}},B):=\mathcal{P}[\hat{\bs{r}}(t,\hat{\bs{r}}_0)\in B\big|\hat{\bs{r}}_0=\hat{\bs{x}}].
\end{equation}
That is, \((t,\hat{\bs{x}},B)\rightarrow\hat{P}(t,\hat{\bs{x}},B)\) gives the conditional probability that the pair-process \(\hat{\bs{r}}(
    t,\hat{\bs{r}}_0)\) lies in the Borel set \(B\) at time \(t\) given that its state at time \(t=0\) was
\(\hat{\bs{r}_0}=\hat{\bs{x}}\).\footnote{Since the process is homogeneous with respect to time, the probability does not depend on the absolute time, but only on the \emph{time interval \(t\);}
that is,
\begin{equation*}
    (t,\hat{\bs{x}},B)\mapsto
\hat{P}(t,\hat{\bs{x}},B):=\mathcal{P}[\hat{\bs{r}}(s+t,\hat{\bs{r}}_0)\in B|\hat{\bs{r}}(s,\hat{\bs{r}}_0)=\hat{\bs{x}}],
\end{equation*}
for any \(t,s\geq 0\). Here, \(B\) denotes an arbitrary Borel
subset of \(\mathbb{R}^{2d}\).} Let \(\hat{f}\) be a bounded, measurable, real-valued function on \(\mathbb{R}^{2d}\). Define \(\hat{T}_t \hat{f}\) by the conditional expectation given through
\begin{equation}\label{tp2}
    (\hat{T}_t
    \hat{f})(\hat{\bs{x}})\equiv E_{\hat{\bs{x}}}[\hat{f}(\hat{\bs{r}}(t,\hat{\bs{r}}_0))]:=E[\hat{f}(\hat{\bs{r}}(t,\hat{\bs{r}}_0))\big|\hat{\bs{r}}_0=\hat{\bs{x}}]=\int_{\mathbb{R}^{2d}}\hat{f}(\hat{\bs{y}})\hat{P}(t,\hat{\bs{x}},d\hat{\bs{y}});
\end{equation}
so \((\hat{T}_t \hat{f})(\hat{\bs{x}})\) is the expected value of \(\hat{f}(\hat{\bs{r}}(t,\hat{\bs{r}}_0))\) given that \(\hat{\bs{r}}(0,\hat{\bs{r}}_0)=\hat{\bs{r}}_0  =\hat{\bs{x}}\).  In
particular, if \(\hat{f}=1_{B}\), the indicator function of the Borel set \(B\), then \((\hat{T}_t 1_{B})(\hat{\bs{x}})=\hat{P}(t,\hat{\bs{x}},B)\), which means that the transition probability can
be recovered from the semigroup.

The (infinitesimal) generator, \(\hat{A}\), of this semigroup operating on \(\hat{f}\) at \(\hat{\bs{x}} \) is the \emph{mean instantaneous rate of change of
\(\hat{f(}\hat{\bs{r}}(t,\hat{\bs{r}}_0))\) at \(t=0\) given that \(\hat{\bs{r}}(0,\hat{\bs{r}}_0)=\hat{\bs{x}} \)}:\footnote{The limit, whenever it exists, is strong in the sense that it is with
respect to the norm \(||\hat{f}||=\sup_{\hat{\bs{x}}\in\mathbb{R}^{2d}}\{| \hat{f}(\hat{\bs{x}}) |\}\).}
\begin{equation}\label{tp3}
    (\hat{A}\hat{f})(\hat{\bs{x}}):=\lim_{t\downarrow 0}\frac{1}{t}\big((\hat{T}_t \hat{f})(\hat{\bs{x}})-\hat{f}(\hat{\bs{x}})\big).
\end{equation}
The generator, \( \hat{A} \), of the semigroup \(\{\hat{T}_t:t\geq 0\}\) is determined explicitly through the covariance matrix  \(\hat{\bs{D}}(\hat{\bs{x}})\) of Equation~(\ref{qyqr2}).
\begin{thm}[Generator of the semigroup \(\{\hat{T}_t:t\geq 0\}\)]\label{semigrp}
The generator, \( \hat{A} \), of the semigroup \(\{\hat{T}_t:t\geq 0\}\) is determined by the second order elliptic differential operator, defined for twice continuously differentiable functions
\(\hat{f}\) on \(\mathbb{R}^{2d}\) that vanish at infinity, by\footnote{The domain of \(\hat{A}\), \(\mathcal{D}(\hat{A})\), densely contains the twice continuously differentiable functions on
\(\mathbb{R}^{2d}\) that vanish at infinity, \(\mathbb{C}^2_0(\mathbb{R}^{2d},\mathbb{R})\). See Ethier and Kurtz~\cite{ET}. }
\begin{equation}\label{g2}
\begin{split}
    (\hat{A}\hat{f})(\hat{\bs{x}}):&=\frac{1}{2}\sum_{l,m=1}^{2d}\hat{D}_{lm}(\hat{\bs{x}})\bigg(\frac{\partial^2}{\partial \hat{x}_l\partial
    \hat{x}_m}\hat{f}\bigg)(\hat{\bs{x}})\\
    &=\frac{1}{2}\sum_{\alpha,\beta=1}^2\sum_{i,j=1}^d  D_{ij}^{\alpha\beta}
    \bigg(\begin{pmatrix}\bs{x}^{1}\\\bs{x}^{2}  \end{pmatrix}\bigg)\bigg(\frac{\partial^2}{\partial x_i^{\alpha}\partial
    x_j^{\beta}}\hat{f}\bigg)\bigg(\begin{pmatrix}\bs{x}^{1}\\\bs{x}^{2}  \end{pmatrix}\bigg),
    \end{split}
\end{equation}
where \(\hat{\bs{D}}(\hat{\bs{x}})\) is the covariance matrix defined in Equation~(\ref{qyqr2}). Therefore we refer to \(\hat{\bs{D}}(\hat{\bs{x}})\) as the
\hyperlink{pairdiffusion}{\textit{diffusion matrix for the pair-process.}}
\end{thm}
An outline of the proof of this result is given in Appendix~\ref{app-semigrp}.

\begin{com}\label{com-diffnorm}
The normalization that we have used to comply with Desideratum~\ref{desid3b} is equivalent to the requirement that \(\bs{C}=\bs{1}\) and, in the limit as \(\varepsilon\downarrow 0\), the
infinitessimal generator \(\hat{A}\) becomes one half the Laplacian. That is,  \(\bs{D}^{\alpha\beta}\rightarrow\bs{0}\), as \(\varepsilon\downarrow 0\). In other words, in this limit, each particle
experiences a \emph{standard} Brownian motion and these motions are independent.

We could replace the requirement that the limititng generator be one half the Laplacian with the requirement that it be some other constant times the Laplacian, say Einstein's diffusion constant,
\(D\), times the Laplacian. Such a re-normalization will have no effect on our results.
\end{com}

We now focus on the difference-process \((\bs{r}^2(\cdot)-\bs{r}^1(\cdot))\) in \(\mathbb{R}^d\). The generator for this process can be extracted from the generator \(\hat{A}\) for the process in
\(\mathbb{R}^{2d}\), defined in Equation~(\ref{g2}), by means of an orthogonal transformation (rotation) in \(\mathbb{R}^{2d}\) followed by a projection. Here are the main steps. Define the
\((2d\times 2d)\) proper orthogonal matrix (rotation) \(\hat{\bs{R}}\) in block form by
\begin{equation}\label{g3}
    \hat{\bs{R}}=\frac{1}{\sqrt{2}}
    \begin{pmatrix}
        \bs{1}&\bs{1}\\
        -\bs{1}&\bs{1}
    \end{pmatrix}
\end{equation}
and the \((2d\times 2d)\) projection matrix \(\hat{\bs{P}}\) by the block form
\begin{equation}\label{g4}
    \hat{\bs{P}}=
    \begin{pmatrix}
        \bs{0}&\bs{0}\\
        \bs{0}&\bs{1}
    \end{pmatrix}.
\end{equation}
The matrix \(\hat{\bs{P}}\) determines the projection in \(\mathbb{R}^{2d}\) onto the subspace \( \bigg\{ \begin{pmatrix}\bs{x}^1\\
\bs{x}^2 \end{pmatrix}\in\mathbb{R}^{2d}: \bs{x}^1=\bs{0} \bigg\} \), which we will identify with \(\mathbb{R}^d\). That is, we identify \((\leftrightarrow)\) the \(\mathbb{R}^{2d}\) vector
\(\begin{pmatrix}\bs{0}\\\bs{b}  \end{pmatrix}\) with the \(\mathbb{R}^{d}\) vector \(\bs{b}\). In particular, we extract the (normalized) difference \(\frac{1}{\sqrt{2}}(\bs{r}^2-\bs{r}^1)\) from
\(\hat{\bs{R}}\) through the rotation followed by the projection:
\begin{equation}\label{g6}
\frac{1}{\sqrt{2}}(\bs{r}^2-\bs{r}^1)\leftrightarrow\hat{\bs{P}}\hat{\bs{R}}\hat{\bs{r}}.
\end{equation}
Write \(\check{\bs{x}}=\hat{\bs{R}}\hat{\bs{x}}\) for rotated vectors \(\hat{\bs{x}}\) in  \(\mathbb{R}^{2d}\). Then the semigroup \(\{\hat{T}_t:t\geq 0\}\) induces the ``rotated" semigroup
\(\{\check{T}_t:t\geq 0\}\)
whose ``rotated" generator we denote by \(\check{A}\). We will obtain a generator, \(\tilde{A}\), for the difference-process by extracting that part of \(\check{A}\) associated with the
difference-process. It is convenient, and perhaps more illuminating, to work directly through the original transition function, \(\hat{P}\), instead of the transition function, \(\check{P}\),
associated with the rotated semigroup, \(\{\check{T}_t:t\geq 0\}\).

 For any Borel set \(B\) in \(\mathbb{R}^d\), consider the cylinder set \(\mathbb{R}^d\times B\) in \(\mathbb{R}^{2d}\). Write \(\Gamma_B:=\hat{\bs{R}}^T(\mathbb{R}^d\times B)\) for the rotated
cylinder set. Now \(\hat{P}(t,\hat{\bs{r}},\Gamma_B)\) is the probability that the two particle system, which started at \(\hat{\bs{r}}_0\) at time \(t=0\), lies in \(\Gamma_B\) at time \(t\), so
\begin{equation}\label{gen8}
    \begin{split}
    \hat{P}(t,\hat{\bs{r}}_0,\Gamma_B)  &=\mathcal{P}[\hat{\bs{r}}(t,\hat{\bs{r}}_0)\in
    \hat{\bs{R}}^T(\mathbb{R}^d\times B)|\hat{\bs{r}}(0,\hat{\bs{r}}_0)=\hat{\bs{r}}_0]\\
    &=\mathcal{P}\bigg[\frac{1}{\sqrt{2}}
    \begin{pmatrix}
        \bs{r}^2(t,\bs{r}^2_0)+\bs{r}^1(t,\bs{r}^1_0)\\
        \bs{r}^2(t,\bs{r}^2_0)-\bs{r}^1(t,\bs{r}^1_0)\end{pmatrix}
    \in\mathbb{R}^d\times B]\\
        &=\mathcal{P}\bigg[\frac{1}{\sqrt{2}}(\bs{r}^2(t,\bs{r}^2_0)-\bs{r}^1(t,\bs{r}^1_0))\in B\bigg].
    \end{split}
\end{equation}
Recall that the solutions \(\bs{r}^1(\cdot,\bs{r}^1_0)\) and \(\bs{r}^2(\cdot,\bs{r}^2_0)\) have the \emph{same} driving noise \(w(d\bs{r},dt)\), which we emphasize by writing
\(\bs{r}^1(\cdot,\bs{r}^1_0;w)\) and \(\bs{r}^2(\cdot,\bs{r}^2_0;w)\).

Two random variables, say \(X\) and \(Y\), are said to be equivalent, written \(X\sim Y\), whenever they have the same distribution. The next result shows that the pair-process has the following
restricted translation invariance:\footnote{Kotelenez~\cite{KO4}}
\begin{lem}\label{lem1}
For all \(\bs{h}\) in \(\mathbb{R}^d\)
\begin{equation}\label{gen9}
    \begin{pmatrix}
    \bs{r}^1(\cdot,\bs{r}^1_0;w)+\bs{h}\\
    \bs{r}^2(\cdot,\bs{r}^2_0;w)+\bs{h}
    \end{pmatrix}
    \sim
    \begin{pmatrix}
    \bs{r}^1(\cdot,\bs{r}^1_0+\bs{h};w)\\
    \bs{r}^2(\cdot,\bs{r}^2_0+\bs{h};w)
    \end{pmatrix},
\end{equation}
considered as \(C([0,\infty),\mathbb{R}^{2d})\)-valued random variables.
\end{lem}
As a consequence of Lemma~(\ref{lem1}), the transition probability in Equation~(\ref{gen8}) depends only upon the  \emph{(normalized) difference}
\(\bs{x}_0:=\frac{1}{\sqrt{2}}(\bs{r}^2_0-\bs{r}^1_0)\). Therefore, suppressing the zero subscript, we can define the following transition probability:
\begin{equation}\label{gen10}
    \tilde{P}(t,\bs{x},B):=\hat{P}(t,\hat{\bs{r}},\Gamma_B)\big|_{\frac{1}{\sqrt{2}}(\bs{r}^2-\bs{r}^1)=\bs{x}}
    =\hat{P}\bigg(t,\hat{\bs{R}}^T\begin{pmatrix}\bs{0}\\ \bs{x} \end{pmatrix},\Gamma_B\bigg).
\end{equation}
The transition probability, \(\tilde{P}\), is that part of the transition probability, \(\hat{P}\) (or~\(\check{P}\)), associated with the difference-process.\footnote{The transition probability \(
\tilde{P} \) may also be considered a marginal transition probability.} We now obtain the generator,~\(\tilde{A}\), associated with \(\tilde{P}\). It will turn out to be that obtained from
\(\hat{A}\) through the change-of-variables given by the rotation followed by projection.

Following our previous notational scheme, if \(\hat{f}\in C_0^2(\mathbb{R}^{2d},\mathbb{R})\), then \(\check{f}\in C_0^2(\mathbb{R}^{2d},\mathbb{R})\) is just
\(\check{f}=\hat{f}\circ\hat{\bs{R}}^T\). Furthermore, if \(\check{f}\bigg(\begin{pmatrix}\bs{x}^1\\\bs{x}^2
\end{pmatrix}\bigg)=\check{f}\bigg(\begin{pmatrix}\bs{0}\\\bs{x}^2
\end{pmatrix}\bigg)\), for all \(\bs{x}^1,\bs{x}^2\in\mathbb{R}^d\), we identify \((\leftrightarrow)\) the latter with the element
\(\tilde{f}(\bs{x}^2)\) in \(C_0^2(\mathbb{R}^d,\mathbb{R})\).

Using Definition~(\ref{qvqrd}), define the \((d\times d)\)-matrix valued \(\tilde{\bs{D}}\) function on \(\mathbb{R}^d\) by
\begin{equation}\label{gen11}
    \tilde{\bs{D}}(\bs{x}):=\bs{D}^{11}\bigg(\hat{\bs{R}}^T\begin{pmatrix}\bs{0}\\\bs{x}
\end{pmatrix}\bigg)-\bs{D}^{12}\bigg(\hat{\bs{R}}^T\begin{pmatrix}\bs{0}\\\bs{x}
\end{pmatrix}\bigg)=\bs{C}-\bs{D}^{12}\bigg(
\frac{1}{\sqrt{2}}\begin{pmatrix}\bs{-x}\\\bs{x}
\end{pmatrix}\bigg).
\end{equation}
Henceforth, we call \(\tilde{\bs{D}}(\bs{x})\) the \hypertarget{pairmatrix}{\textit{diffusion matrix for the difference-process.}}

From the definition of \(\hat{\bs{D}}\) in Equation~(\ref{qyqr2}) and the comments following it, recall that \(\bs{D}^{11}(\hat{\bs{r}})=\bs{D}^{22}(\hat{\bs{r}})=\bs{C}\), a constant matrix, and
\(\bs{D}^{21}(\hat{\bs{r}})=\bs{D}^{12}(\hat{\bs{r}})\), where the latter matrix is not constant and depends only on the normalized difference \(\bs{x}:=\frac{1}{\sqrt{2}}(\bs{r}^2-\bs{r}^1)\). So
\( \tilde{\bs{D}}(\bs{x}) \) defined in Equation~(\ref{gen11}) is

\begin{equation}\label{gen20}
    \tilde{\bs{D}}(\bs{x})= (\hat{\bs{P}}\hat{\bs{R}})\hat{\bs{D}}\bigg(\hat{\bs{R}}^T\begin{pmatrix}\bs{0}\\\bs{x}
\end{pmatrix}\bigg)(\hat{\bs{P}}\hat{\bs{R}})^T,
\end{equation}
the rotation and projection of \(\hat{\bs{D}}\).  It is important to note that the \hyperlink{isotropic}{isotropy} of \bs{g} implies that the diffusion matrix \( \tilde{\bs{D}}(\bs{x})\) is
isotropic:
\begin{equation}\label{frame1}
\tilde{\bs{D}}(\bs{x})=\bs{Q}^T\tilde{\bs{D}}(\bs{Q}\bs{x})\bs{Q},
\end{equation}
for every \bs{x} in \(\mathbb{R}^d\) and every orthogonal transformation \bs{Q} on \(\mathbb{R}^d\). Putting all this together yields
\begin{thm}[The Diffusion Matrix for the Difference-Process]\label{gen13}
The transition probability function \((t,\bs{x},B)\mapsto\tilde{P}(t,\bs{x},B)\), defined in Equation~(\ref{gen10}), generates the \emph{difference diffusion process,}  a Markov-Feller process in
\(\mathbb{R}^d\) whose generator, \(\tilde{A}\), is given by the second order elliptic operator

\begin{equation}\label{gen12}
    (\tilde{A}\tilde{f})(\bs{x}):=\frac{1}{2}\sum_{i,j=1}^{d}\tilde{D}_{ij}(\bs{x})\bigg(\frac{\partial^2}{\partial x_i\partial
    x_j}\tilde{f}\bigg)(\bs{x}),
\end{equation}
where the \hyperlink{pairmatrix}{diffusion matrix for the difference-process,} \( \tilde{\bs{D}}(\bs{x}) \), is defined in Equation~(\ref{gen11}) (or, equivalently, in Equation~(\ref{gen20})).
\end{thm}

Observe that the differential operator in Equation~(\ref{gen12}) is precisely that which would be obtained from the formula in Equation~(\ref{g2}) by the change-of-variables given through the
rotation followed by the projection. Indeed, if we change coordinates in \(\mathbb{R}^{2d}\) by the rotation, \(\hat{\bs{x}}\rightarrow\check{\bs{x}}=\hat{\bs{R}}\hat{\bs{x}}\) in
Equation~(\ref{g2}) and denote the differential operator \(\hat{A}\), now expressed in the rotated coordinates, by \(\check{A}\), we get

\begin{equation}\label{gen16}
\begin{split}
    (\check{A}\check{f})(\check{\bs{x}}):&=\frac{1}{2}\sum_{l,m=1}^{2d}\check{D}_{lm}(\check{\bs{x}})
    \bigg(\frac{\partial^2}{\partial \check{x}_l\partial
    \check{x}_m}\check{f}\bigg)(\check{\bs{x}})\\
    &=\frac{1}{2}\sum_{i,j=1}^d  (C_{ij}+D_{ij}^{12})
    \bigg(\frac{1}{\sqrt{2}}\begin{pmatrix}-\bs{x}^2\\\bs{x}^2  \end{pmatrix}\bigg)\bigg(\frac{\partial^2}{\partial x_i^1\partial
    x_j^1}\check
    {f}\bigg)\bigg(\begin{pmatrix}\bs{x}^{1}\\\bs{x}^{2}  \end{pmatrix}\bigg)\\
    &+\frac{1}{2}\sum_{i,j=1}^d  (C_{ij}-D_{ij}^{12})
    \bigg(\frac{1}{\sqrt{2}}\begin{pmatrix}-\bs{x}^2\\\bs{x}^2  \end{pmatrix}\bigg)\bigg(\frac{\partial^2}{\partial x_i^2\partial
    x_j^2}\check
    {f}\bigg)\bigg(\begin{pmatrix}\bs{x}^{1}\\\bs{x}^{2}  \end{pmatrix}\bigg),
    \end{split}
\end{equation}
where  \(\check{\bs{D}}(\check{\bs{x}}):=\hat{\bs{R}}\hat{\bs{D}}(\hat{\bs{R}}^T\check{\bs{x}})\hat{\bs{R}}^T\) is the block-diagonal matrix valued function of
\(\check{\bs{x}}=\bigg(\begin{pmatrix}\bs{x}^{1}\\\bs{x}^{2}
\end{pmatrix}\bigg)\) given by

\begin{equation}\label{gen17}
   \begin{split}\check{\bs{D}}(\check{\bs{x}})
    &=
        \begin{pmatrix}
        \Bigg(\bs{C}+\bs{D}^{12} \bigg(\frac{1}{\sqrt{2}}\begin{pmatrix}\bs{x}^1-\bs{x}^2\\\bs{x}^1+\bs{x}^2  \end{pmatrix}\bigg)\Bigg)  &   \bs{0}\\
        \bs{0}  & \Bigg(\bs{C}-\bs{D}^{12} \bigg(\frac{1}{\sqrt{2}}\begin{pmatrix}\bs{x}^1-\bs{x}^2\\\bs{x}^1+\bs{x}^2
        \end{pmatrix}\bigg)\Bigg)
        \end{pmatrix}.\\
    \end{split}
\end{equation}
Since \(\bs{D}^{12}\) depends only on \(\bs{x}^2\), we may take \(\bs{x}^1=\bs{0}\).\footnote{Recall that \(\bs{x}^1=\frac{1}{\sqrt{2}}(\bs{r}^1+\bs{r}^2)\)) and
\(\bs{x}^2=\frac{1}{\sqrt{2}}(\bs{r}^1-\bs{r}^2)\)} Thus,
\begin{equation}\label{gen25}
   \bs{D}^{12}(\hat{\bs{R}}^T\check{\bs{x}})= \bs{D}^{12} \bigg(\frac{1}{\sqrt{2}}\begin{pmatrix}\bs{x}^1-\bs{x}^2\\\bs{x}^1+\bs{x}^2  \end{pmatrix}\bigg)
   =\bs{D}^{12} \bigg(\frac{1}{\sqrt{2}}\begin{pmatrix}-\bs{x}^2\\\bs{x}^2  \end{pmatrix}\bigg).
\end{equation}
\(\check{A}\) must be the generator of the rotated semigroup  \(\{\check{T}_t:t\geq 0\}\). Finally, if we apply \(\check{A}\) to functions of the form \(\tilde{f}\) and call the result
\(\tilde{A}\tilde{f}\) we recover Equation~(\ref{gen12}), provided the symbol \(\bs{x}^2\) is replaced by \(\bs{x}\). We need Lemma~(\ref{lem1}) to justify this last step, namely that we indeed have
the generator of \(\tilde{P}\).
\begin{define}\label{def-x}
Henceforth, \(\bs{x}\) will denote the \emph{normalized difference} \(\bs{x}:=\frac{1}{\sqrt{2}}(\bs{r}^2-\bs{r}^1) \) and~\(\xi\) will denote its magnitude: \(\xi:=|\bs{x}|\). We also use the term
\emph{separation} for this magnitude.
\end{define}
Since \(\hat{\bs{r}}\rightarrow\bs{D}^{12}(\hat{\bs{r}})\equiv\bs{D}^{12}\bigg(\begin{pmatrix}\bs{r}^1\\
\bs{r}^2  \end{pmatrix}\bigg)= \bs{D}^{12}\bigg(\begin{pmatrix}\bs{0}\\
\bs{r}^2-\bs{r}^1  \end{pmatrix}\bigg)\)  depends only on~\(\bs{x}\), we will usually identify it \((\leftrightarrow)\) with a function \(\bs{x}\rightarrow\bs{D}(\bs{x})\); that is
\begin{equation}\label{diffmatrix}
\bs{D}(\bs{x}):=\bs{D}^{12} \bigg(\frac{1}{\sqrt{2}}\begin{pmatrix}-\bs{x}\\ \bs{x}  \end{pmatrix}\bigg)=\bs{D}^{12} \bigg(\begin{pmatrix}\bs{0}\\ \sqrt{2}\bs{x}
\end{pmatrix}\bigg)
\end{equation}

Earlier we asserted that whenever \(\bs{r}^1 \not =\bs{r}^2\) the matrix \(\hat{\bs{D}}(\hat{\bs{r}})\) was  positive definite. It is equivalent to show that this assertion is true for the block
diagonal matrix \(\check{\bs{D}}(\check{\bs{r}})\). The following Lemma, proved in Appendix~\ref{app-c}, suffices.

\begin{lem}\label{posdef}
The two \((d\times d)\)-blocks \( \bs{C}\pm\bs{D}^{12}(\hat{\bs{r}})\leftrightarrow\bs{C}\pm\bs{D}(\bs{x}) \) are each positive definite whenever \(\bs{x}\not =\bs{0}\).\footnote{ Of course, if
\(\bs{r}^1=\bs{r}^2\), then \( \bs{C}+\bs{D}^{12}=2\bs{C} \) is positive definite while \( \bs{C}-\bs{D}^{12}=\bs{0}\), in which case \(\hat{\bs{D}}\) is non-negative definite but not positive
definite.}
\end{lem}

\subsection{The Diffusion Matrix for the Difference Process}\label{pairdiffsec}
Henceforth we will be concerned exclusively with the \hyperlink{pairmatrix}{diffusion matrix for the difference-process} or, succinctly, the  \hyperlink{pairmatrix}{diffusion matrix}; that is, with
the matrix function \(\bs{x}\mapsto\tilde{\bs{D}}(\bs{x}):=\bs{C}-\bs{D}(\bs{x})\). Earlier we saw that the \hyperlink{isotropic}{isotropy} of the \hyperlink{kernel}{forcing kernel} \bs{g} implied
that \bs{C} had the form \(\bs{C}=c\bs{1}\), for some constant \(c>0\). We now use the \hyperlink{isotropic}{isotropy} of \bs{g} to reveal the deeper structure of \(\bs{D}(\bs{x})\). We know from
Equation~(\ref{frame1}), that it, too, is isotropic; but we go further.

Using Equation~(\ref{qvqrd}) and \bs{x} of Definition~\ref{def-x} we easily see that
\begin{equation}\label{d1}
\bs{D}(\bs{x})=\int_{\mathbb{R}^d}\bs{g}\bigg(\bs{q}+\frac{\sqrt{2}}{2}\bs{x}\bigg)\bs{g}^T\bigg(\bs{q}-\frac{\sqrt{2}}{2}\bs{x}\bigg)\,d\bs{q}.
\end{equation}
For convenience,  temporarily set \(\bs{x}=\sqrt{2}\bs{y}\) and use the representation of \bs{g} in Equation~(\ref{frinv11}), to get
\begin{equation}\label{d2}
\bs{D}(\sqrt{2}\bs{y})=\int_{\mathbb{R}^d}(\bs{q}+\bs{y})(\bs{q}-\bs{y})^T\phi(|\bs{q}+\bs{y}|^2)\phi(|\bs{q}-\bs{y}|^2)        \,d\bs{q}.
\end{equation}
Since we have already shown that  \(\bs{D}(\sqrt{2}\bs{y})\) is symmetric, Equation~(\ref{d2}) reduces to
\begin{equation}\label{d3}
\bs{D}(\sqrt{2}\bs{y})=\int_{\mathbb{R}^d}\bs{q}\bs{q}^T\phi(|\bs{q}+\bs{y}|^2)\phi(|\bs{q}-\bs{y}|^2)\,d\bs{q}-\bs{y}\bs{y}^T\int_{\mathbb{R}^d}\phi(|\bs{q}+\bs{y}|^2)\phi(|\bs{q}-\bs{y}|^2)\,d\bs{q}.
\end{equation}
Define a symmetric-matrix-valued function of \bs{y} by
\begin{equation}\label{d4}
\bs{A}(\bs{y}):=\int_{\mathbb{R}^d}\bs{q}\bs{q}^T\phi(|\bs{q}+\bs{y}|^2)\phi(|\bs{q}-\bs{y}|^2)\,d\bs{q}
\end{equation}
and a scalar-valued function of \bs{y} by
\begin{equation}\label{d5}
b(\bs{y}):=\int_{\mathbb{R}^d}\phi(|\bs{q}+\bs{y}|^2)\phi(|\bs{q}-\bs{y}|^2)\,d\bs{q},
\end{equation}
so that
\begin{equation}\label{d6}
\bs{D}(\sqrt{2}\bs{y})=\bs{A}(\bs{y})-b(\bs{y})\bs{y}\bs{y}^T.
\end{equation}
A bit of algebra and a change-in-variables yields the following: For any orthogonal transformation \(\bs{Q}\) on \(\mathbb{R}^d\) we have
\begin{equation}\label{d7}
\bs{Q}\bs{A}(\bs{y})\bs{Q}^T=\bs{A}(\bs{Q}\bs{y})\text{ and }b(\bs{y})=b(\bs{Q}\bs{y});
\end{equation}
so each is isotropic in the sense of Definition~(\ref{def1}). From Lemma~(\ref{lem0}) (using the alternate form of Part~(\ref{lem03})) we see that \(\bs{A}(\bs{y})\) and \(b(\bs{y})\) must have the
form
\begin{equation}\label{d9}
\bs{A}(\bs{y})=\alpha(|\bs{y}|^2)\bs{P}(\bs{y})+\alpha_{\perp}(|\bs{y}|^2)\bs{P}^{\perp}(\bs{y}) \text{ and } b(\bs{y})=\beta(|\bs{y}|^2)
\end{equation}
for some scalar-valued functions \(\alpha,\alpha_{\perp},\beta:\mathbb{R}^+\rightarrow\mathbb{R};\xi\mapsto\alpha(\xi^2),\alpha_{\perp}(\xi^2),\beta(\xi^2)\). Since \(\phi\) is non-negative and not
identically zero, \(\bs{A}(\bs{y})\) is positive definite, so we can also conclude that, for every \(\bs{y}\not =\bs{0}\),
\begin{equation}\label{d9a}
\alpha(|\bs{y}|^2)>0\qquad\text{and}\qquad\alpha_{\perp}(|\bs{y}|^2)>0.
\end{equation}
At \(\bs{y}=\bs{0}\)
\begin{equation}\label{d9aa}
\alpha(0):=c>0\qquad\text{and}\qquad\alpha_{\perp}(0)=c>0.
\end{equation}
Of course, \(\beta(|\bs{y}|^2)>0\) for all \bs{y}, as well. Our regularity assumptions on the \hyperlink{kernel}{forcing kernel} \bs{g}, Hypotheses~\ref{g0},  also imply the limits
\begin{equation}\label{d9bb}
\lim_{|\bs{y}|\rightarrow\infty}\alpha(|\bs{y}|^2)=0,\quad\lim_{|\bs{y}|\rightarrow\infty}\alpha_{\perp}(|\bs{y}|^2)=0,\quad\text{and}\quad\lim_{|\bs{y}|\rightarrow\infty}|\bs{y}|^2\beta(|\bs{y}|^2)=0.
\end{equation}
Combining the results above, we see that \bs{D} has the special form
\begin{equation}\label{d9b}
\bs{D}(\sqrt{2}\bs{y})=\alpha_{\perp}(|\bs{y}|^2)\bs{P}^{\perp}(\bs{y})+(\alpha(|\bs{y}|^2)-|\bs{y}|^2\beta(|\bs{y}|^2))\bs{P}(\bs{y}).
\end{equation}

Putting it all together we see that the \hyperlink{pairmatrix}{diffusion matrix}, \(\tilde{\bs{D}}(\bs{x})\),  of Equation~(\ref{gen20}) must have the form given in the following.
\begin{thm}[Structure of the \hyperlink{pairmatrix}{Diffusion Matrix}]\label{pairdiff}
The \hyperlink{pairmatrix}{diffusion matrix} \(\tilde{\bs{D}}(\bs{x})=\bs{C}-\bs{D}(\bs{x})\) has the form
\begin{equation}\label{d10}
\tilde{\bs{D}}(\bs{x})=\sigma_{\perp}(|\bs{x}|^2)\bs{P}^{\perp}(\bs{x})+\sigma(|\bs{x}|^2)\bs{P}(\bs{x}),
\end{equation}
where \(\sigma,\sigma_{\perp}:\mathbb{R}^+\rightarrow\mathbb{R}\) are given in terms of \(\alpha\), \(\alpha_{\perp}\) and \(\beta\) by
\begin{equation}\label{d10a}
\begin{split}
\sigma_{\perp}(|\bs{x}|^2)  &=c-\alpha_{\perp}\bigg(\frac{1}{2}|\bs{x}|^2\bigg)\\
\sigma(|\bs{x}|^2)          &=c-\alpha\bigg(\frac{1}{2}|\bs{x}|^2\bigg)+\bigg(\frac{1}{2}|\bs{x}|^2\bigg)\beta\bigg(\frac{1}{2}|\bs{x}|^2\bigg).
\end{split}
\end{equation}
We also have
\begin{equation}\label{d10aa}
0<\sigma_{\perp}(\xi^2)<c,\text{ for }\xi>0,\qquad\sigma_{\perp}(0)=0,\qquad\lim_{\xi\rightarrow\infty}\sigma_{\perp}(\xi^2)=c,
\end{equation}
and
\begin{equation}\label{d10aaa}
0<\sigma(\xi^2),\text{ for }\xi>0,\qquad\sigma(0)=0,\qquad\lim_{\xi\rightarrow\infty}\sigma(\xi^2)=c.
\end{equation}
\end{thm}
Note that \(\sigma(|\bs{x}|^2)\) and \(\sigma_{\perp}(|\bs{x}|^2)\) are the eigenvalues of \(\tilde{\bs{D}}(\bs{x})\). Unless they are equal, \(\sigma(|\bs{x}|^2)\) is a simple eigenvalue and
\(\sigma_{\perp}(|\bs{x}|^2)\) has geometric multiplicity \(d-1\). We refer to \(\sigma\) as the \textit{radial eigenvalue} and \(\sigma_{\perp}\) as the \textit{lateral eigenvalue.}

\begin{define}\label{d10a4}
We say the \hyperlink{pairmatrix}{diffusion matrix} \(\tilde{\bs{D}}\) is \hypertarget{raddom}{\emph{radially dominant}} whenever its eigenvalues satisfy:
\begin{equation}\label{d10aaaa}
\sigma_{\perp}(\xi^2)<\sigma(\xi^2),\quad\forall\xi>0.
\end{equation}
\end{define}
Radial dominance of \(\tilde{\bs{D}}\) will be shown later to guarantee that the difference-process is recurrent in two dimensions. The next result gives a sufficient condition for radial dominance
in terms of the \hyperlink{forcingfunction}{forcing function} \(\phi\).


\begin{thm}\label{d10a6}
If, in addition to the Hypothesis~\ref{g0}, the scalar-valued forcing function~\(\phi\) satisfies the logarithmic convexity condition \((\ln\phi)''\leq0\), the \hyperlink{pairmatrix}{diffusion
matrix} is \hyperlink{raddom}{radially dominant.}\footnote{This condition for guaranteing \hyperlink{raddom}{radial dominantce} is a bit strong. See Kotelenez~\cite{KO4} for a somewhat weaker
condition. The logarithmic convexity condition also guarantees that the functions \(\xi\mapsto\alpha(\xi^2)\) and  \(\xi\mapsto\alpha_{\perp}(\xi^2)\) are strictly decreasing.}
\end{thm}
This result is proved in in Appendix~\ref{detstruct}.

The general properties for \(\tilde{\bs{D}}(\bs{x})\) are all typified by the \hyperlink{maxker}{Maxwell kernel}, for which \(\tilde{\bs{D}}_{\varepsilon}(\bs{x})\) can be computed explicitly.

\begin{maxwell}\label{maxpair}For the \hyperlink{maxker}{Maxwell kernel}, \(\bs{g}_{\varepsilon}\),  given by Equation~(\ref{kernel1}), the \hyperlink{pairmatrix}{diffusion matrix} \(\tilde{\bs{D}}_{\varepsilon}(\bs{x})\), defined in Equation~(\ref{gen11}), has the explicit form:
\begin{equation}\label{gen26}
    \begin{split}
    \tilde{\bs{D}}_{\varepsilon}(\bs{x})    &=\bs{C}_{\varepsilon}-\bs{D}_{\varepsilon}(\bs{x})\\
    &=\sigma_{\varepsilon\perp}(|\bs{x}|^2)\bs{P}^{\perp}(\bs{x})+\sigma_{\varepsilon}(|\bs{x}|^2)\bs{P}(\bs{x})\\
    &=c_{\varepsilon,d}\bigg(1-e^{-\frac{|\bs{x}|^2}{2\varepsilon}}\bigg)\bs{P}^{\perp}(\bs{x})
    +c_{\varepsilon,d}\bigg(1-e^{-\frac{|\bs{x}|^2}{2\varepsilon}}+\frac{|\bs{x}|^2}{\varepsilon}e^{-\frac{|\bs{x}|^2}{2\varepsilon}}   \bigg)\bs{P}(\bs{x}),
    \end{split}
\end{equation}
where
\begin{equation}\label{gen27}
    c_{\varepsilon,d}=\frac{\varepsilon}{2}\kappa^2_{\varepsilon,d}(\pi\varepsilon)^{\frac{d}{2}}.
\end{equation}
In Equation~(\ref{gen26})
\begin{equation}\label{gen27a}
\sigma_{\varepsilon\perp}(|\bs{x}|^2)=c_{\varepsilon,d}\bigg(1-e^{-\frac{|\bs{x}|^2}{2\varepsilon}}\bigg)\quad\text{and}
\quad\sigma_{\varepsilon}(|\bs{x}|^2)=c_{\varepsilon,d}\bigg(1-e^{-\frac{|\bs{x}|^2}{2\varepsilon}}+\frac{|\bs{x}|^2}{\varepsilon}e^{-\frac{|\bs{x}|^2}{2\varepsilon}} \bigg).
\end{equation}
The \hyperlink{maxker}{Maxwell kernel} clearly induces a positive-definite, symmetric, radially dominant \hyperlink{pairmatrix}{diffusion matrix}, provided \(\bs{x}\not=\bs{0}\). Finally, if we
invoke the normalization of Comment~\ref{com-delcor}, we have
\begin{equation}\label{gen27b}
    c_{\varepsilon,d}=1.
\end{equation}
\end{maxwell}

We close this section with the observation that the difference-process, \(\bs{z}(\cdot)\), generated in Theorem~(\ref{gen13}) can be represented as the unique solution to a stochastic It\^{o}
differential equation.
\begin{thm}\label{gen14}
Let \(\sqrt{\tilde{\bs{D}}}\)  denote the unique positive definite symmetric square root of \(\tilde{\bs{D}}\) and let
\(\bs{\beta}\) be a process in \(\mathbb{R}^d\) whose components \(\beta_j,\,j=1,2,\ldots,d\) are \emph{i.i.d.}
one-dimensional standard Brownian motions.
Then the Markov-Feller difference-process of Theorem~(\ref{gen13}), \(\bs{z}(\cdot)\), generated by the transition probability function \((t,\bs{x},B)\mapsto\tilde{P}(t,\bs{x},B)\), can be
represented as the unique solution to the following It\^{o} stochastic initial value problem:
\begin{align}\label{gen15}
    d\bs{z}(t)  &=\sqrt{\tilde{\bs{D}}(\bs{z}(t))}\bs{\beta}(dt),\\
    \bs{z}(0)   &=\bs{x}.
\end{align}
Or, equivalently, the It\^{o} stochastic integral equation
\begin{equation}\label{gen18}
    \bs{z}(t)=\bs{x}+\int_0^t \sqrt{\tilde{\bs{D}}(\bs{z}(s))}\bs{\beta}(ds).
\end{equation}
By virtue of the structure Theorem~\ref{pairdiff}, the latter can be cast in the form
\begin{equation}\label{gen18b}
    \bs{z}(t)=\bs{x}+\int_0^t \sqrt{\sigma_{\perp}(|\bs{z}(s)|^2)}\bs{P}^{\perp}(\bs{z}(s))\bs{\beta}(ds)+\int_0^t \sqrt{\sigma(|\bs{z}(s)|^2)}\bs{P}(\bs{z}(s))\bs{\beta}(ds).
\end{equation}
\end{thm}

\section{Separation-Process}\label{sec-sep}

The matrix function \(\bs{x}\mapsto\tilde{\bs{D}}(\bs{x})\),  characterized in the previous section, determines the generator of the difference-process, which we have denoted by~\(\bs{x}(\cdot)\).
However, we are most concerned with the magnitude, \(|\bs{x}(\cdot)|\), of this process or, succinctly, the \hypertarget{sepproc}{\textit{separation-process}}. The following argument shows that the
separation-process is Markovian.

If $\bs{y}, \tilde{\bs{y}} \in {\mathbb{R}}^d$ and $|\bs{y}| = |\tilde{\bs{y}}|$ then there is an orthogonal matrix \bs{Q} such that $\tilde{\bs{y}} = \bs{Q}\bs{y}.$ Now set $ \bs{y} := \bs{r}^2_0-
\bs{r}^1_0, \ \ \tilde{\bs{y}}: = \tilde{\bs{r}}^2_0 - \tilde{\bs{r}}^1_0,$ so $\tilde{\bs{r}}^2_0 -  \tilde{\bs{r}}^1_0 = \bs{Q}\bs{r}^2_0 - \bs{Q}\bs{r}^1_0.$ Then, by Lemma~\ref{lem1} and
recalling that \(\sqrt{2}\bs{x}=\bs{r}^2-\bs{r}^1,\sqrt{2}\bs{x}_0=\bs{r}^2_0-\bs{r}^1_0\),
\begin{equation}\label{eq5.1}
    \begin{split}
    &\bs{x}(t,\bs{x}_0)\\
    &\sim \frac{\bs{r}(t,\bs{r}^2_0,w)-\bs{r}(t,\bs{r}^1_0,w)}{\sqrt{2}}\\
    &\sim \frac{\bs{r}(t,\bs{Q}\bs{r}^2_0,w)-\bs{r}(t,\bs{Q}\bs{r}^1_0,w)}{\sqrt{2}}\\
    &\sim \bs{x}(t,\tilde{\bs{x}}_0)
    \end{split}
\end{equation}
Let $B$ be a Borel subset of ${\mathbb{R}}$ and $\Psi(\cdot)$ denote the inverse image of the map $\bs{x} \longrightarrow |\bs{x}|$. Then Equation~(\ref{eq5.1})  and the preceding considerations
imply by Dynkin~\cite{DY} that

\begin{equation}\label{eq5.2}
\dot{P}(t,\xi,B)\big|_{\xi=|\bs{x}_0|} := \mathcal{P}\{\bs{x}(t,\bs{x}_0)\in \Psi(B)|\bs{x}(0,\bs{x}_0)=\bs{x}_0\}
\end{equation}
is a transition probability function for a Markov diffusion in ${\mathbb{R}}^+$.\\

To find its generator, denoted by \(\dot{A}\), we compute the action of the differential operator in Equation~(\ref{gen12}) on functions \(\tilde{f}\) of the form
\(\tilde{f}(\bs{x})=\varphi(|\bs{x}|)\). As a preliminary, compute the second derivative matrix \(\bs{\nabla}_{\bs{x}}\bs{\nabla}_{\bs{x}}^T\tilde{f}\) whose components are given by
\(\frac{\partial^2}{\partial x_i\partial
    x_j}\varphi(|\bs{x}|)\). For \(\bs{x}\not =\bs{0}\), this computation yields
\begin{equation}\label{mag1a}
\big(\bs{\nabla}_{\bs{x}}\bs{\nabla}_{\bs{x}}^T\tilde{f}\big)(\bs{x})=\frac{\varphi'(|\bs{x}|)}{|\bs{x}|}\bigg(\bs{1}-\frac{\bs{x}\bs{x}^T}{|\bs{x}|^2}
\bigg)+\varphi''(|\bs{x}|)\frac{\bs{x}\bs{x}^T}{|\bs{x}|^2}.
\end{equation}
Here it is useful to use the projections \(\bs{P}(\bs{x})\) and \(\bs{P}^{\perp}(\bs{x})\), defined in Equation~(\ref{proj}): In terms of these projections, Equation~(\ref{mag1a}) is cast in the
form
\begin{equation}\label{mag3a}
\big(\bs{\nabla}_{\bs{x}}\bs{\nabla}_{\bs{x}}^T\tilde{f}\big)(\bs{x})=\frac{\varphi'(|\bs{x}|)}{|\bs{x}|}\bs{P}^{\perp}(\bs{x})+\varphi''(|\bs{x}|)\bs{P}(\bs{x}).
\end{equation}
Using the diffusion matrix \(\tilde{\bs{D}}(\bs{x})\) in the form of Equation~(\ref{d10}) and Equation~(\ref{mag3a}),
the partial differential operator of Equation~(\ref{gen12}) applied to \(\tilde{f}(\bs{x}):=\varphi(|\bs{x}|)\) reduces to the ordinary differential operator\footnote{Recall that the inner product
of two square matrices, \(\bs{M}\) and \(\bs{N}\) was defined by \(\bs{M}\bullet\bs{N}:=trace(\bs{M}\bs{N}^T)\).}
\begin{equation}\label{mag5a}
\begin{split}
(\dot{A}\varphi)(|\bs{x}|)    &:=(\tilde{A}\tilde{f})(\bs{x})\\
                        &=\frac{1}{2}\bigg(\tilde{\bs{D}}(\bs{x})\bullet \big(\bs{\nabla}\bs{\nabla}^T\tilde{f}\big)(\bs{x})\bigg)\\
                        &=\frac{1}{2}\bigg(\frac{d-1}{|\bs{x}|}\sigma_{\perp}(|\bs{x}|^2)\varphi'(|\bs{x}|)
                        +\sigma(|\bs{x}|^2)\varphi''(|\bs{x}|)\bigg).
\end{split}
\end{equation}
In Section~\ref{sec-dir} we give an independent argument, using  Dirichlet forms, that this operator is indeed the generator for the separation-process.

\begin{maxwell}
For the special case of the \hyperlink{maxker}{Maxwell kernel} the operation of Equation~(\ref{mag5a}) yields
\begin{equation}\label{mag6}
    \begin{split}
     (\dot{A}\varphi)(\xi)
     &=\frac{c_{\varepsilon,d}}{2}\Bigg(\frac{(d-1)}{\xi}\bigg(1-e^{-\frac{\xi^2}{2\varepsilon}}\bigg)\varphi^{'}(\xi)\\
     &+\bigg(\big(1-e^{-\frac{\xi^2}{2\varepsilon}}\big)
     +\frac{\xi^2}{\varepsilon}e^{-\frac{\xi^2}{2\varepsilon}}  \bigg)\varphi^{''}(\xi)
     \Bigg).
    \end{split}
\end{equation}
Again, if we invoke the normalization of Comment~\ref{com-delcor}, we have
\begin{equation}\label{gen27b2}
    c_{\varepsilon,d}=1.
\end{equation}
\end{maxwell}

\begin{com}\label{sep1}
For large \(\xi\), the operator  \(\dot{A}\varphi\) in Equation~(\ref{mag5a}) (or Equation~(\ref{mag6})) behaves like
\begin{equation}\label{eq5.}
 \begin{array}{clcr}
 (\dot{A}\varphi)(\xi) \approx
\frac{c}{2}[\frac{(d-1)}{\xi}\varphi'(\xi)
                       +\varphi''(\xi)].
\end{array}
\end{equation}
Up to the constant multiplier,~\(c\), the right-hand side is precisely the generator of the Bessel process associated with a $d-$dimensional standard Brownian motion $\beta(\cdot)$ (See
Dynkin~\cite{DY} Ch. X, (10.87)). In other words, for large $\xi=|\bs{x}|$ the \hyperlink{sepproc}{separation-process} should behave like the Bessel process of the standard Brownian motion
$\beta(\cdot)$. Therefore, we say that the separation-process with the generator in Equation~(\ref{mag5a}) is a process of \emph{Bessel-type.}
\end{com}

For the general Bessel-type process with generator from Equation~(\ref{mag5a}), we obtain the following stochastic
ordinary differential equation for the magnitude process, \(\xi=|\bs{x}|\):\\

\begin{equation}\label{eq5.91}
 \begin{array}{clcr} {d\xi  = \frac{1}{2}\frac{(d-1)}{\xi}\sigma_{\perp}(\xi^2)dt
                       +\sqrt{\sigma(\xi^2)}\beta(dt),
                       }\\ \\
                       {\xi(0)= \xi_0,}
\end{array}
\end{equation}
where $\beta(\cdot)$ is a standard real-valued Brownian motion. Observe that the drift term in Equation~(\ref{eq5.91}) is proportional to \((d-1)\) and depends only on the lateral eigenvalue while
the diffusive (stochastic) term depends only on the radial eigenvalue.

We may now apply the criteria from stochastic analysis of one-dimensional diffusions to obtain a proof of the \emph{long-time behavior} of the \hyperlink{sepproc}{separation-process} for $d \geq 2$
and, therefore, of the attractive and repulsive behavior of the \hyperlink{sepproc}{separation-process} for large times
(See Proposition 5.9 and Theorem 5.13 of Kotelenez~\cite{KO4}).\\

\begin{thm}[Long-Term Behavior of the \hyperlink{sepproc}{Separation-Process}]\label{longterm}
\begin{enumerate}
\item For \(d =2\), the solution of Equation~(\ref{eq5.91}) is recurrent
whenever the \newline\hyperlink{pairmatrix}{diffusion matrix} \(\tilde{\bs{D}}\) is  \hyperlink{raddom}{radially dominant.}\footnote{A sufficient condition for radial dominance is that the forcing
function \(\phi\) satisfies the logarithmic convexity condition \((\ln\phi)''\leq0\). (See Theorem~\ref{d10a6}.)}
\item For \(d \geq 3\), the solution of Equation~(\ref{eq5.91}) is transient.
\end{enumerate}
\end{thm}
As a consequence, if $d = 2$, the two large particles will attract and repel each other infinitely often and, if $d\geq 3$, the distance between the two large particles will tend to $\infty$ almost surely as \(t\rightarrow\infty\).\\

\noindent\begin{proof} The following functional was used by Gikhman and Skorokhod~(~\cite{GI1} or~\cite{GI2}, Ch. 4, Section 16), to study the asymptotic behavior of solutions of one-dimensional
stochastic ordinary differential equations (See also Skorokhod~\cite{SK}, Ch. 1.3).  We will follow the notation of Ikeda and Watanabe~\cite{IK}, Ch. VI.3).\\

\begin{equation}\label{eq5.92}
{s(\xi,\zeta) := \int_{\zeta}^\xi\exp\bigg[-\int_{\zeta}^y (d-1)\frac{\sigma_{\perp}(z^2)} {\sigma(z^2)}\,\frac{dz}{z}\bigg]dy,}
\end{equation}\\
where $\zeta$ is an arbitrary point in ${\mathbb{R}}$. 
We have to show that \\

\begin{equation}\label{eq5.93}
    s(\xi,\zeta) \longrightarrow
    \begin{cases}
    \pm\infty   &\text{ as }\xi \longrightarrow {\pm}\infty \text{ if }d=2,\\
                \text{and}\\
    s^{\pm}     &\text{ as }\xi \longrightarrow {\pm}\infty \text{ if }d\geq3,
    \end{cases}
\end{equation}
where $-\infty < s^- \leq s^+ < \infty$. Then the conclusion then follows from Theorem 3.1 in Ikeda and Watanabe~\cite{IK}. The arguments that follow rely on the properties of the
\hyperlink{pairmatrix}{diffusion matrix} given in Theorem~(\ref{pairdiff}). In particular, we note that\footnote{The ratio in the second inequality of Equation~(\ref{rec1}) is \(<1\) in the case of
the \hyperlink{maxker}{Maxwell kernel}. This simplifies the proof of the first assertion for that case.}

\begin{equation}\label{rec1}
\lim_{z\rightarrow\infty}\frac{\sigma_{\perp}(z^2)}{\sigma(z^2)}=1\qquad\text{and}\qquad 0<\frac{\sigma_{\perp}(z^2)}{\sigma(z^2)},\quad\forall z>0.
\end{equation}
Furthermore, when the \hyperlink{pairmatrix}{diffusion matrix} \(\tilde{\bs{D}}\) is radially dominant we also have
\begin{equation}\label{rec1aa}
0<\frac{\sigma_{\perp}(z^2)}{\sigma(z^2)}<1,\,\forall z>0.
\end{equation}

\begin{enumerate}

\item For any numbers $\xi, \rho, \zeta$, Equation~(\ref{eq5.92}) implies
\begin{equation}\label{rec1a}
s(\xi,\zeta) = s(\rho,\zeta) + \exp\bigg[-\int_{\zeta}^{\rho}\frac{(d-1)\sigma_{\perp}(z^2)} {z\sigma(z^2)}\,dz\bigg]s(\xi,\rho).
\end{equation}
Thus we may choose appropriate values for $\rho$ and prove (\ref{eq5.93}) for $s(\xi,\rho)$ instead of $s(\xi,\zeta)$. It is easy to see that $s(\xi,\zeta)$ is positive whenever, \(\xi>\zeta\), and
\(\xi\mapsto s(\xi,\zeta)\) is increasing.  Furthermore, a straightforward computation gives
\begin{equation}\label{rec6}
s(-\xi,-\rho)=-s(\xi,\rho).
\end{equation}

\item Fix an arbitrary positive number \(\delta<\frac{1}{2}\). Then, by the limit in Equation~(\ref{rec1}) and Equation~(\ref{rec1aa}) there is a \(\rho>0\) such that
\begin{equation}\label{rec2}
(1-\delta)<\frac{\sigma_{\perp}(z^2)}{\sigma(z^2)}<(1+\delta),
\end{equation}
whenever \(|z|>\rho\). We can replace the upper bound \((1+\delta)\) in Equation~(\ref{rec2}) by \(1\)  whenever the \hyperlink{pairmatrix}{diffusion matrix} \(\tilde{\bs{D}}\) is radially dominant,
say when \((\ln\phi)''\leq 0\). In this case, if \(z>\rho>0\) we  have
\begin{equation}\label{rec3}
(d-1)\frac{1}{z}>\frac{(d-1)\sigma_{\perp}(z^2)}{z\sigma(z^2)}>(d-1)(1-\delta)\frac{1}{z}.
\end{equation}

\item Suppose $d = 2$ and $\xi > \rho > 0$. Then, assuming the convexity condition,
\begin{equation}\label{rec4}
    \exp\bigg[-(d-1)\int_{\rho}^{y}\frac{\sigma_{\perp}(z^2)} {z\sigma(z^2)}\,dz\bigg]
    >\exp\bigg[-\int_{\rho}^{y}\frac{dz}{z}\bigg]
    =\bigg(\frac{\rho}{y}\bigg).
\end{equation}
So,
\begin{equation}\label{rec5}
s(\xi,\zeta) > s(\rho,\zeta) +\zeta\ln\frac{\xi}{\rho}\rightarrow \infty\text{ as }\xi\rightarrow\infty.
\end{equation}


\item Suppose $d=2$ and $-\xi < -\rho < 0$. Then, using (\ref{rec6}), the previous argument yields
\begin{equation}\label{rec7}
\lim_{\xi\rightarrow\infty}s(-\xi,-\rho)=-\infty.
\end{equation}
This verifies (\ref{eq5.93}) when \(d=2\) and  \((\ln\phi)''\leq0\).

\item Next, suppose $\xi > \rho > 0$ and $d \geq 3$. Then

\begin{equation}\label{rec8}
    \begin{split}
    s(\xi,\rho) &=\int_{\rho}^{\xi}\exp\bigg[-(d-1)\int_{\rho}^{y}\frac{\sigma_{\perp}(z^2)} {z\sigma(z^2)}\,dz\bigg]\,dy\\
                &<\int_{\rho}^{\xi}\exp\bigg[-(d-1)(1-\delta)\int_{\rho}^{y}\frac{dz}{z}\bigg]\,dy\\
                &=\frac{\rho}{(d-1)(1-\delta)-1}\Bigg(1-\bigg(\frac{\rho}{\xi}\bigg)^{(d-1)(1-\delta)-1}  \Bigg)\\
                &<\frac{\rho}{(d-1)(1-\delta)-1},
    \end{split}
\end{equation}
since \((d-1)(1-\delta)-1>0\) whenever \(d\geq3\) and \(0\leq\delta<\frac{1}{2}\). Therefore, since \(\xi\mapsto s(\xi,\rho)\) is increasing we have

\begin{equation}\label{rec9}
\lim_{\xi\rightarrow\infty}s(\xi,\rho)=s^+<\infty,
\end{equation}
for some \(s^+\in\mathbb{R}\).

\item In view of (\ref{rec6}) we also get

\begin{equation}\label{rec10}
\lim_{\xi\rightarrow\infty}s(-\xi,-\rho)=s^->-\infty,
\end{equation}
where \(s^-=-s^+\). So (\ref{eq5.93}) is verified for \(d\geq3\).

\end{enumerate}
The statements about the difference of the two large particles are a simple consequence of the recurrence and transience properties
\end{proof}

\section{Depletion Effect --- Stochastic Model}\label{sec-clump}
\subsection{van~Kampen's Probability Flux}\label{ssec-vk}
 We generalize van~Kampen's notion of
\textit{probability flux rate} to vector processes \(\bs{x}(\cdot)\) in \(\mathbb{R}^d\) in order to characterize regions with a bias in favor of attraction or repulsion between particle
pairs.\footnote{The \emph{regions} in \(\mathbb{R}^d\) to which we refer in this section are regions in the space of the variable~\bs{x}, the normalized vector particle difference. Thus, if one
particle is viewed as the origin, the other is located at~\(\sqrt{2}\bs{x}\).} (See van~Kampen~\cite{KA}.) Let \(\bs{x}\not =\bs{0}\) be a point and \bs{v} be a unit vector in \(\mathbb{R}^d\). Such
a pair determines an oriented hyperplane in \(\mathbb{R}^d\) through \bs{x} with orienting normal \bs{v}. Let \(X(\bs{x},t)\) denote the probability density for the process \(\bs{x}(\cdot)\) at the
point \bs{x} and time \(t\). We define the \emph{probability flux rate vector (at \((\bs{x},t)\) with the orientation \bs{v})} by
\begin{equation}\label{vk1}
\bs{j}(\bs{x},t,\bs{v}):=-\frac{1}{2}X(\bs{x},t)\tilde{\bs{D}}(\bs{x})\bs{v}.
\end{equation}
The vector \( \bs{j}(\bs{x},t,\bs{v}) \) is the area density of the inststantaneous probability flow rate at time \(t\) through a surface element at \bs{x} oriented by \bs{v}. (See
Comment~\ref{com3} for a \(1-\)dimensional motivation.)  The \textit{flux rate  (at \((\bs{x},t)\) with orientation \bs{v})}, denoted by \(J(\bs{x},\bs{v},t)\),
is defined as the divergence of this flux, namely\footnote{\(\bs{\nabla_{x}}\) denotes the spatial gradient in \(\mathbb{R}^d\). The symmetry of \(\tilde{\bs{D}}\) is used in the second line of
Equation~(\ref{fd1}).}
\begin{equation}\label{fd1}
\begin{split}
J(\bs{x},t,\bs{v})  &:=\bs{\nabla_{x}}\bullet\bs{j}(\bs{x},t,\bs{v})=\bs{\nabla_{x}}\bullet\bigg(-\frac{1}{2}X(\bs{x},t)\tilde{\bs{D}}(\bs{x})\bs{v}    \bigg)\\
                    &=-\frac{1}{2}X(\bs{x},t)\bs{div}\tilde{\bs{D}}(\bs{x})\bullet\bs{v}-\frac{1}{2}\tilde{\bs{D}}(\bs{x})\bs{\nabla_{x}}X(\bs{x},t)\bullet\bs{v}.
\end{split}
\end{equation}
Note that \(J(\bs{x},t,\bs{v})\) is associated with an \emph{oriented} surface element at \bs{x}; in fact, the functional \(\bs{v}\mapsto J(\bs{x},t,\bs{v})\) is linear.

Henceforth, for simplicity, we assume that, at time \(t\), the probability density, \(\bs{y}\mapsto X(\bs{y},t)\), is uniform (constant) in a neighborhood of \bs{x}.  In this case, the flux rate
reduces to \footnote{For a square matrix valued function \(\bs{M}(\bs{y})\), \(\bs{div}\bs{M}(\bs{y})\) is defined through the identity
\(\bs{div}\bs{M}(\bs{y})\bullet\bs{a}=\bs{\nabla_{y}}\bullet\bs{M}^T(\bs{y})\bs{a}\). Recall that \(\tilde{\bs{D}}\) is symmetric.}
\begin{equation}\label{vk2}
J(\bs{x},t,\bs{v})=-\frac{1}{2}X(\bs{x},t)\bs{div}\tilde{\bs{D}}(\bs{x})\bullet\bs{v},
\end{equation}
the \bs{v} component of the vector \(-\frac{1}{2}X(\bs{x},t)\bs{div}\tilde{\bs{D}}(\bs{x})\). Thus, if the probability density \(X\) is assumed locally spatially uniform at a given time, the  flux
rate is essentially a multiple of the divergence of the \hyperlink{pairmatrix}{diffusion matrix}; that is, a multiple of the vector \(\bs{div}\tilde{\bs{D}}\). Roughly speaking, this assumption
means that, at time \(t\), the probability of finding a second large particle near \(\sqrt{2}\bs{x}\), given that the first is at the origin, is locally constant in \bs{x}. Since our computations
will be essentially local with respect to a given point \bs{x}, the assumption is reasonable, provided \(\bs{y}\mapsto X(\bs{y},t)\) at time \(t\) is reasonably regular.

In view of the special form of \(\tilde{\bs{D}}(\bs{x})\), given in Equation~(\ref{d10}), \( \bs{div}\tilde{\bs{D}}(\bs{x}) \) is radial (parallel to \bs{x}); specifically,\footnote{
\(\bs{div}\bs{P}(\bs{x})=\frac{d-1}{|\bs{x}|}\bs{u}(\bs{x})\). Recall the definition of \(\bs{u}(\bs{x})\) in Equation~(\ref{normvec}).}
\begin{equation}\label{fd3a}
    \begin{split}
    \bs{div}\tilde{\bs{D}}(\bs{x})   &=\bs{\nabla_{x}}\sigma_{\perp}(|\bs{x}|^2)+\bs{P}(\bs{x})(\bs{\nabla_{x}}\sigma(|\bs{x}|^2)-\bs{\nabla_{x}}\sigma_{\perp}(|\bs{x}|^2))\\
    &+(\sigma(|\bs{x}|^2)-\sigma_{\perp}(|\bs{x}|^2))\bs{div}\bs{P}(\bs{x})\\
    &=\bs{\nabla_{x}}\sigma(|\bs{x}|^2)+\big(\sigma(|\bs{x}|^2)-\sigma_{\perp}(|\bs{x}|^2)\big)\bs{div}\bs{P}(\bs{x})\\
    &=\bigg(2\sigma'(|\bs{x}|^2)|\bs{x}|+\frac{d-1}{|\bs{x}|}\big(\sigma(|\bs{x}|^2)-\sigma_{\perp}(|\bs{x}|^2)\big)\bigg)\bs{u}(\bs{x}).
    \end{split}
\end{equation}
For later convenience, define \(\xi\mapsto\psi(\xi)\) on \((0,\infty)\) by
\begin{equation}\label{vk3}
\psi(\xi):=\frac{d}{d\xi}\sigma(\xi^2)+\frac{d-1}{\xi}\big(\sigma(\xi^2)-\sigma_{\perp}(\xi^2)\big),
\end{equation}
so that \(\bs{div}\tilde{\bs{D}}(\bs{x})  =\psi(|\bs{x}|)\bs{u}(\bs{x})\).\footnote{For \(d\geq 2\) all spatial integrals involving \(\psi\) will converge at the origin.} Thus, when the probability
density \(X(\bs{x},t)\) is assumed to be uniform (constant) in a neighborhood of \bs{x} at time~\(t\), the  flux rate, \(J(\bs{x},t,\bs{v}) \), is completely determined by \(\psi(|\bs{x}|)\). We
will be concerned with the \emph{radially oriented flux rate,} \( J(\bs{x},t):= J(\bs{x},t,\bs{u}(\bs{x})) \), which we henceforth call the \textit{\(d-\)dimensional van-Kampen flux rate.} \(
J(\bs{x},t)\) takes the form\footnote{Provided the initial probability density \( X(\bs{x},t) \) is locally spatially constant.}
\begin{equation}\label{vk3a}
 J(\bs{x},t)=\frac{1}{2}\Big(\frac{d-1}{\xi}\sigma_{\perp}(\xi^2)-\frac{d-1}{\xi}\sigma(\xi^2)-\frac{d}{d\xi}\sigma(\xi^2)   \Big)X(\bs{x},t).
\end{equation}
Note that, except for the probability density multiplier \(X(\bs{x},t)\), \( J(\bs{x},t)\) depends only on the magnitude \(\xi\) of \bs{x}; that is, on the separation.

In the next sub-section we give a nice geometric interpretation of the \(d-\)dimensional van~Kampen flux rate, \( J(\bs{x},t)\).

\subsection{Interpretation of van~Kampen's Flux --- Pill-Box}
In this section suppose \(d\geq2\).\footnote{With some obvious modifications, the results of this section hold as well if \(d=1\).} Let  \(\mathbb{S}^{d-1}\) denote the (unit) sphere in
\(\mathbb{R}^d\) centered at the origin and let \(\xi\mathbb{S}^{d-1}\) denote the sphere in \(\mathbb{R}^d\) of radius \(\xi\) centered at the origin. The (dimensionless) surface area of
\(\mathbb{S}^{d-1}\) is denoted by \(\omega_{d-1}\) so the surface area of \(\xi\mathbb{S}^{d-1}\) is \(\xi^{d-1}\omega_{d-1}\).\footnote{\(\omega_{d-1}=2\pi^{\frac{d}{2}}\Gamma(\frac{d}{2})^{-1}\)
is the surface area of \(\mathbb{S}^{d-1}\).} Fix a point \bs{z} on \(\mathbb{S}^{d-1}\) and a \textit{cone angle} \(\varphi,\, 0\leq\varphi\leq\pi\). The spherical cap,
\(\varpi_{d-1}(\varphi,\bs{z})\), on \(\mathbb{S}^{d-1}\) consists of those points \(\bs{z}'\) on \(\mathbb{S}^{d-1}\) for which \(\bs{z}'\bullet\bs{z}\geq\cos\varphi\).\footnote{Thus, \(\varphi\)
is the azimuthal angle in a general spherical coordinate system for \(\mathbb{R}^d\) whose ``north pole" is determined by \bs{z}} The (dimensionless) area of this cap, which depends only on the cone
angle \(\varphi\), is denoted by \(\omega_{d-1}(\varphi)\).\footnote{ Hence, \(\omega_{d-1}(\pi)=\omega_{d-1}\) and \(\omega_{d-1}(\varphi)=\omega_{d-2}\int_0^{\varphi}\sin^{d-2}\theta\,d\theta\).}
The spherical cap on \(\xi\mathbb{S}^{d-1}\) subtended by \(\varpi_{d-1}(\varphi,\bs{z})\) is denoted by \(\xi\varpi_{d-1}(\varphi,\bs{z})\) and its surface area is
\(\xi^{d-1}\omega_{d-1}(\varphi)\).

Now fix a point \(\bs{x}\not=\bs{0}\). Recall that \(\bs{u}(\bs{x}):=\bs{x}/|\bs{x}|\), the associated outward unit vector (or point on \(\mathbb{S}^{d-1}\)), and set \(\xi:=|\bs{x}|\). For some
small positive~\(\lambda\), \(0<\lambda\ll 1\), consider the spherical shell in \(\mathbb{R}^d\) of thickness \(\lambda\) between the concentric spheres  \(\xi\mathbb{S}^{d-1}\) and
\((\xi+\lambda)\mathbb{S}^{d-1}\). Construct a small  ``pill-box" \(B(\bs{x},\varphi,\lambda)\) in \(\mathbb{R}^d\) by taking that portion of the shell subtended by the spherical cap
\(\varpi(\varphi,\bs{u}(\bs{x}))\).\footnote{The bounding surfaces of the pill-box are coordinate surfaces in general spherical coordinates for \(\mathbb{R}^d\).} The volume of the pill-box
\(B(\bs{x},\varphi,\lambda)\) is given by the expression
\begin{equation}\label{pb0}
\frac{\omega_{d-1}(\varphi)}{d}\big((\xi+\lambda)^d-\xi^d \big)\overset{0<\lambda\ll 1}{\approx}\omega_{d-1}(\varphi)\xi^{d-1}\lambda.
\end{equation}
The  boundary, \(\partial B(\bs{x},\varphi,\lambda)\), of the pill-box consists of the inner and outer ends, with areas
\begin{equation}\label{pb1}
\omega_{d-1}(\varphi)\cdot\xi^{d-1}\qquad \text{and}\qquad  \omega_{d-1}(\varphi)\cdot(\xi+\lambda)^{d-1} ,
\end{equation}
together with the lateral surface, whose area is
\begin{equation}\label{pb2}
\frac{\omega_{d-2}}{d-1}\sin^{d-1}\varphi\cdot\big((\xi+\lambda)^{d-1}-\xi^{d-1}\big)\overset{0<\lambda\ll 1}{\approx}\omega_{d-2}\sin^{d-1}\varphi\cdot\xi^{d-2}\lambda
\end{equation}
Refer to the Figure~(\ref{figpb}) below for a visualization.\footnote{Geometric terms such as \textit{surface area, volume, perimeter, unit outward normal vector, etc.} must be construed in the
\(d\)-dimensional context.}

\begin{figure}[H]
\begin{center}
\includegraphics[natheight=8.0in, natwidth=12.0in,height=3in, width=6in, keepaspectratio=true]
{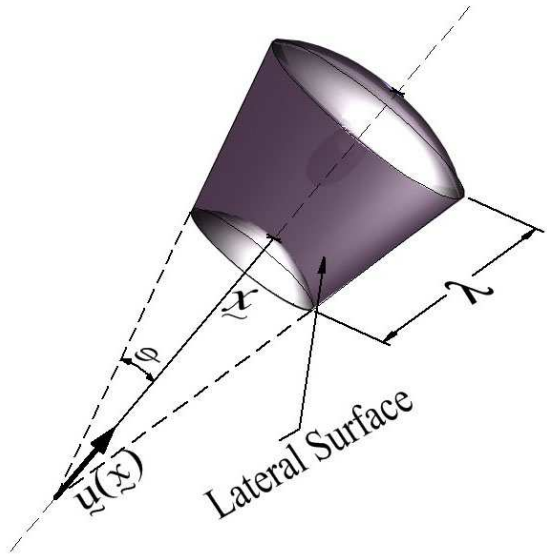}
\caption{The Pill Box $B(\bs{x},\varphi,\lambda)$ \label{figpb}}
\end{center}
\end{figure}

The \textit{net probability flux rate vector out of \(B(\bs{x},\varphi,\lambda)\) (at (\bs{x},t))} is the surface integral
\begin{equation}\label{vk4}
\int_{\partial B}\bs{j}(\bs{y},t,\bs{n}(\bs{y}))\,dA(\bs{y})=\int_{\partial B}-\frac{1}{2}X(\bs{x},t)\tilde{\bs{D}}(\bs{y})\bs{n}(\bs{y})\,dA(\bs{y}),
\end{equation}
where \(\bs{n}(\bs{y})\) always denotes the \emph{unit outward normal vector} to \(\partial B(\bs{x},\varphi,\lambda)\) at \bs{y}. Using the form of \(\tilde{\bs{D}}(\bs{y})\) in
Equation~(\ref{d10}), we can compute this flux explicitly. Here is the computation:\footnote{On the ends \(\bs{P}^{\perp}\bs{n}=\bs{0},\bs{P}\bs{n}=\bs{n}\). On the lateral surface
\(\bs{P}^{\perp}\bs{n}=\bs{n},\bs{P}\bs{n}=\bs{0}\).}
\begin{equation}\label{vk5}
    \begin{split}
    &\int_{\partial B}\bs{j}(\bs{y},t,\bs{n}(\bs{y}))\,dA(\bs{y})\\
    &=\int_{\partial B}-\frac{1}{2}X(\bs{x},t)\tilde{\bs{D}}(\bs{y})\bs{n}(\bs{y})\,dA(\bs{y})\\
    &=-\frac{1}{2}X(\bs{x},t)\Bigg\{\int_{\text{ends}}\sigma(|\bs{y}|^2)\bs{n}(\bs{y})\,dA(\bs{y})
    +\int_{\text{lateral}}\sigma_{\perp}(|\bs{y}|^2)\bs{n}(\bs{y})\,dA(\bs{y}) \Bigg\}\\
    &=-\frac{1}{2}X(\bs{x},t)\Bigg\{\omega_{d-2}\frac{\sin^{d-1}\varphi}{d-1}(\xi+\lambda)^{d-1}\sigma((\xi+\lambda)^2)-\omega_{d-2}\frac{\sin^{d-1}\varphi}{d-1}\xi^{d-1}\sigma(\xi^2)\\
    &-\omega_{d-2}\big(\sin^{d-1}\varphi\big)\int_{\xi}^{\xi+\lambda}\sigma_{\perp}(\zeta^2)\zeta^{d-2}\,d\zeta\Bigg\}\bs{u}(\bs{x})\\
    &=-\frac{1}{2}X(\bs{x},t)\omega_{d-2}\frac{\sin^{d-1}\varphi}{d-1}
    \Bigg\{(\xi+\lambda)^{d-1}\sigma((\xi+\lambda)^2)-\xi^{d-1}\sigma(\xi^2)\\
    &-(d-1)\bigg(\int_{\xi}^{\xi+\lambda}\sigma_{\perp}(\zeta^2)\zeta^{d-2}\,d\zeta\bigg)\Bigg\}\bs{u}(\bs{x})\\
    &\overset{0<\lambda\ll 1}{\approx}-\frac{1}{2}X(\bs{x},t)\omega_{d-2}\frac{\sin^{d-1}\varphi}{d-1}\xi^{d-1}\lambda\psi(\xi)\bs{u}(\bs{x}),
    \end{split}
\end{equation}
where \(\psi\) is given in Equation~(\ref{vk3})).\footnote{Recall that the probability density \(\bs{y}\mapsto X(\bs{y},t)\) is assumed to be constant in a neighborhood of \bs{x} for each fixed
\(t\).}

There are two important observations to make regarding the computation in Equation~(\ref{vk5}). First, the contribution to this net flux rate vector through the ends depends only upon the radial
eigenvalue \(\sigma\) while that through the lateral surface depends only upon the lateral eigenvalue  \(\sigma_{\perp}\); moreover, the lateral contribution is always positive. Second, the net flux
rate vector out of the pill-box is always radial in the sense that it is a multiple of the outward unit radial vector \(\bs{u}(\bs{x})\); thus,
\begin{equation}\label{vk6}
\int_{\partial B}\bs{j}(\bs{y},t,\bs{n}(\bs{y}))\,dA(\bs{y})=-\frac{1}{2}X(\bs{x},t)\mu(\xi,t,\varphi,\lambda)\bs{u}(\bs{x}),
\end{equation}
where \(\mu=\mu(\xi,t,\varphi,\lambda)\) is a scalar multiplier, which can be computed from the penultimate expression in Equation~(\ref{vk5}). The Divergence Theorem implies that
\begin{equation}\label{vk7}
-\frac{1}{2}X(\bs{x},t)\int_B \bs{div}\tilde{\bs{D}}(\bs{y})\,dV(\bs{y})=-\frac{1}{2}X(\bs{x},t)\int_{\partial B}\tilde{\bs{D}}(\bs{y})\bs{n}(\bs{y})\,dA(\bs{y}).
\end{equation}
Hence, using Equation~(\ref{vk4}),
\begin{equation}\label{vk8}
-\frac{1}{2}X(\bs{x},t)\int_B \bs{div}\tilde{\bs{D}}(\bs{y})\,dV(\bs{y})=-\frac{1}{2}X(\bs{x},t)\mu(|\bs{x}|,t,\varphi,\lambda)\bs{u}(\bs{x}).
\end{equation}

If the scalar multiplier \(\mu\) in Equation~(\ref{vk6}) is \emph{positive,} we can say that net probability flux rate vector out of \(B(\bs{x},\varphi,\lambda)\) is a radial vector pointing
\emph{inward} toward the origin. We interpret this to mean that, if the probability distribution at time \(t\) is locally spatially uniform, there is an instantaneous net statistical tendency for
points \(\bs{x}\) in \(B(\bs{x},\varphi,\lambda)\) to leave the pill-box and move radially inward toward the origin; in this sense, for \(B(\bs{x},\varphi,\lambda)\), there is a net statistical
tendency for \(|\bs{x}|\) to decrease and, hence, pairs of large particles have a net statistical tendency to move closer together. In Section~\ref{sec-dir} we we give an alternative context in
which this interpretation is reenforced.
\begin{define}
If the scalar multiplier \(\mu\) in Equations~(\ref{vk6},\ref{vk8}) is positive, we say that the region \(B(\bs{x},\varphi,\lambda)\) is \emph{attraction-biased} or \emph{attractive}. Similarly, if
\(\mu\) in Equations~(\ref{vk6},\ref{vk8}) is negative, we say that the region \(B(\bs{x},\varphi,\lambda)\) is \emph{repulsion-biased} or \emph{repulsive}. If \(\mu\) is zero we say the region
\(B(\bs{x},\varphi,\lambda)\) is \emph{neutral}.
\end{define}

By Equations~(\ref{fd3a}-\ref{vk3a}), we can decide whether \(B(\bs{x},\varphi,\lambda)\) is repulsive or attractive (or neutral) by assessing the algebraic sign of \(\psi\) in Equation~(\ref{vk3}).
Specifically, if \(\lambda\) is small, the volume of  \(B(\bs{x},\varphi,\lambda)\) is small, so \(B(\bs{x},\varphi,\lambda)\) is attractive if \(\psi(|\bs{x}|)>0\) and repulsive if
\(\psi(|\bs{x}|)<0\). Indeed, evaluating the left side of Equation~(\ref{vk8}) directly yields\footnote{Observe that
\begin{equation*}\label{vk9a}
\psi(\xi)\xi^{d-1}=\frac{\partial}{\partial\xi}\Big(\xi^{d-1}\sigma(\xi^2)\Big)-\frac{d-1}{\xi}\xi^{d-1}\Big(\sigma_{\perp}(\xi^2)\Big),
\end{equation*}
which shows that the expressions in Equations~(\ref{vk5}) and~(\ref{vk9}) are the same and, hence, independently verifies the Divergence Theorem for the pill-box. }
\begin{multline}\label{vk9}
-\frac{1}{2}X(\bs{x},t)\int_B \bs{div}\tilde{\bs{D}}(\bs{y})\,dV(\bs{y})\\
    =-\frac{1}{2}X(\bs{x},t)\omega_{d-2}\frac{\sin^{d-1}\varphi}{d-1}\bigg\{\int_{|\bs{x}|}^{|\bs{x}|+\lambda}\psi(\zeta)\zeta^{d-1}\,d\zeta\bigg\}\bs{u}(\bs{x})\\
    \overset{0<\lambda\ll 1}{\approx}-\frac{1}{2}X(\bs{x},t)\omega_{d-2}\frac{\sin^{d-1}\varphi}{d-1}|\bs{x}|^{d-1}\lambda\psi(|\bs{x}|)\bs{u}(\bs{x}).
\end{multline}

Based on this discussion we extend our definitions to more general regions in \(\mathbb{R}^d\). By a \textit{region} we here mean a closed and bounded set in \(\mathbb{R}^d\) so
regular that the Divergence Theorem applies.
\begin{define}
A region in \(\mathbb{R}^d\), not containing the origin, is \emph{attraction-biased} or \emph{attractive} whenever the scalar function \(\psi\) in Equation~(\ref{vk3}) is positive inside the region,
\emph{repulsion-biased} or \emph{repulsive} whenever \(\psi\) is negative inside the region, or \emph{neutral} whenever  \(\psi\) is zero inside the region. By extension, we say that the
\emph{point} \(\bs{x}\not=\bs{0}\) is \emph{attractive, repulsive,} or \emph{neutral} whenever \(\psi(|\bs{x}|)\) is positive, negative, or zero.\footnote{Since all integrals over regions including
the origin converge, we could extend the definition to such regions by a straightforward limiting process. Bear in mind that the regions to which we refer here are regions in the \emph{difference
space} of the variable~\bs{x}.}
\end{define}

\begin{com}\label{com-probdens}
If the probability density \(\bs{y}\mapsto X(\bs{y},t)\) is not constant in a neighborhood of \bs{x}, but \(|\bs{\nabla_{x}}X(\bs{x},t)|\) is small, our conclusions are not much affected. In this
case, the flux rate \(J(\bs{x},t)\) in Equation~(\ref{vk3a}) will have the small additional term \(-\frac{1}{2}\sigma(|\bs{x}|^2)\big(\bs{u}(\bs{x})\bullet\bs{\nabla_{x}}X(\bs{x},t)\big)\).
\end{com}

\begin{com}\label{com-geom}
A geometric observation is in order here. The outward radial component of the net probability flux rate vector out of \(B(\bs{x},\varphi,\lambda)\), defined through Equation~(\ref{vk4}),  is not the
same as the flux out of \(B(\bs{x},\varphi,\lambda)\) of the outward radial component of  \(\bs{y}\mapsto \bs{j}(\bs{y},t,\bs{u}(\bs{y}))\). The latter depends only upon the radial eigenvalue,
\(\sigma\), while the former depends on both \(\sigma\) and the lateral eigenvalue,~\(\sigma_{\perp}\). Indeed, they differ by
\(\frac{1}{2}X(\bs{x},t)\int_B\frac{d-1}{|\bs{y}|}\sigma_{\perp}(|\bs{y}|^2)\,d\bs{y}\), which is a consequence of the identity
\[(\bs{div}\tilde{\bs{D}}(\bs{x}))\bullet\bs{u}(\bs{x})-div(\tilde{\bs{D}}(\bs{x})\bs{u}(\bs{x}))=-\frac{1}{|\bs{x}|}\tilde{\bs{D}}(\bs{x})\bullet\bs{P}^{\perp}(\bs{x})=-\frac{d-1}{|\bs{x}|}\sigma_{\perp}(|\bs{x}|^2).\]
Thus our van~Kampen flux rate depends on the lateral as well as the radial effects induced by \(\tilde{\bs{D}}\).
\end{com}

\begin{com}
For \(\bs{v}=\bs{u}(\bs{x})\) we see that the algebraic sign of the van~Kampen flux \(J(\bs{x},t)\) is opposite that of \(\psi(|\bs{x}|)\). The algebraic sign of the radially oriented van~Kampen
flux rate density determines whether a point is attractive, repulsive, or neutral. Of course, the same is true for any outward \bs{v}; that is, for any \bs{v} such that
\(\bs{v}\bullet\bs{u}(\bs{x})>0\).
\end{com}

\begin{com}
Our notion of attraction/repulsion-bias is \emph{instantaneous} in the sense that it is expected to be valid over a short time interval after the initial time \(t\). For if there is an attraction-
or repulsion-bias at some time \(t\), we do not expect the distribution \(X(\bs{x},t)\), initially assumed uniform, to remain so for long.
\end{com}

\begin{com}\label{com3}
It is instructive to relate our definition of \emph{probability flux rate vector} in Equation~(\ref{vk1}) with the one-dimensional case. Let \(B\in\mathbb{R}\) denote the interval
\(B:=[x,x+\lambda]\), where for definiteness we take \(0<x<x+\lambda\). Again, \(X(x,t)\) denotes the probability density. The probability transfer rate from \(x\text { to } x+\lambda\)
(left-to-right) is \(\frac{1}{2}X(x,t)\tilde{D}(x)\) while the probability transfer rate from \(x+\lambda\text { to } x\) (right-to-left) is \(-\frac{1}{2}X(x+\lambda,t)\tilde{D}(x+\lambda)\). The
net probability transfer rate between \(x\) and \(x+\lambda\)  is therefore \(\frac{1}{2}X(x,t)\tilde{D}(x)-\frac{1}{2}X(x+\lambda,t)\tilde{D}(x+\lambda,t)\). If the the initial probability density
is locally spatially uniform, then,
\begin{equation}\label{vkc1}
\begin{split}
\frac{1}{2}X(x,t)\tilde{D}(x)-\frac{1}{2}X(x+\lambda,t)\tilde{D}(x+\lambda)
                        &=\frac{1}{2}X(x,t)[\tilde{D}(x)-\tilde{D}(x+\lambda)]\\
                        &=-\frac{1}{2} X(x,t)\lambda\frac{\tilde{D}(x+\lambda)-\tilde{D}(x)}{\lambda}\\
                        &\approx -\frac{1}{2} X(x,t)\lambda\tilde{D}'(x).
\end{split}
\end{equation}
In particular, if \(\tilde{D}'(x)>0\) and  \(0<\lambda\ll 1\), the net probability transfer rate between the endpoints of the small interval, \(B\), is negative; that is, toward the origin.

Now reproduce our pill-box discussion for \(d=1\); that is, in the vector space~\(\mathbb{R}^1\). The one-dimensional version of Equation~(\ref{vk1}) is \(j(x,t,v)=-\frac{1}{2}X(x,t)\tilde{D}(x)v\).
Here, \(v\) is construed as a one-dimensional unit vector in~\(\mathbb{R}^1\), namely \(v=\pm1\). In this context, the unit \emph{outward} normal vector to \(B\) at \(x\) is \(n(x)=-1\) and the unit
\emph{outward} normal vector to \(B\) at \(x+\lambda\) is \(n(x+\lambda)=+1\). Therefore, the net probability efflux rate vector out of \(B\) is (assuming that the initial probability density is
locally spatially uniform)
\begin{multline}\label{vkc2}
j(x,t,n(x))+ j(x+\lambda,t,n(x+\lambda))    =\\
                                            \bigg(-\frac{1}{2}X(x,t)\tilde{D}(x)n(x)\bigg)
                                            +\bigg(-\frac{1}{2}X(x+\lambda,t)\tilde{D}(x+\lambda)n(x+\lambda)\bigg)\\
                                            =\bigg(-\frac{1}{2}X(x,t)\tilde{D}(x)(-1)\bigg)
                                            +\bigg(-\frac{1}{2}X(x,t)\tilde{D}(x+\lambda)(+1)\bigg)\\
                                            =-\frac{1}{2}X(x,t)(\tilde{D}(x+\lambda)-\tilde{D}(x))\\
                                            \approx -\frac{1}{2} X(x,t)\lambda\tilde{D}'(x).
\end{multline}
Thus, \emph{for \(d=1\), the net probability flux rate vector out of \(B\) is precisely the net probability transfer rate between the end-points of~\(B\).} If \(\tilde{D}'(x)>0\) and \(0<\lambda\ll
1\), this vector points toward the origin. Note that \(-\frac{1}{2} X(x,t)\lambda\tilde{D}'(x)\) is approximately the integral of the van~Kampen flux: \(\int_B J(y,t,+1)\,dy\).

Of course, when \(d=1\) there is no lateral effect to consider. When \(d\geq2\) the lateral effect is significant.
\end{com}

\begin{maxwell}\label{com-vkflux}For the \hyperlink{maxker}{Maxwell kernel} \(\bs{g}_{\varepsilon}\),  given by Equation~(\ref{kernel1}), the
divergence of the diffusion coefficient matrix \(\tilde{\bs{D}}_{\varepsilon}(\bs{x})\) in Equation~(\ref{gen11}) is the vector valued function
\begin{equation}\label{flux10}
    \bs{div}\tilde{\bs{D}}_{\varepsilon}(\bs{x})=\psi(|\bs{x}|)\bs{u}(\bs{x})=\frac{c_{\varepsilon,d}}{\varepsilon^2}
    e^{-\frac{|\bs{x}|^2}{2\varepsilon}}|\bs{x}|\big((2+d)\varepsilon-|\bs{x}|^2\big)
    \bs{u}(\bs{x}).
\end{equation}

The van~Kampen flux when the \hyperlink{kernel}{forcing kernel}, \(\bs{g}_{\varepsilon}\), is derived from a Maxwellian velocity field for the small particles,  is immediate from
Equation~(\ref{flux10}):
\begin{equation}\label{flux3a}
J(\bs{x},t,\bs{v}) =\bigg(-\frac{1}{2}X(\bs{x},t)\frac{c_{\varepsilon,d}}{2\varepsilon^2}e^{-\frac{|\bs{x}|^2}{2\varepsilon}}
    |\bs{x}|\big((2+d)\varepsilon-|\bs{x}|^2\big)\bigg)\big(\bs{u}(\bs{x})\bullet\bs{v}\big)
\end{equation}
and, hence,
\begin{equation}\label{flux3}
J(\bs{x},t) =-\frac{1}{2}X(\bs{x},t)\frac{c_{\varepsilon,d}}{2\varepsilon^2}e^{-\frac{|\bs{x}|^2}{2\varepsilon}}
    |\bs{x}|\big((2+d)\varepsilon-|\bs{x}|^2\big)
\end{equation}

We conclude that, for  \(\bs{g}_{\varepsilon}\), the region \(\{\bs{x}\in\mathbb{R}^d:|\bs{x}| < \sqrt{(2+d)\varepsilon
}\}\)~is attraction-biased, and the region \(\{\bs{x}\in\mathbb{R}^d:|\bs{x}|
> \sqrt{(2+d)\varepsilon }\}\)   is repulsion-biased. The points of neutral bias have measure zero in \(\mathbb{R}^d\).
Note that the radius of the attraction-biased region is proportional to the correlation length \(\sqrt{\varepsilon}\) and increases with dimension. We should expect a similar result for any
unimodular distribution. Figure~(\ref{f3-flux1}) is a graph of the normalized van Kampen flux \(J\) as a function of the separation \(\xi=|\bs{x}|\) for the Maxwell kernel.

Observe that if we invoke the normalization of Comment~\ref{com-delcor}, so that \(\bs{x}\mapsto\frac{1}{d} |\bs{g}_{\varepsilon}(\bs{x})|^2\) is a probability density on \(\mathbb{R}^d\), then
\(\frac{1}{\sqrt{2}}\sqrt{(2+d)\varepsilon }\) is the standard deviation of the corresponding probability distribution for \bs{x}. In view of the fact that
\(\bs{x}=\frac{1}{\sqrt{2}}(\bs{r}^2-\bs{r}^1)\), we see that \emph{the separation that determines the region of attractive bias is precisely the standard deviation of the corresponding distribution
for the difference \(\bs{r}^2-\bs{r}^1\).}

In view of Comment~\ref{com-diffnorm}, this conclusion will not change if another normalization is used, provided that the probability density \(\bs{x}\mapsto\frac{1}{d}
|\bs{g}_{\varepsilon}(\bs{x})|^2\) is replaced by the appropriate probability density, say \(\bs{x}\mapsto\frac{2D}{d} |\bs{g}_{\varepsilon}(\bs{x})|^2\).

The curve in Figure~(\ref{f3-flux1}) measures the statistical tendency toward clustering with separation. Compare it with the curves in Figure~(\ref{LaserTweezerLayout} B) that measure the potential
between a large particle and a substrate with separation. Qualitatively they are similar. We expect that any assumed velocity distribution of the small particles that is qualitatively similar to the
Maxwell distribution will tell the same story.

\end{maxwell}

\begin{figure} [H]
\begin{center}
\includegraphics[height=3in, width=5in]{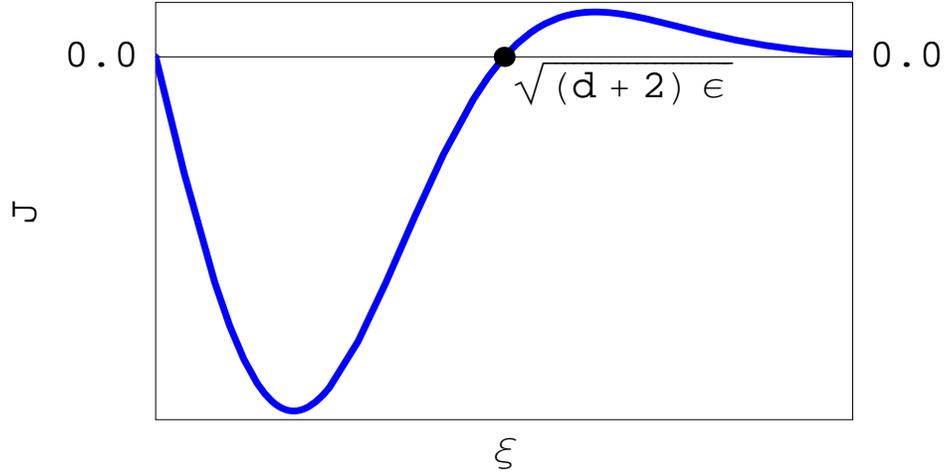}
\caption[van~Kampen Flux]{The (normalized) van~Kampen flux , $J(\bs{x},t)$, as a function of the separation distance, $\xi=|\bs{x}|$.}
\label{f3-flux1} %
\end{center}
\end{figure}

\begin{com}
Starting from a uniform distribution, at least for short times, we expect that the distance between the two particles in attraction-biased regions would tend to decrease and in repulsion-biased
regions this distance would tend to increase. In particular, for particles sufficiently close together, the attraction bias would be consistent with the depletion phenomena (clustering) observed in
colloids as described in Section~\ref{sec-jam}. However, the long-time behavior may be different. Indeed, a proof of the following result can be found in Skorohod~\cite{SK}, Ch. 1.3: Assume \(d=1\)
and let \(z(t,b)\) be the unique solution of Equation~(\ref{gen18}) with \(z(0) = b\neq 0\). Then,
\begin{equation}\label{p1}
    \mathcal{P}[ \lim_{t\rightarrow \infty} z(t,b) = \bs{0}] = 1.
\end{equation}
This means that, on the average, the distance between two particles on the real line will eventually tend to zero. In other words, the whole positive and negative real lines are attractive regions;
any repulsion-biased regions will have no long-term effect on the distances between the two particles. We believe the difference between long-time behavior and short-time behavior in repulsive
regions for dimension \(d=1\) to be the result of the recurrence of one-dimensional Brownian motions. 
\end{com}
\begin{com}
In Theorem~\ref{longterm} we showed that the long term behavior of the difference \(\bs{r}^2-\bs{r}^1\) is recurrent for \(d=2\), provided the diffusion matrix is \hyperlink{raddom}{radially
dominant}, and transient for \(d>2\) (See Kotelenez~\cite{KO4}). Our discussion, based as it is on the initial uniformity of the probability distribution \(X(\bs{x},t)\), is necessarily a short-time
result. The connection between the short- and long-time behavior remains an open and interesting problem, which will be investigated in future research.
\end{com}

\subsection{A Tale of Two Fluxes\label{sec-dir}} In Section~\ref{sec-sep}, the generator, \(\dot{A}\), of the \hyperlink{sepproc}{separation-process}, \(\xi=|\bs{x}|\), was derived.
If \(\dot{A}\) of Equation~(\ref{mag5a}) is used \emph{directly} in \((0,\infty)\) to compute the \(1-\)dimensional van~Kampen flux rate for the process~\(\xi\), the result is
\begin{equation}\label{dir0}
\dot{J}(\bs{x},t)=\frac{1}{2}\Big(\frac{d-1}{\xi}\sigma_{\perp}(\xi^2)-\frac{d}{d\xi}\sigma(\xi^2)\Big)X(\bs{x},t),
\end{equation}
assuming the initial probability density, \(X(\bs{x},t)\), to be locally spatially uniform.\footnote{For this computation we suppose that the probability density \(X\) depends on \bs{x} through
\(\xi=|\bs{x}|\). In view of the assumption, this is no restriction.} van~Kampen's procedure for this is provided in Appendix~\ref{app-b2}. This flux rate, of course, is different than our
\(d\)-dimensional van-Kampen flux rate, given in Equation~(\ref{vk3a}). The difference between them is that our flux has an additional \emph{radial} term, \(-\frac{d-1}{\xi}\sigma(\xi^2)\), which is
significant. Indeed, for the Maxwell kernel, our version predicts an attractive region for every value of \(d\geq 1\), which region increases with \(d\). On the other hand, for the same Maxwell
kernel, the van~Kampen flux of Equation~(\ref{dir0}) predicts such an attractive region only for \(d=1,2,3\), which region decreases with \(d\).\footnote{See Comment~\ref{com-vkflux2}.} This begs
the questions: Why are the two fluxes so different?\footnote{Of course they do agree when \(d=1\).} And, which one provides an appropriate means to measure a tendency toward clustering?

To address the relationship between the two fluxes we consider the Dirichlet quadratic form associated with the generator of the Markov-Feller process given in Equation~(\ref{gen12}).\footnote{For a
discussion of Dirichlet forms see Fukushima~\cite{FU} or Ma and R\"{o}ckner~\cite{MA}.} For suitable scalar-valued functions \(f,\,g\) defined in \(\mathbb{R}^d\) this quadratic form is\footnote{We
follow the convention using \( -\tilde{A} \) in the quadratic form. Also, recall that \(\bs{\nabla}_{\bs{x}}\bs{\nabla}_{\bs{x}}^Tf \) denotes the square, symmetric matrix whose components are given
by \(\frac{\partial^2}{\partial x_i\partial x_j}f(\bs{x})\).}
\begin{equation}\label{dir1}
\breve{\mathcal{E}}(f,g):=\left<-\tilde{A}f,g  \right>=-\int_{\mathbb{R}^d}\frac{1}{2}\Big(\tilde{\bs{D}}(\bs{x})\bullet\bs{\nabla}_{\bs{x}}\bs{\nabla}_{\bs{x}}^Tf(\bs{x})\Big)g(\bs{x})\,d\bs{x}
\end{equation}
The \emph{breve} (\(\,\breve{}\,  \)) over the \(\mathcal{E}\) here signifies that the underlying equilibrium measure for the quadratic form is Lebesgue measure, \(d\bs{x}\), in \(\mathbb{R}^d\). If
\(f\) and \(g\) are sufficiently regular to ensure that the boundary terms associated with the Divergence Theorem vanish, the quadratic form in Equation~(\ref{dir1}) can be rewritten
as\footnote{Surely this is the case if \(f\) and \(g\) have compact support.}
\begin{equation}\label{dir1a}
\begin{split}
\breve{\mathcal{E}}(f,g)=\left<-\tilde{A}f,g  \right>   &=\int_{\mathbb{R}^d}\frac{1}{2}\tilde{\bs{D}}(\bs{x})\bs{\nabla}_{\bs{x}}f(\bs{x})\bullet\bs{\nabla}_{\bs{x}} g(\bs{x})   \,d\bs{x}\\
                            &+\int_{\mathbb{R}^d}\frac{1}{2}\bs{div}\tilde{\bs{D}}(\bs{x})\bullet\bs{\nabla}_{\bs{x}}f(\bs{x}) g(\bs{x})   \,d\bs{x}.\\
\end{split}
\end{equation}

Next we compute the quadratic form, \(\breve{\mathcal{E}}\), of Equation~(\ref{dir1}) or~(\ref{dir1a}) for radially dependent functions:
\(\tilde{f}(\bs{x}):=\varphi(|\bs{x}|),\,\tilde{g}(\bs{x}):=\gamma(|\bs{x}|)\), as we did in Section~\ref{sec-sep}. If we suppose that the functions \(\varphi\) and \(\gamma\) are sufficiently
regular on \(\mathbb{R}^+\),
Equation~(\ref{dir1a}) becomes\footnote{We suppose that \(\varphi,\,\gamma\) are smooth and have compact support in \((0,\infty)\). 
Recall that, for \(\bs{x}\not=\bs{0}\), \(\bs{u}(\bs{x})=\frac{\bs{x}}{|\bs{x}|}\), \(\bs{\nabla}_{\bs{x}}\varphi(|\bs{x}|) =\varphi'(|\bs{x}|)\bs{u}(\bs{x})  \), etc.}
\begin{equation}\label{dir2}
\begin{split}
\breve{\mathcal{E}}(\tilde{f},\tilde{g})=\left<-\tilde{A}\tilde{f},\tilde{g} \right>
                            &=\frac{1}{2}\int_{\mathbb{R}^d}\tilde{\bs{D}}(\bs{x})\bs{u}(\bs{x})\bullet\bs{u}(\bs{x})\varphi'(|\bs{x}|)\gamma'(|\bs{x}|)\,d\bs{x}\\
                            &+\frac{1}{2}\int_{\mathbb{R}^d}\bs{div}\tilde{\bs{D}}(\bs{x})\bullet\bs{u}(\bs{x})\varphi'(|\bs{x}|)\gamma(|\bs{x}|)\,d\bs{x}.
\end{split}
\end{equation}
Observe that the first term in Equation~(\ref{dir1a} or \ref{dir2}) itself determines a Dirichlet quadratic form since it is positive definite and, more important, it is symmetric with respect to
Lebesgue measure, \(d\bs{x}\), in \(\mathbb{R}^d\). The second is neither symmetric with respect to Lebesgue measure, \(d\bs{x}\), in \(\mathbb{R}^d\) nor definite.\footnote{The second quadratic
form, while not symmetric, is not anti-symmetric in our context; in fact, it is anti-symmetric if and only if \(\psi(\xi)=\frac{k}{\xi^{d-1}}\), for some constant \(\kappa\).}

We now use the representation Equation~(\ref{d10}) of Theorem~\ref{pairdiff} together with the formulas of Equations~(\ref{fd3a}) and~(\ref{vk3}) to rewrite Equation~(\ref{dir2}) as
\begin{equation}\label{dir3}
\begin{split}
\breve{\mathcal{E}}(\tilde{f},\tilde{g})=\left<-\tilde{A}\tilde{f},\tilde{g} \right>
                                            &=\frac{1}{2}\int_{\mathbb{R}^d}\sigma(|\bs{x}|^2)\varphi'(|\bs{x}|)\gamma'(|\bs{x}|)\,d\bs{x}\\
                                            &+\frac{1}{2}\int_{\mathbb{R}^d}\psi(|\bs{x}|)\varphi'(|\bs{x}|)\gamma(|\bs{x}|)\,d\bs{x}.
\end{split}
\end{equation}
Recall that \(\bs{div}\tilde{\bs{D}}(\bs{x})=\psi(|\bs{x}|)\bs{u}(\bs{x})\), so the direction of this vector is always radial; the scalar-valued function \(\psi\) is just its outward radial
component. Since the integrands in Equation~(\ref{dir3}) depend only on the radial component of \bs{x},  \(\xi=|\bs{x}|\), we can rewrite Equation~(\ref{dir3}) in terms of integrals over
\([0,\infty)\). Thus\footnote{Recall that \(\omega_{d-1}\) is the surface area of the unit sphere \(\mathbb{S}^{d-1}\).}
\begin{equation}\label{dir4}
\begin{split}
\bar{\mathcal{E}}(\varphi,\gamma)=\left<-\tilde{A}\tilde{f},\tilde{g} \right>
                                            &=\frac{1}{2}\omega_{d-1}\int_0^{\infty}\sigma(\xi^2)\varphi'(\xi)\gamma'(\xi)\xi^{d-1}\,d\xi\\
                                            &+\frac{1}{2}\omega_{d-1}\int_0^{\infty}\psi(\xi)\varphi'(\xi)\gamma(\xi)\xi^{d-1}\,d\xi.
\end{split}
\end{equation}
The \emph{bar} (\( \,\bar{}\, \)) over the \(\mathcal{E}\) indicates that we regard the underlying equilibrium measure of the quadratic form restricted to radially dependent functions to be the
\(1\)-dimensional \emph{weighted} measure \( m_d(d\xi) \) on \([0,\infty)\) given by the \(d-\)dimensional volume element
\begin{equation}\label{md} m_d(d\xi):=\omega_{d-1}\xi^{d-1}\,d\xi.
 \end{equation}
Observe now that the first term in Equation~(\ref{dir4}) is symmetric with respect to the measure defined by \(m_d(d\xi):=\omega_{d-1}\xi^{d-1}\,d\xi\) and the second is not symmetric with respect
to this measure. For later convenience, write
\begin{align}
\bar{\mathcal{E}}_s(\varphi,\gamma)  &:=\frac{1}{2}\int_0^{\infty}\sigma(\xi^2)\varphi'(\xi^2)\gamma'(\xi)\,m_d(d\xi),\\
\bar{\mathcal{E}}_{ns}(\varphi,\gamma)  &:=\frac{1}{2}\int_0^{\infty}\psi(\xi)\varphi'(\xi)\gamma(\xi)\,m_d(d\xi),
\end{align}
so \(\bar{\mathcal{E}}(\varphi,\gamma)=\bar{\mathcal{E}}_s(\varphi,\gamma)+\bar{\mathcal{E}}_{ns}(\varphi,\gamma) \).

Finally, after a bit of standard manipulation, involving another integration-by-parts, we obtain
\begin{equation}\label{dir5}
\bar{\mathcal{E}}(\varphi,\gamma)=\left<-\tilde{A}\tilde{f},\tilde{g} \right> =-\int_0^{\infty}\frac{1}{2}\Big(\sigma(\xi^2)\varphi''(\xi)+\frac{d-1}{\xi}\sigma_{\perp}(\xi^2)\varphi'(\xi)
\Big)\gamma(\xi)\,m_d(d\xi),
\end{equation}
or, equivalently,
\begin{equation}\label{dir6}
\breve{\mathcal{E}}(\tilde{f},\tilde{g})=\left<-\tilde{A}\tilde{f},\tilde{g} \right>    =\int_{\mathbb{R}^d}(-\dot{A}\varphi)(|\bs{x}|)\gamma(|\bs{x}|)\,d\bs{x},
\end{equation}
where \(\dot{A}\) is the Bessel-type operator defined in Equation~(\ref{mag5a}). 

Using the quadratic form, we can decompose the generator \(\dot{A}\) of Equation~(\ref{mag5a}) into two parts, \(\dot{A}=\dot{A}_{s}+\dot{A}_{ns}\), where \(\dot{A}_{s}\) is the generator associated
with the negative definite, symmetric part of the quadratic form in Equation~(\ref{dir4}) and \(\dot{A}_{ns}\) is associated with the remaining non-symmetric part. Thus,
\begin{align}\label{dir6a}
    \begin{split}
    (\dot{A}_{s}\varphi)(\xi) &=\frac{1}{2}\Big(\frac{d-1}{\xi}\sigma(\xi^2)\varphi'(\xi)+2\xi\sigma'(\xi^2)\varphi'(\xi)+\sigma(\xi^2)\varphi''(\xi)\Big)\\
    &=\frac{1}{2}\frac{1}{\xi^{d-1}}\Big(\big(\xi^{d-1}\sigma(\xi^2)\big)'\varphi'(\xi)+ \big(\xi^{d-1}\sigma(\xi^2)\big)\varphi''(\xi)  \Big)\\
    &=\frac{1}{2}\frac{1}{\xi^{d-1}}\Bigg(\xi^{d-1}\sigma(\xi^2)\varphi'(\xi)  \Bigg)',
    \end{split}\\
    \begin{split}
    (\dot{A}_{ns}\varphi)(\xi) &=\frac{1}{2}\Big(\frac{d-1}{\xi}\sigma_{\perp}(\xi^2)\varphi'(\xi)-\frac{d-1}{\xi}\sigma(\xi^2)\varphi'(\xi)-2\xi\sigma'(\xi^2)\varphi'(\xi)\Big)\\
    &=\frac{1}{2}\frac{1}{\xi^{d-1}}\Bigg(\frac{d-1}{\xi}\Big(\xi^{d-1}\sigma_{\perp}(\xi^2)\Big)\varphi'(\xi)-\Big(\xi^{d-1}\sigma(\xi^2)\Big)'\varphi'(\xi)   \Bigg)\\
    &=-\frac{1}{2}\psi(\xi)\varphi'(\xi).
    \end{split}
\end{align}
Our \(d\)-dimensional flux corresponds to the second, \emph{non-symmetric}, part of the Dirichlet quadratic form in Equation~(\ref{dir1a} or~\ref{dir2}). Indeed,
\(J(\bs{x},t)=-\frac{1}{2}\psi(|\bs{x}|)X(\bs{x},t)\) is the \(d\)-dimensional van~Kampen flux rate given in Equation~(\ref{vk3a}) of Section~\ref{ssec-vk}.

There is a Markov diffusion in \(\mathbb{R}^d\) associated with the first, \emph{symmetric}, part of this form, generated by \(\dot{A}_s\); moreover, this diffusion has a symmetric transition
probability density. That is, in a given time interval, the probability of moving from \(\bs{x}\) to \(\bs{y}\) is the same as moving from \(\bs{y}\) to \(\bs{x}\).\footnote{See
Kotelenez~\cite{KO4}.} For the magnitudes, there is also a Markov diffusion in \([0,\infty)\) associated with this term with a symmetric transition probability density. So, \textit{mutatis mutandis}
the probability of moving from \(|\bs{x}|\) to \(|\bs{y}|\) is the same as moving from \(|\bs{y}|\) to \(|\bs{x}|\). Thus there can be no net flux rate associated with the symmetric part of \(
\left<-\tilde{A}\tilde{f},\tilde{g} \right> \).
The remaining non-symmetric part, associated with \(\dot{A}_{ns}\), therefore characterizes the deviation from this neutral bias; it agrees precisely with our \(d\)-dimensional flux rate,
Equation~(\ref{vk3a}), computed by the pill-box argument in Section~\ref{sec-clump}. Note that \(\dot{A}_{ns}\) also generates a Markov process, but it is deterministic (first-order)  and, hence,
not a diffusion.

Next, with a bit of manipulation, we produce an alternative decomposition of the quadratic form \(\bar{\mathcal{E}}(\varphi,\gamma)\) of Equation~(\ref{dir4}). Namely,

\begin{align}\label{dir6b}
\bar{\mathcal{E}}(\varphi,\gamma) &=\frac{1}{2}\int_0^{\infty}\Big(\frac{d}{d\xi}\sigma(\xi^2)-\frac{d-1}{\xi}\sigma_{\perp}(\xi^2)\Big)\varphi'(\xi)\gamma(\xi)\,m_d(d\xi)\\
                                    &-\frac{1}{2}\int_0^{\infty}\frac{d}{d\xi}\Big(\sigma(\xi^2)\frac{d}{d\xi}\varphi(\xi)  \Big)\gamma(\xi)\,m_d(d\xi).
\end{align}
Write \(\bar{\mathcal{E}}(\varphi,\gamma)=\bar{\mathcal{E}}_1(\varphi,\gamma)+\bar{\mathcal{E}}_2(\varphi,\gamma)\), where
\begin{align}
\bar{\mathcal{E}}_1(\varphi,\gamma)   &:=\frac{1}{2}\int_0^{\infty}\Big(\frac{d}{d\xi}\sigma(\xi^2)-\frac{d-1}{\xi}\sigma_{\perp}(\xi^2)\Big)\varphi'(\xi)\gamma(\xi)\,m_d(d\xi),\label{dir6c}\\
\bar{\mathcal{E}}_2(\varphi,\gamma)   &:=-\frac{1}{2}\int_0^{\infty}\frac{d}{d\xi}\Big(\sigma(\xi^2)\frac{d}{d\xi}\varphi(\xi)  \Big)\gamma(\xi)\,m_d(d\xi).\label{dir6d}
\end{align}
In this alternative decomposition neither \(\bar{\mathcal{E}}_1(\varphi,\gamma)\) nor \(\bar{\mathcal{E}}_2(\varphi,\gamma)\) is symmetric with respect to the equilibrium measure \(m_d(d\xi)\).
However, suppose we \emph{formally} replace the equilibrium measure \(m_d(d\xi)\) with Lebesgue measure \(d\xi\) in both \(\bar{\mathcal{E}}_1(\varphi,\gamma)\) and
\(\bar{\mathcal{E}}_2(\varphi,\gamma)\), thereby ignoring the fact that the underlying context is \(d\)-dimensional.  Then \(\bar{\mathcal{E}}_2(\varphi,\gamma)\) is symmetric with respect to
Lebesgue measure \(d\xi\) and would correspond, in this context, to a symmetric transition probability density; that is, it would be neutral in the sense that it would contribute no net flux rate.
The non-neutral term, \(\bar{\mathcal{E}}_1(\varphi,\gamma)\), now corresponds to van~Kampen's \(1\)-dimensional   flux rate in Equation~(\ref{dir0}) associated with the generator \(\dot{A}\),
provided the initial probability density \(X(\xi,t)\) is locally spatially constant. (See Appendix~\ref{app-b2}.)

This discussion shows that for the diffusions in \(\mathbb{R}^d,\,d\geq 2\), which we consider, the appropriate equilibrium measure to compute a van~Kampen flux rate for the
\hyperlink{sepproc}{separation-process}, \(\xi=|\bs{x}|\), is the \emph{weighted} measure \(m_d(d\xi)\),  \emph{not} Lebesgue measure, \(d\xi\). The weighting reflects the adjustment needed to
account for the fact that the \(1\)-dimensional generator is, in this case, a radial reduction from \(d\)-dimensions. If van~Kampen's procedure is applied directly to the generator~\(\dot{A}\), the
underlying equilibrium measure is tacitly Lebesgue measure. For \(d\geq2\), the resulting flux rate will not then reflect the \(d\)-dimensional nature of the problem.  In view of the spherical
symmetry, this is essentially a lateral effect.\footnote{An examination of the the lateral term (the term containing \(\sigma_{\perp}\)) in the last equality of Equation~(\ref{vk5}) shows that this
lateral contribution is always positive.} We therefore claim that the \(d\)-dimensional van-Kampen flux rate we defined in Equation~\ref{vk2} should provide an appropriate and faithful measure of
the tendency toward clustering.

\begin{maxwell}\label{com-vkflux2}In Comment~\ref{com-vkflux} we compute our \(d\)-dimensional van~Kampen flux for the specific case of the Maxwell Kernel. As the computation shows, there is a small region of
attraction bias of magnitude \(\sqrt{(2+d)\varepsilon}\). The size of this region \emph{increases} with dimension \(d\). On the other hand, In Appendix~\ref{app-b2} we compute the \(1\)-dimensional
van~Kampen flux associated with the generator \(\dot{A}\) for the Maxwell kernel. For this flux there is a region of attraction bias for dimensions \(d=1,2,3\) only; moreover, the size of the region
\emph{decreases} with~\(d\). The drift term corresponding to \(\dot{A}\) depends only on the lateral probability flux, whereas, from our pill-box argument, the drift term in our \(d\)-dimensional
flux includes the effect of the radial probability flux through the ends. Without it this lateral effect is geometrically dominant. For \(d>3\), it completely overwhelms any inward radial bias. We
have argued that our \(d\)-dimensional flux is an appropriate measure of attraction bias since it does not include the geometrical, dimension dependent, effects associated with the neutral
(symmetric) part of the Dirichlet quadratic form. We claim that it faithfully represents the \emph{net} flux rate of \bs{x} out of a small region in \(\mathbb{R}^d\). It is not the same as the
\(1-\)dimensional van~Kampen flux rate of \(|\bs{x}|\) for a small region of \((0,\infty)\) associated with the generator \(\dot{A}\).
\end{maxwell}

\section{Summary of Results --- Further Research\label{sec-sum}}
The depletion effect, which accounts for the tendency toward clustering in colloids, even in the absence of electrostatic forces or intermolecular (van der Waals) forces, can now be measured with
some precision. The very term used, namely \textit{depletion,} is predicated on a notion of the interaction between the relatively few large particles, which comprise the solute, and the very large
number of very small particles, which comprise the solvent. We reviewed the classical model, based on the interaction of a few large spheres with many more small spheres. We also reviewed some
related experimental methods and data.

In this work we took an essentially statistical/probalistic approach to modeling the depletion effect. The foundation of our model is an interactive system of finitely many large particles and
infinitely many small particles. The dynamic equations for the interaction at the microscopic scale are essentially Newtonian. A kinetic model for the positions of the large particles at the
mesoscopic scale is obtained through an appropriate stochastic limit.  These equations are kinetic in the sense that the scaling limit renders negligible the inertial effects of the ensemble of
small particles on each large particle; that is, in the limit, only the fluctuation effects of the small particles on the position of a large particle are retained through a mean field force. 
The kinetic equations retain a small length scale, the correlation length, which derives essentially from the variance of the distribution of velocities assumed for the small particles.

We started with four phenomenological desiderata that a good model should have and that our model exhibits. Starting with our kinetic model we computed the correlation of the joint motion of a pair
of large particles. Although the motion for each individual particle is Brownian, the joint motion of the pair is not. (The joint motion becomes uncorrelated in the limit as the correlation length
approaches zero or as the large particles move far apart.) Specifically, we computed the covariance (diffusion) matrix for the joint motion of two large particles. As a consequence of fundamental
material restrictions (\hyperlink{isotropic}{isotropy}), this matrix was shown to have a very special structure. Using this special structure and a classical pill-box argument, we then defined and
computed a probability flux rate that, starting from a uniform distribution, was used to measure the tendency for the separation between two particles to decrease or increase. We computed this
tendency explicitly in the case that the velocity distribution of the small particles at the microscopic scale is Maxwellian. We showed that if the two particles are sufficiently close together, as
measured by the correlation length, there is a statistical bias in favor of further decrease in separation --- a depletion effect. We also showed how our flux rate is an appropriate generalization
of van~Kampen's one-dimensional flux rate. In parallel discussions we explicitly computed the infinitessimal generators for the two stochastic processes consisting of the (vector) difference and the
(scalar) distance between the two particles. We then showed the precise relationship between these generators and our version of van~Kampen's one-dimensional flux rate.

We emphasize that our results predict an initial statistical tendency toward clustering from a uniform state. Our procedure can easily be modified to account for a non-uniform initial state, but we
did not do so here. Of course, after the onset of any clustering, the state may no longer be uniform.

On the basis of our model, the long-term behavior of the separation between two large particles was shown to be recurrent in two dimensions and transient in higher dimensions. The relationship
between our short-term results and our long-term results needs further study.

Our model is kinetic; there are no \emph{forces} in the Newtonian sense. There is an underlying force structure in the original dynamic interactions, but they are only implicit in the derived
kinetic model. To associate forces directly with the clustering effect we derive would require embedding our model in some sort of variational structure. This is desirable from a conceptual and
physical point of view but, we believe, not strictly necessary.

There is a compelling qualitative comparison between the behavior of our flux rate, as computed for the Maxwellian case, and  potentials experimentally measured between large particles and a
substrate or, equivalently, between two large particles. Since the latter are in good agreement with the classical hard sphere model for the depletion effect, we are encouraged to believe that our
stochastic approach can faithfully model this effect in colloids.

Avenues of further investigation that may be fruitful using our approach are:
\begin{itemize}
    \item Devise experiments to measure directly the effects we compute. A force model may be needed.
    \item Incorporate interactions between large particles and inertial effects.
    \item Compute the tendency to clump from non-uniform initial states.
    \item Produce a directed model; that is, give the particles some preferred direction (shape).
    \item Reproduce the Vrij model computations in \(d\)-dimensions.
    \item Investigate the difference between long- and short-term behavior.
\end{itemize}

\newpage
\appendix
\section{Isotropic Functions}\label{app-a}
Here we outline a proof of Lemma~\ref{lem0}. We begin with a definition, a statement of Cauchy's Representation Theorem for isotropic functions, and some useful results of linear algebra.
\begin{define}A scalar-valued function \(\phi(\bs{v}_1,\bs{v}_2,\ldots ,\bs{v}_m)\) of \(m\)~vectors in \(\mathbb{R}^d\) is \textit{isotropic} whenever
\begin{equation}\label{a1}
\phi(\bs{v}_1,\bs{v}_2,\ldots ,\bs{v}_m)=\phi(\bs{Q}\bs{v}_1,\bs{Q}\bs{v}_2,\ldots ,\bs{Q}\bs{v}_m),
\end{equation}
for every orthogonal transformation \bs{Q}.\footnote{If the underlying invariance group consists of the subgroup of rotations (proper orthogonal transformations), \(\phi\) is said to be
\textit{hemitropic.}}
\end{define}
\begin{thm}[Cauchy's Representation Theorem]\label{cauchy}
\(\phi(\bs{v}_1,\bs{v}_2,\ldots ,\bs{v}_m)\) is an isotropic scalar-valued function of \(m\)~vectors in \(\mathbb{R}^d\) if and only if it can be expressed as a scalar-valued function \(\varphi\) of
the  \(\frac{m(m+1)}{2}\) inner products \(\{\bs{v}_i\bullet\bs{v}_j:\,i,j=1,2,\ldots,m\}\).\footnote{This is the version stated and proved in Truesdell and Noll~\cite{TR}(1965, \S B,II,11). There
is a corresponding result for hemitropic functions that requires the inclusion of all determinants.}
\end{thm}

For a transformation \bs{B} in \(\mathbb{R}^d\) and  a scalar \(\beta\), the set \(\{\bs{y}:\bs{By}=\beta\bs{y}\}\) is a subspace of \(\mathbb{R}^d\). If this subspace contains a non-zero vector,
say \(\bs{v}\not=\bs{0}\), it is called a \textit{characteristic subspace for \bs{B} corresponding to \(\beta\)} and \bs{v} is  a \textit{characteristic vector (eigenvector) for~\bs{B} corresponding
to characteristic value (eigenvalue)~\(\beta\)}.

Let \bs{v} be a non-zero vector. A transformation \(\bs{Q}(\bs{v})\) is a \textit{simple reflection (through the hyperplane~\(\{\bs{v}\}^{\perp}\))} whenever \(\bs{Q}(\bs{v})\bs{v}=-\bs{v}\) and
\(\bs{Q}(\bs{v})\bs{a}=\bs{a}\), for all \(\bs{a}\in\{\bs{v}\}^{\perp}\). It is easy to see that \(\bs{Q}(\bs{v})\) is an orthogonal, but not proper orthogonal, transformation; that is,
\(det(\bs{Q}(\bs{v}))=-1\). \(\bs{Q}(\bs{v})\) is also symmetric and satisfies
\begin{equation}\label{reflect}
\begin{split}
\bs{Q}(\bs{v})\bs{P}(\bs{v})&=\bs{P}(\bs{v})\bs{Q}(\bs{v})=-\bs{P}(\bs{v}),\\
\bs{Q}(\bs{v})\bs{P}^{\perp}(\bs{v})&=\bs{P}^{\perp}(\bs{v})\bs{Q}(\bs{v})=\bs{P}^{\perp}(\bs{v}),\\
\bs{Q}(\bs{v})&=-\bs{P}(\bs{v})+\bs{P}^{\perp}(\bs{v}).
\end{split}
\end{equation}

\begin{thm}[Commutation Theorem]\label{commute}
Let \bs{A} and \bs{B} be two transformations that commute. Then \bs{A} leaves each characteristic space of \bs{B} invariant.\footnote{Conversely, if a  transformation \bs{A} leaves each
characteristic space of a \emph{symmetric}  transformation \bs{B} invariant, then \bs{A} and \bs{B} commute.}
\end{thm}
An immediate corollary of the Commutation Theorem is
\begin{cor}[Reflection Corrolary]\label{commute2}
If \bs{B} is a transformation that commutes with the simple reflection \(\bs{Q}(\bs{v})\), then \bs{v} is a characteristic vector for \bs{B}. In this case, \bs{B} has the representation
\begin{equation}\label{brep}
\bs{B}=\big(\bs{u}(\bs{v})\bullet\bs{B}\bs{u}(\bs{v})\big)\bs{P}(\bs{v})+\bs{P}^{\perp}(\bs{v})\bs{B}\bs{P}^{\perp}(\bs{v})=\lambda(\bs{v})\bs{P}(\bs{v})+\bs{P}^{\perp}(\bs{v})\bs{B}\bs{P}^{\perp}(\bs{v}),
\end{equation}
where\ \(\lambda(\bs{v})\) is the eigenvalue associated with \bs{v}.
\end{cor}
As a consequence of the Reflection Corrolary~\ref{commute2} we see that if \bs{B} commutes with \(\bs{Q}(\bs{v})\), for \emph{every}  \(\bs{v}\not=\bs{0}\), then \emph{every}  \(\bs{v}\not=\bs{0}\)
is a characteristic vector for \bs{B}. The latter is equivalent to \bs{B} is a multiple of the identity.\footnote{For every  \(\bs{v}\not=\bs{0}\), there is a number \(\lambda(\bs{v})\) such that
\(\bs{B}\bs{v}=\lambda(\bs{v})\bs{v}\). It follows easily that \(\lambda(\bs{v})\) is constant for all  \(\bs{v}\not=\bs{0}\).} Since simple reflections are orthogonal, but not proper orthogonal, we
have
\begin{thm}\label{full}
A transformation commutes with all orthogonal transformations if, and only if, it is a multiple of the identity.\footnote{In fact, a transformation that commutes with all proper orthogonal
transformations and a single simple reflection is a multiple of the identity.}
\end{thm}
\begin{com}
The following result is  useful when the group of orthogonal transformations is replace by the sub-group of proper orthogonal transformations (rotations).
\begin{thm}\label{proper}
A symmetric transformation  commutes with all proper orthogonal transformations if, and only if, it is a multiple of the identity.
\end{thm}
To see this, consider a linear transformation \bs{B} that commutes with all proper orthogonal transformations~\bs{Q}. Suppose further that \bs{B} has at least one characteristic vector; that is,
there is a vector \(\bs{v}\not=\bs{0}\) and a scalar \(\beta\) such that \(\bs{B}\bs{v}=\beta\bs{v}\).  Then, by the Commutation Theorem~\ref{commute}, \(\bs{Q}\bs{v}\) is a characteristic vector
for \bs{B} corresponding to \(\beta\) for every proper orthogonal~\bs{Q}. This implies that \emph{all} non-zero vectors are eigenvectors for \bs{B} corresponding to the eigenvalue \(\beta\) or,
succinctly, \(\bs{B}=\beta\bs{1}\). Finally, if \bs{B} is symmetric it surely has a characteristic vector.
\end{com}

\noindent\begin{proof}[Proof of Lemma~\ref{lem0}] First note that when the dimension \(d=1\), Parts~(\ref{lem01}) and~(\ref{lem03}) simply assert that the function is even, while Part~(\ref{lem02})
asserts that it is odd. In all three cases the result clearly holds. Henceforth we suppose that \(d\geq2\).

Part~(\ref{lem01}) follows at once from Cauchy's Representation Theorem~\ref{cauchy} with \(m=1\).

To verify Part~(\ref{lem02}), observe that \hyperlink{isotropic}{isotropy}  implies that \(\bs{f}(\bs{0})=\bs{0}\), so we may suppose that \(\bs{r}\not =\bs{0}\).  For an arbitrary vector~\bs{a},
define \(\phi(\bs{r},\bs{a})\) by
\begin{equation}\label{a2}
\phi(\bs{r},\bs{a}):=\bs{f}(\bs{r})\bullet\bs{a}.
\end{equation}
By its construction, the function \(\bs{a}\mapsto\phi(\bs{r},\bs{a})\) is \emph{linear} for each  \(\bs{r}\not =\bs{0}\). It is easy to show that \(\phi\) is an isotropic scalar-valued function of
the two vector variables \(\bs{r}\) and \(\bs{a}\).  Hence, by Cauchy's Representation Theorem~\ref{cauchy},
\begin{equation}\label{a3}
\bs{f}(\bs{r})\bullet\bs{a}=\varphi(|\bs{r}|^2,|\bs{a}|^2,\bs{r}\bullet\bs{a}),
\end{equation}
for some scalar-valued function \(\varphi\) of three scalar variables. With \(\bs{r}\not=\bs{0}\) fixed Consider the action of the linear function
\(\bs{a}\mapsto\varphi((|\bs{r}|^2,|\bs{a}|^2,\bs{r}\bullet\bs{a})\) on the subspace \(\{\bs{r}\}^{\perp}\). The function \(\bs{a}\mapsto\varphi((|\bs{r}|^2,|\bs{a}|^2,0)\) on \(\{\bs{r}\}^{\perp}\)
must also be linear, so
\begin{equation}
\varphi((|\bs{r}|^2,|-\bs{a}|^2,0)=\varphi((|\bs{r}|^2,|\bs{a}|^2,0)=-\varphi((|\bs{r}|^2,|\bs{a}|^2,0).
\end{equation}
Hence, \(\bs{a}\mapsto\varphi((|\bs{r}|^2,|\bs{a}|^2,0)=0\) on \(\{\bs{r}\}^{\perp}\) and, by Equation~(\ref{a3}), \(\bs{a}\mapsto\bs{f}(\bs{r})\bullet\bs{a}=0\) on \(\{\bs{r}\}^{\perp}\). In other
words, \(\bs{f}(\bs{r})\in\{\{\bs{r}\}^{\perp}\}^{\perp}\). In our finite dimensional (reflexive) space, \(\{\{\bs{r}\}^{\perp}\}^{\perp}=\{\bs{r}\}\), so \(\bs{f}(\bs{r})\in\{\bs{r}\}\).
Equivalently, there must be a scalar function \(\bs{r}\mapsto\hat{\varphi}(\bs{r})\) such that \(\bs{f}(\bs{r})=\hat{\varphi}(\bs{r})\bs{r}\). (For definiteness, define
\(\hat{\varphi}(\bs{0}):=0\).) It is easy to show that the \hyperlink{isotropic}{isotropy} of \bs{f} implies the \hyperlink{isotropic}{isotropy} of \(\hat{\varphi}\). By Part~(\ref{lem01}) of the
Lemma, there is a scalar function \(\hat{\hat{\varphi}}\) on \(\mathbb{R}^+\) such that \(\hat{\varphi}(\bs{r})=\hat{\hat{\varphi}}(|\bs{r}|^2)\) and, hence,
\(\bs{f}(\bs{r})=\hat{\hat{\varphi}}(|\bs{r}|^2)\bs{r}\). Part~(\ref{lem02}) is proved.

To verify Part~(\ref{lem03}) we start with two observations. From Equation~(\ref{frinv3}) we see that the matrix \(\bs{F}(\bs{0})\)  commutes with all orthogonal matrices. By Theorem~\ref{full},
\(\bs{F}(\bs{0})=\kappa\bs{1}\), for some scalar constant \(\kappa\). Thus Part~(\ref{lem03}) certainly holds if \(\bs{r}=\bs{0}\). Henceforth, we suppose \(\bs{r}\not =\bs{0}\). Note also that
Equation~(\ref{frinv3}) implies that \(\bs{F}(\bs{r})=\bs{F}(-\bs{r})\). We will verify the alternate statement of Part~(\ref{lem03}).

Set \(\bs{f}(\bs{r}):=\bs{F}(\bs{r})\bs{r}\). It is easy to see that \bs{f} satisfies Equation~(\ref{frinv2}) for all orthogonal transformations \bs{Q}, so it is \hyperlink{isotropic}{isotropic.} By
Part~(\ref{lem02}) of Lemma~(\ref{lem0}), we must have \(\bs{f}(\bs{r})=\lambda(|\bs{r}|^2)\bs{r}\), for some (real-valued) scalar function \(\lambda\). Thus,
\(\bs{F}(\bs{r})\bs{r}=\lambda(|\bs{r}|^2)\bs{r}\) and, hence, \(\left<\lambda(|\bs{r}|^2);\bs{r}\right>\) is  an eigenpair for \(\bs{F}(\bs{r})\). The same argument shows that
\(\left<\lambda(|\bs{r}|^2);\bs{r}\right>\) is also an eigenpair for the transpose \(\bs{F}^T(\bs{r})\).  Indeed, since the vector-valued function \(\bs{F}^T(\bs{r})\bs{r}\) is also
\hyperlink{isotropic}{isotropic,} there is a scalar function \(\mu\) such that \(\bs{F^T}(\bs{r})\bs{r}=\mu(|\bs{r}|^2)\bs{r}\), so \(\bs{F}(\bs{r})\) has the eigenpair
\(\left<\mu(|\bs{r}|^2);\bs{r}\right>\). But \(\bs{F}^T(\bs{r})\bs{r}\bullet\bs{r}=\bs{F}(\bs{r})\bs{r}\bullet\bs{r}\), so \(\lambda=\mu\). (In fact, all the principal invariants of
\(\bs{F}(\bs{r})\) are isotropic scalar functions of \bs{r}, so all the eigenvalues of  \(\bs{F}(\bs{r})\) and \(\bs{F}^T(\bs{r})\) are scalar-valued functions of \(|\bs{r}|^2\).)\footnote{The
arguments of this paragraph and the next do not require that \(\bs{F}(\bs{r})\) be symmetric.}  Surely the subspace \(\{\bs{r}\}\) is invariant under the action of \(\bs{F}(\bs{r})\). Let
\(a\in\{\bs{r}\}^{\perp}\). Then
\begin{equation}
\bs{F}(\bs{r})\bs{a}\bullet\bs{r}=\bs{a}\bullet\bs{F}^T(\bs{r})\bs{r}=\lambda(|\bs{r}|^2)\bs{a}\bullet\bs{r}=0.
\end{equation}
Hence \(\bs{F}(\bs{r})\bs{a}\in\{\bs{r}\}^{\perp}\); that is, the subspace \(\{\bs{r}\}^{\perp}\) is also invariant under the action of \(\bs{F}(\bs{r})\). Using the projections \(\bs{P}(\bs{r})\)
and \(\bs{P}^{\perp}(\bs{r})\) defined in Equation~(\ref{proj}), \(\bs{F}(\bs{r})\) must have the decomposition
\begin{equation}\label{a14}
\begin{split}
\bs{F}(\bs{r})  &=\bs{P}(\bs{r})\bs{F}(\bs{r})\bs{P}(\bs{r})+\bs{P}^{\perp}(\bs{r})\bs{F}(\bs{r})\bs{P}^{\perp}(\bs{r})\\
                &=\lambda(|\bs{r}|^2)\bs{P}(\bs{r})+\bs{P}^{\perp}(\bs{r})\bs{F}(\bs{r})\bs{P}^{\perp}(\bs{r}).
\end{split}
\end{equation}

Define  \(\bs{F}^{\perp}(\bs{r}):=\bs{P}^{\perp}(\bs{r})\bs{F}(\bs{r})\bs{P}^{\perp}(\bs{r})\). Clearly \(\bs{F}^{\perp}(\bs{r})\) is \hyperlink{isotropic}{isotropic.} Then, for every orthogonal
\(\hat{\bs{Q}}\) that leaves the subspace \(\{\bs{r}\}\) invariant, we have
\begin{equation}\label{a15}
\bs{F}^{\perp}(\hat{\bs{Q}}\bs{r})=\bs{F}^{\perp}(\pm\bs{r})=\bs{F}^{\perp}(\bs{r})=\hat{\bs{Q}}\bs{F}^{\perp}(\bs{r})\hat{\bs{Q}}^T.
\end{equation}
This follows since \(\hat{\bs{Q}}\bs{r}=\pm\bs{r}\) and \(\bs{F}^{\perp}(\bs{r})=\bs{F}^{\perp}(-\bs{r})\). Therefore, \(\bs{F}^{\perp}(\bs{r})\) commutes with every orthogonal
matrix~\(\hat{\bs{Q}}\) that leaves \(\{\bs{r}\}\) invariant. But the set of all orthogonal transformations that leave the subspace \(\{\bs{r}\}\) invariant is a subgroup equivalent to the group of
all orthogonal transformations on \(\{\bs{r}\}^{\perp}\). Again, using Theorem~\ref{full},  the action of \(\bs{F}^{\perp}(\bs{r})\) on \(\{\bs{r}\}^{\perp}\) is a multiple of the identity on
\(\{\bs{r}\}^{\perp}\), which multiplier depends only on \bs{r}. Thus,
\begin{equation}\label{a16}
\bs{F}^{\perp}(\bs{r})=\psi(\bs{r})\bs{P}^{\perp}(\bs{r}),
\end{equation}
for some scalar function \(\psi(\bs{r})\). Again, it is easy to show that \(\psi(\bs{r})\) is a scalar-valued isotropic function, so \(\psi(\bs{r})=\lambda_{\perp}(|\bs{r}|^2)\) for some
scalar-valued function \(\lambda_{\perp}\). In summary,
\begin{equation}\label{a17}
\bs{F}(\bs{r})=\lambda(|\bs{r}|^2)\bs{P}(\bs{r})+\lambda_{\perp}(|\bs{r}|^2)\bs{P}^{\perp}(\bs{r}),
\end{equation}
which is precisely Equation~(\ref{frinv7}) in the alternate statement of Part~(\ref{lem03}). The proof of Part~(\ref{lem03}) is complete upon identifying the eigenvalue \(\lambda(|\bs{r}|^2)\) with
\(\big(\gamma(|\bs{r}|^2)+|\bs{r}|^2\eta(|\bs{r}|^2)\big)\) and the eigenvalue \(\lambda_{\perp}(|\bs{r}|^2)\) with \(\gamma(|\bs{r}|^2)\) in Equation~(\ref{frinv6}).

Finally, for \(d\geq 2\), \(\bs{F}(\bs{r})\) is always symmetric with either one eigenvalue of multiplicity \(d\) or two distinct eigenvalues with multiplicities \(1\) and \(d-1\).
\end{proof}

\section{Stochastics}\label{app-stoch}

\subsection{Space-Time White Noise}\label{app-white}
Let \(\beta(ds)\) denote the standard scalar-valued Brownian motion and let \(A,\,A_1,\,A_2\) be Borel subsets of \([0,\infty)\) of finite Lebesgue measure. Here, \(|A|\) denotes the Lebesgue
measure of \(A\), etc. Then \(\int_A \beta(ds)\) is normally distributed with mean \(0\) and variance \(|A|\). Furthermore,  if \(A_1\cap A_2=\emptyset\) for two Borel subsets of \([0,\infty)\) of
finite Lebesgue measure, then \(\int_{A_1}\beta(ds)\) and \(\int_{A_2}\beta(ds)\) are independent.

Space-time white noise, \(w(d\bs{q},ds)  \), is a straightforward generalization. Let \(B,B_1,B_2\) be Borel subsets of \(\mathbb{R}^d\) of finite Lebesgue measure. 
Then \(\int_A\int_B w(d\bs{q},ds)\) is normally distributed with mean \(0\) and variance \(|A|\cdot|B|\). Furthermore, if \((B_1\times A_1)\cap(B_2\times A_2)=\emptyset\), then
\(\int_{A_1}\int_{B_1} w(d\bs{q},ds)\) and \(\int_{A_2}\int_{B_2} w(d\bs{q},ds)\) are independent.

\subsection{Marginals for Pair-Paths}\label{app-marg} Let \(\bs{r}^{\alpha}(\cdot,\bs{r}^{\alpha}_0)\) denote the random path in \(\mathbb{R}^d\) traced out by the position of the \(\alpha\)th large
particle with initial random position \(\bs{r}^{\alpha}_0\in\mathbb{R}^d\),
\(\alpha=1,2\). Then the marginal distribution of each component of the pair-path \(\begin{pmatrix}\bs{r}^1(\cdot,\bs{r}^1_0)\\
\bs{r}^2(\cdot,\bs{r}^2_0)\end{pmatrix}\)  in \(C([0,\infty),\mathbb{R}^{2d})\) is the distribution of a Wiener measure on \(C([0,\infty),\mathbb{R}^d)\) with initial support in the point
\(\bs{r}^{\alpha}_0\); that is, on the space of \(\mathbb{R}^d\)-valued continuous functions, \(\bs{f}(\cdot)\), on \([0,\infty)\) such that \(\bs{f}(0)=\bs{r}^{\alpha}_0,\,\forall\bs{f}\).

\subsection{Quadratic Variation}\label{app-qv} Let \(\bs{m}(\cdot)\) and \(\bs{n}(\cdot)\) be two vector-valued continuous martingales on the interval
\(I=[0,T],\,T>0\) or \(I=[0,\infty)\). The \textit{quadratic variation} of \(\bs{m}(\cdot)\), denoted by \(\langle\negthinspace\langle\bs{m}\rangle\negthinspace\rangle(\cdot)\), is the unique,
increasing, adapted, matrix-valued process on~\(I\) such that \(\langle\negthinspace\langle\bs{m}\rangle\negthinspace\rangle(0)=\bs{0}\) and
\(\bs{m}(\cdot)\bs{m}^T(\cdot)-\langle\negthinspace\langle\bs{m}\rangle\negthinspace\rangle(\cdot)\) is a continuous martingale.\footnote{Here, \textit{increasing} means that the quadratic form
\(t\mapsto\langle\negthinspace\langle\bs{m}\rangle\negthinspace\rangle(t)\bs{a}\bullet\bs{a}\) is increasing for every constant vector \bs{a}.} Similarly, the \textit{cross quadratic variation} of
\(\bs{m}(\cdot)\) and \(\bs{n}(\cdot)\), denoted by \(\langle\negthinspace\langle\bs{m},\bs{n}\rangle\negthinspace\rangle(\cdot)\), is the unique, adapted, matrix-valued process on \(I\) such that
\(\langle\negthinspace\langle\bs{m},\bs{n}\rangle\negthinspace\rangle(0)=\bs{0}\) and \(\bs{m}(\cdot)\bs{n}^T(\cdot)-\langle\negthinspace\langle\bs{m},\bs{n}\rangle\negthinspace\rangle(\cdot)\) is a
continuous martingale. Equivalently, we have
\(\langle\negthinspace\langle\bs{m},\bs{n}\rangle\negthinspace\rangle=\frac{1}{4}\big(\langle\negthinspace\langle\bs{m}+\bs{n}\rangle\negthinspace\rangle-
\langle\negthinspace\langle\bs{m}-\bs{n}\rangle\negthinspace\rangle \big)\).

Some of our conclusions depend on the following two well-known results. (See Ikeda and Watanabe~(1981)~\cite{IK} for details.)
\begin{thm}[L\'{e}vy-It\^{o} Theorem]\label{L-I}
If \(\langle\negthinspace\langle\bs{m}\rangle\negthinspace\rangle(t)=t\bs{B},\,t\in I\), where \bs{B} is a positive definite symmetric matrix, then \(\bs{m}(\cdot)\) is an adapted Brownian motion
(Wiener process) with independent increments. If \(\bs{B}=\bs{1}\), then \(\bs{m}(\cdot)\) is a \emph{standard} Brownian motion.
\end{thm}

\begin{thm}
If the martingales \(\bs{m}(\cdot)\) and \(\bs{n}(\cdot)\) are uncorrelated on \(I\), then \(\langle\negthinspace\langle\bs{m},\bs{n}\rangle\negthinspace\rangle(t)=\bs{0},\,t\in I\).
\end{thm}

\noindent\textbf{Computation of the (Tensor-Valued) Quadratic Variation of Equation~(\ref{qvqr1})} \newline Start with the the definition of \(\bs{m}^{\alpha}(t)\) given in Equation~(\ref{ito1b}).
To simplify the argument, assume that the underlying position process \(\bs{r}^{\alpha}(\cdot)\) is deterministic rather than stochastic. Next, partition the interval~\([0,t],\,t\in I,\,t>0,\) into
non-overlapping small intervals \(\{A_i:i=1,2,\ldots,M\}\) and partition \(\mathbb{R}^d\) into non-overlapping small \(d-\)parallelepipeds \(\{B_j:j=1,2,\ldots\infty\}\). Let \(s_i\) denote a time
in \(A_i\) and let \(\bs{q}^j\) denote a  location in \(B_j\). Then, \(\bs{m}^{\alpha}(t) \) is given approximately by
\begin{equation}\label{app-qv1}
\bs{m}^{\alpha}(t)\approx\sum_{i=1}^M\sum_{j=1}^{\infty}\bs{g}(\bs{r}^{\alpha}(s_i)-\bs{q}^j)\int_{A_i}\int_{B_j}w(d\bs{q},ds).
\end{equation}
Using the same partition, do the same, \textit{mutatis mutandis,} for \(\bs{m}^{\beta}(t)\). Under our simplifying assumption, the quadratic variation is the covariance
\(\langle\negthinspace\langle\bs{m}^{\alpha},\bs{m}^{\beta}\rangle\negthinspace\rangle(t)=E[(\bs{m}^{\alpha}(t)-E[\bs{m}^{\alpha}(t)])(\bs{m}^{\beta}(t)-E[\bs{m}^{\beta}(t)])^T]\). The properties of
the generalized space-time white noise guarantee that the processes \(\int_{A_i}\int_{B_j}w(d\bs{q},ds) \) each have mean zero and variance \(|A_i||B_j|\); moreover, these processes are independent;
that is,
\begin{equation}
\begin{split}
E\bigg[\int_{A_i}\int_{B_j}w(d\bs{q},ds)\bigg]&=0,\\
E\bigg[\bigg(\int_{A_i}\int_{B_j}w(d\bs{q},ds)\bigg)^2\bigg]&=|A_i||B_j|\\
E\bigg[\int_{A_i}\int_{B_j}w(d\bs{q},ds)\int_{A_k}\int_{B_l}w(d\bs{q},ds)\bigg]&=0\text{, if }i\not= k,j\not= l .
\end{split}
\end{equation}
Thus,
\begin{multline}\label{app-qv2}
\langle\negthinspace\langle\bs{m}^{\alpha},\bs{m}^{\beta}\rangle\negthinspace\rangle(t)=E[(\bs{m}^{\alpha}(t))(\bs{m}^{\beta}(t))^T]\\
\approx\sum_{i,k=1}^M\sum_{j,l=1}^{\infty}\bs{g}(\bs{r}^{\alpha}(s_i)-\bs{q}^j)\bs{g}^T(\bs{r}^{\beta}(s_k)-\bs{q}^l)E\bigg[\int_{A_i}\int_{B_j}w(d\bs{q},ds)\int_{A_k}\int_{B_l}w(d\bs{q},ds)\bigg]\\
=\sum_{i=1}^M\sum_{j=1}^{\infty}\bs{g}(\bs{r}^{\alpha}(s_i)-\bs{q}^j)\bs{g}^T(\bs{r}^{\beta}(s_i)-\bs{q}^j)E\bigg[\bigg(\int_{A_i}\int_{B_j}w(d\bs{q},ds)\bigg)^2\bigg]\\
=\sum_{i=1}^M\sum_{j=1}^{\infty}\bs{g}(\bs{r}^{\alpha}(s_i)-\bs{q}^j)\bs{g}^T(\bs{r}^{\beta}(s_i)-\bs{q}^j)|A_i||B_j|\\
=\sum_{i=1}^M\sum_{j=1}^{\infty}\bs{g}(\bs{r}^{\alpha}(s_i)-\bs{q}^j)\bs{g}^T(\bs{r}^{\beta}(s_i)-\bs{q}^j)\int_{A_i}\int_{B_j}d\bs{q}\,ds\\
\approx\int_0^t\int_{\mathbb{R}^d}\bs{g}(\bs{r}^{\alpha}(s)-\bs{q})\bs{g}^T(\bs{r}^{\beta}(s)-\bs{q})\,d\bs{q}\,ds.
\end{multline}
This essentially establishes the formula in Equation~(\ref{qvqr1}). The argument above should properly be modified to take into account the fact that the underlying position processes,
\(\bs{r}^{\alpha}(\cdot)=\bs{r}^{\alpha}(\cdot,\omega)\) and \(\bs{r}^{\beta}(\cdot)=\bs{r}^{\beta}(\cdot,\omega)\), are stochastic,  not deterministic. In this case, the expectations involved must
be conditioned on the underlying filtration.\footnote{Recall that the underlying probability space is \((\Omega,\mathcal{F},\mathcal{P})\). The underlying filtration is an increasing family of
\(\sigma-\)fields \(\{\mathcal{F}_t:t\in I\}\) in \(\mathcal{F}\) determined in a natural way through the histories of the processes \(\bs{r}^{\alpha}(\cdot)=\bs{r}^{\alpha}(\cdot,\omega)\) and
\(\bs{r}^{\beta}(\cdot)=\bs{r}^{\beta}(\cdot,\omega)\).} We omit the details here. 

Following the convention for pair-processes established in Section~\ref{subsec-not}, consider the  pair-martingale for \(\bs{m}(\cdot)\) and \(\bs{n}(\cdot)\):
\begin{equation}\label{joint1}
\hat{\bs{m}}=\begin{pmatrix}\bs{m}\\\bs{n}  \end{pmatrix}.
\end{equation}
 Then we have the block form for the martingale
\begin{equation}\label{joint2}
\hat{\bs{m}}\hat{\bs{m}}^T=
\begin{pmatrix}
    \bs{m}\bs{m}^T & \bs{m}\bs{n}^T\\
    \bs{n}\bs{m}^T & \bs{n}\bs{n}^T\\
    \end{pmatrix}.
\end{equation}
Hence, the quadratic variation of \(\hat{\bs{m}}\) must have the block form
\begin{equation}\label{joint3}
\langle\negthinspace\langle\hat{\bs{m}}\rangle\negthinspace\rangle=
\begin{pmatrix}
    \langle\negthinspace\langle\bs{m}\rangle\negthinspace\rangle & \langle\negthinspace\langle\bs{m},\bs{n}\rangle\negthinspace\rangle\\
    \langle\negthinspace\langle\bs{n},\bs{m}\rangle\negthinspace\rangle & \langle\negthinspace\langle\bs{n}\rangle\negthinspace\rangle
    \end{pmatrix}.
\end{equation}
From this and Theorem~\ref{L-I} we see that the joint motion is  Brownian if \( \langle\negthinspace\langle\hat{\bs{m}}\rangle\negthinspace\rangle(t)=t\bs{B} \), for some positive definite and
symmetric matrix \bs{B}, and standard Brownian if \( \langle\negthinspace\langle\hat{\bs{m}}\rangle\negthinspace\rangle(t)=t\bs{1} \). Of course,
\(\langle\negthinspace\langle\hat{\bs{m}}\rangle\negthinspace\rangle(t,\hat{\bs{r}}_0)=\int_0^t \hat{\bs{D}}(\hat{\bs{r}}(s,\hat{\bs{r}}_0))\,ds\), where \(\hat{\bs{D}}\) is given in
Equation~(\ref{qyqr2}) and \(\hat{\bs{r}}(\cdot,\hat{\bs{r}}_0)\) is the joint motion of the two large particles. The computations in Appendix~\ref{app-b} verify these observations explicitly for
the \hyperlink{maxker}{Maxwell kernel}.

\section{Outline of the Proof of Theorem~\ref{semigrp}}\label{app-semigrp}The proof of Theorem~\ref{semigrp} follows directly from It\^{o}'s famous formula. 
Here are the essential steps in the argument.

Recall that the pair-process \(\hat{\bs{r}}(\cdot,\hat{\bs{r}}_0)\) is the solution of the stochastic integral equation given in Equation~(\ref{ito4}). Rewrite this equation using the notation
established in Appendix~\ref{app-qv}, in particular Equations~(\ref{joint1}), (\ref{joint2}), and (\ref{joint3}). Thus,
\begin{equation}\label{sg1}
\hat{\bs{r}}(t)=\hat{\bs{r}}_0+\hat{\bs{m}}(t),
\end{equation}
where \(\hat{\bs{m}}(t)\) is the martingale defined by the second term on the right-hand side of Equation~(\ref{ito4}). Let \(\hat{\bs{y}}\mapsto\hat{f}(\hat{\bs{y}})\) be a twice continuously
differentiable scalar-valued function on \(\mathbb{R}^{2d}\).\footnote{These functions are dense in the class of bounded measurable functions on \(\mathbb{R}^{2d}\).}
Then, using Equation~(\ref{ito4}) and noting the differentiability of  \( \langle\negthinspace\langle\hat{\bs{m}}\rangle\negthinspace\rangle(\cdot) \), It\^{o}'s formula reduces to\footnote{Ikeda
and Watanabe~\cite{IK}(1981, Section~5, Theorem~5.1).  The matrix components \( \langle\negthinspace\langle\hat{\bs{m}}\rangle\negthinspace\rangle_{lm} \) are \(
\langle\negthinspace\langle\hat{m}_l,\hat{m}_m\rangle\negthinspace\rangle,\,l,m=1,2,\ldots,2d \).}
\begin{equation}\label{sg2}
\begin{split}
\hat{f}(\hat{\bs{r}}(t))-\hat{f}(\hat{\bs{r}}_0)&=\sum_{l=1}^{2d}\int_0^t\bigg(\frac{\partial}{\partial\hat{r}_l}\hat{f}\bigg)(\hat{\bs{r}}(s))d\hat{m}_l(s)\\
&+\frac{1}{2}\sum_{l,m=1}^{2d}\int_0^t
\bigg(\frac{\partial^2}{\partial \hat{r}_l\,\partial
\hat{r}_m}\hat{f}\bigg)(\hat{\bs{r}}(s))\frac{d}{ds}\langle\negthinspace\langle\hat{m}_l,\hat{m}_m\rangle\negthinspace\rangle(s)\,ds.\\
\end{split}
\end{equation}
The first term in Equation~(\ref{sg2}), containing the first-order derivatives of~\(\hat{f}\), is an It\^{o} integral that is a mean zero (square integrable) martingale. Suppose the initital
condition \(\hat{\bs{r}}_0\) is deterministic; call it \(\hat{\bs{x}}\), as in Section~\ref{sec-gen}. Compute the mathematical expectation of both sides of Equation~(\ref{sg2}) (conditioned on
\(\hat{\bs{r}}_0=\hat{\bs{x}}\)) and then differentiate with respect to \(t\) at \(t=0\). Noting that
\(\frac{d}{dt}\langle\negthinspace\langle\hat{\bs{m}}\rangle\negthinspace\rangle(0)=\hat{\bs{D}}(\hat{\bs{x}})\), we get\footnote{ \( \frac{d}{dt}\hat{f}(\hat{\bs{x}}) \) means
\(\frac{d}{dt}\hat{f}(\hat{\bs{r}}(t))\big|_{t=0}  \).}
\begin{equation}\label{sg3}
\frac{d}{dt}\hat{f}(\hat{\bs{x}})=
\frac{1}{2}\sum_{l,m=1}^{2d} \hat{D}_{lm}(\hat{\bs{x}})\bigg(\frac{\partial^2}{\partial \hat{x}_l\,\partial
\hat{x}_m}\hat{f}\bigg)(\hat{\bs{x}}).\\
\end{equation}
 The partial differential operator on the right-hand side of Equation~(\ref{sg3}) is the
generator of the pair-process semigroup \(\{\hat{T}_t:t\geq0 \}\)  given in Equation~(\ref{tp2}) by \((\hat{T}_t \hat{f})(\hat{\bs{x}})\equiv
E_{\hat{\bs{x}}}[\hat{f}(\hat{\bs{r}}(t,\hat{\bs{r}}_0))]\), where \(\hat{f}\) is now a bounded and measurable function on \(\mathbb{R}^{2d}\) and \(\hat{\bs{x}}\) is any vector in
\(\mathbb{R}^{2d}\).

\section{Diffusion Matrix for the Maxwell Kernel}\label{app-b}

Here are the essential details of the computation of the  covariance matrix function \(\bs{D}=\bs{D}^{12}\), defined in Equation~(\ref{qvqrd}) for \(\alpha,\beta=1,2\), when the underlying velocity
distribution of the small particles is a Maxwell distribution.

The \hyperlink{maxker}{Maxwell kernel} is  $\bs{g}_{\varepsilon}(\bs{r}) = \kappa_{\varepsilon,d}\bs{r}\exp\big(-\frac{|r|^2}{2\varepsilon}\big)$, as defined in Equation~(\ref{kernel1}). We need to
show that\footnote{Observe that the radial component of \(\bs{D}_{\varepsilon}^{12}\) is positive if \(|\bs{x}|<\sqrt{\varepsilon}\) and negative if \(|\bs{x}|>\sqrt{\varepsilon}\), where
\(\sqrt{\varepsilon}\) is the correlation length. The transverse component of \(\bs{D}_{\varepsilon}^{12}\) is always positive.}
\begin{equation}\label{eq5.82n}
\begin{split}
 \bs{D}_{\varepsilon}^{12} (\sqrt{2}\bs{x}  )&=\frac{1}{2}\kappa_{\varepsilon,d}^2(\pi \varepsilon)^{\frac{d}{2}}e^{-
    \frac{|\bs{x}|^2}{2\varepsilon}}\bigg(\varepsilon\bs{1}-|\bs{x}|^2\bs{P}(\bs{x})\bigg)\\
    &=\frac{\varepsilon}{2}\kappa_{\varepsilon,d}^2(\pi \varepsilon)^{\frac{d}{2}}e^{-
    \frac{|\bs{x}|^2}{2\varepsilon}}\bigg(\Big(1-\frac{|\bs{x}|^2}{\varepsilon}\Big)\bs{P}(\bs{x})+\bs{P}^{\perp}(\bs{x})\bigg).
\end{split}
\end{equation}
Equivalently, in component form,
\begin{equation}\label{eq5.82}
 \begin{array}{c}
  D_{\varepsilon, k l}^{12} (\sqrt{2}\bs{x}  ):=
 \frac{1}{2}\kappa_{\varepsilon,d}^2(\pi \varepsilon)^{\frac{d}{2}}e^{-
\frac{|\bs{x}|^2}{2\varepsilon}}\begin{cases}
  -x_k x_{l},&\text{if}\qquad k \neq l,\\ (\varepsilon- x^2_k),&\text{if}\qquad k = l.
\end{cases}\\
\end{array}
\end{equation}

Start with Equation~(\ref{qvqrd}) in component form:
\begin{equation}\label{app-b1-1}
D_{\varepsilon,k l}^{12} (\hat{\bs{r}}) =\int_{\mathbb{R}^d}g_{\varepsilon,k}(\bs{r}^1-\bs{q})g_{\varepsilon,l}(\bs{r}^2-\bs{q})\,d\bs{q}.
\end{equation}
The first observation is that for $\alpha = 1,2$
\begin{equation}\label{app-b1-2}
\begin{array}{c}
{g_{\varepsilon,l}(\bs{r}^{\alpha}-\bs{q}) = \kappa_{\varepsilon,d}(r_{l}^{\alpha}-q_{l})\exp(- \frac{|\bs{r}^{\alpha}-\bs{q}|^2}{2\varepsilon}) }\\ \\
{= \kappa_{\varepsilon,d}\varepsilon {\frac{\partial}{\partial q_{l}}}\exp(- \frac{|\bs{r}^{\alpha}-\bs{q}|^2}{2\varepsilon}) .}
\end{array}
\end{equation}
Hence, using integration-by-parts,
\begin{equation}\label{app-b1-3}
    \begin{split}
D_{\varepsilon,k l}^{12} (\hat{\bs{r}} )
 &= ( \kappa_{\varepsilon,d}\varepsilon)^2\int_{\mathbb{R}^d}\frac{\partial}{\partial q_{k}} \exp(-
\frac{|\bs{r}^1-\bs{q}|^2}{2\varepsilon})
\frac{\partial}{\partial q_{l}}\exp(- \frac{|\bs{r}^2-\bs{q}|^2}{2\varepsilon})\,d\bs{q}\\
&= -( \kappa_{\varepsilon,d}\varepsilon)^2 \int_{\mathbb{R}^d} \exp(- \frac{|\bs{r}^1-\bs{q}|^2}{2\varepsilon})
\frac{\partial^2}{\partial q_k\partial q_{l}}\exp(- \frac{|\bs{r}^2-\bs{q}|^2}{2\varepsilon})\,d\bs{q}  \\
 &= -( \kappa_{\varepsilon,d}\varepsilon)^2\int_{\mathbb{R}^d} \exp(- \frac{|\bs{r}^1-\bs{q}|^2}{2\varepsilon})\\
 &\times\begin{cases}
 \frac{(r_k^2 -q_k)}{\varepsilon}\frac{(r_{l}^2 -q_{l})}{\varepsilon}\exp(- \frac{|\bs{r}^2-\bs{q}|^2}{2\varepsilon})\,d\bs{q},
 \text{ if } k \neq l,\\ (-\frac{1}{\varepsilon} + \frac{(r_k^2 -q_k)^2}{\varepsilon^2})\exp(- \frac{|\bs{r}^2-\bs{q}|^2}{2\varepsilon})\,d\bs{q},
 \text{ if } k = l.
 \end{cases}
    \end{split}
\end{equation}
Next, using shift invariance, we obtain for $k \neq l$
\begin{equation}\label{app-b1-4}
 D_{\varepsilon,k l}^{12} (\hat{\bs{r}} ) = - \kappa_{\varepsilon,d}^2\int_{\mathbb{R}^d} \exp(- \frac{|\bs{r}^1- \bs{r}^2 -\bs{q}|^2 + |\bs{q}|^2}{2\varepsilon})q_kq_{l}\,d\bs{q}
\end{equation}
and for $k = l$

\begin{equation}\label{app-b1-5}
D_{\varepsilon,k l}^{12} (\hat{\bs{r}} )= -\kappa_{\varepsilon,d}^2\int_{\mathbb{R}^d} \exp(- \frac{|\bs{r}^1- \bs{r}^2 -\bs{q}|^2 + |\bs{q}|^2}{2\varepsilon})(-\varepsilon + q_k^2)\,d\bs{q} .
\end{equation}
Now if we use the identity
\begin{equation}\label{app-b1-6}
 \exp(-\frac{|\bs{r}-\bs{q}| + |\bs{q}|^2}{2\varepsilon})=\exp(- \frac{|\bs{r}|^2}{4\varepsilon})\exp(-\frac{|\bs{q}-\frac{1}{2}\bs{r}|^2}{\varepsilon}),
 \end{equation}
we obtain for $k \neq l$
\begin{equation}\label{app-b1-7}
{ D_{\varepsilon,k l}^{12} (\hat{\bs{r}} ) = - \kappa_{\varepsilon,d}^2 2\exp(- \frac{|\bs{r}^1- \bs{r}^2|^2}{4\varepsilon})\int_{\mathbb{R}^d} \exp(- \frac{|\frac{1}{2}(\bs{r}^1- \bs{r}^2)
-\bs{q}|^2 }{\varepsilon})q_kq_{l}\, d\bs{q}}
\end{equation}
 and for $k = l$
\begin{equation}\label{app-b1-8}
{ D_{\varepsilon,k l}^{12} (\hat{\bs{r}} )= -\kappa_{\varepsilon,d}^2\exp(- \frac{|\bs{r}^1- \bs{r}^2|^2}{4\varepsilon})\int_{\mathbb{R}^d} \exp(- \frac{|\frac{1}{2}(\bs{r}^1- \bs{r}^2) -\bs{q}|^2
}{\varepsilon})(-\varepsilon + q_k^2)\,d\bs{q}.}
\end{equation}
A standard integration then yields
\begin{equation}\label{app-b1-9}
D_{\varepsilon,k l}^{12} (\hat{\bs{r}} ):= \begin{cases}
 -\frac{1}{2}\kappa^2_{\varepsilon,d}(\pi \varepsilon)^{\frac{d}{2}}\exp(- \frac{|\bs{r}^1 - \bs{r}^2|^2}{4\varepsilon})\cdot\frac{(r_k^1 - r_k^2)(r_{l}^1 - r_{l}^2)}{2} ,&\text{if}\qquad
k \neq l,\\ -\frac{1}{2}\kappa^2_{\varepsilon,d}(\pi \varepsilon)^{\frac{d}{2}}\exp(- \frac{|\bs{r}^1 - \bs{r}^2|^2}{4\varepsilon})\cdot(-\varepsilon + \frac{(r_k^1 - r_k^2)^2}{2}),&\text{if}\qquad
k = l.
\end{cases}
\end{equation}
Finally, using $\bs{r}^1-\bs{r}^2=\sqrt{2}\bs{x}$, we obtain Equation~(\ref{eq5.82}) or, equivalently, Equation~(\ref{eq5.82n}).

\section{General Structure of the Diffusion Matrix}\label{app-c}

\subsection{Proof of Lemma~\ref{posdef}}

Fix \(\bs{x}\not =\bs{0}\), and let \bs{a} be an arbitrary constant vector. Compute the quadratic form \(\bs{a}^T(\bs{C}\pm\bs{D}(\bs{x}))\bs{a}\). Referring to Equation~(\ref{qvqrd}), this can be
cast in the form
\begin{equation}\label{gen30}
\bs{a}^T(\bs{C}\pm\bs{D}(\bs{x}))\bs{a}  =\int_{\mathbb{R}^{d}}\big(\psi^2(-\bs{q})\pm \psi(-\sqrt{2}\bs{x}-\bs{q})\psi(-\bs{q})\big)\,d\bs{q},
\end{equation}
where \(\psi\) here denotes the function \(\psi:=\bs{a}^T\bs{g}\). Observe that translation-invariance implies
\(\int_{\mathbb{R}^{d}}\psi^2(-\bs{q})\,d\bs{q}=\int_{\mathbb{R}^{d}}\psi^2(-\sqrt{2}\bs{x}-\bs{q})\,d\bs{q}\). Hence, Equation~(\ref{gen30}) is the same as
\begin{equation}\label{gen31}
    \begin{split}
    \bs{a}^T(\bs{C}\pm\bs{D}(\bs{x}))\bs{a}  &=\frac{1}{2}\bigg[\int_{\mathbb{R}^{d}}\big(\psi^2(-\bs{q})\pm 2 \psi(-\sqrt{2}\bs{x}-\bs{q})\psi(-\bs{q})\\
                                                        &+\psi^2(-\sqrt{2}\bs{x}-\bs{q})\big)\,d\bs{q}\bigg]\\
                                                        &=\frac{1}{2}\bigg[\int_{\mathbb{R}^{d}}\bigg(\psi(-\bs{q})\pm \psi(-\sqrt{2}\bs{x}-\bs{q})   \bigg)^2\,d\bs{q}  \bigg].
    \end{split}
\end{equation}
So \(\bs{a}^T(\bs{C}\pm\bs{D}(\bs{x}))\bs{a}\geq 0\), for all \(\bs{a}\not =\bs{0}\).

Next suppose that for some \(\hat{\bs{a}}\not =\bs{0}\), \(\hat{\bs{a}}^T(\bs{C}\pm\bs{D}(\bs{x}))\hat{\bs{a}}=0\). Then \(\psi(-\bs{q})\pm \psi(-\sqrt{2}\bs{x}-\bs{q})=0\) for almost
every~\(\bs{q}\). That is, the function \(\psi\) is periodic, with period either \(\sqrt{2}\bs{x}\) or \(2\sqrt{2}\bs{x}\). Since \(\psi=\hat{\bs{a}}^T\bs{g}\) must also be integrable over
\(\mathbb{R}^d\), it follows that the function \(\psi=\hat{\bs{a}}^T\bs{g}\) is zero almost everywhere in \(\mathbb{R}^d\). Thus, \(\hat{\bs{a}}\bullet\bs{g}(\bs{Q}\bs{q})=0\) for any orthogonal
transformation \bs{Q} and almost~every vector \bs{q}. Since \bs{g} is \hyperlink{isotropic}{isotropic,} we get \(\hat{\bs{a}}\bullet\bs{Q}\bs{g}(\bs{q})=0\) or, equivalently
\(\bs{Q}^T\hat{\bs{a}}\bullet\bs{g}(\bs{q})=0\). But \bs{Q} is any orthogonal transformation, so \(\bs{g}(\bs{q})=\bs{0}\) for almost~every \bs{q}, which contradicts our assumption.

\subsection{Detailed Structure of the Matrix Function \bs{D}}\label{detstruct} Here we consider the matrix function \(\bs{x}\mapsto\bs{D}(\bs{x}) \) defined in Equation~(\ref{d1}) of Section~\ref{pairdiffsec}. There, in Equations~(\ref{d3}-\ref{d9}) we saw that
\begin{equation}\label{app-c-1}
    \begin{split}
    \bs{D}(\sqrt{2}\bs{y})  &=\bs{A}(\bs{y})-|\bs{y}|^2 \beta(\bs{y})\bs{P}(\bs{y})\\
                            &=\alpha_{\perp}(|\bs{y}|^2)\bs{P}_{\perp}(\bs{y})+(\alpha(|\bs{y}|^2)-|\bs{y}|^2\beta(|\bs{y}|^2))\bs{P}(\bs{y}),
    \end{split}
\end{equation}
for scalar functions \(\xi\mapsto\alpha(\xi^2),\alpha_{\perp}(\xi^2),\beta(\xi^2)\). We provide explicit formulas for these functions in terms of the underlying \hyperlink{forcingfunction}{forcing
function} \(\phi\). The following definitions will help simplify the process.
\begin{define}\label{app-c-def-1}
For vectors \(\bs{y},\bs{q}\) in \(\mathbb{R}^d\), define scalars \(\zeta,\xi,\varphi\) through
\begin{align}\label{app-c-def-1a}
    &\xi:=|\bs{y}|\qquad\zeta:=|\bs{q}|\qquad \xi\zeta\cos\varphi:=\bs{y}\bullet\bs{q}
    \intertext{and set}
    &a:=\xi^2+\zeta^2\qquad b:=2\xi\zeta.
\end{align}
\end{define}
It follows from these definitions that
\begin{equation}\label{app-c-def-2}
    \begin{split}
   a\geq b\geq 0\quad\text{and}\quad a=b\quad\text{iff}\quad\xi=\zeta\\
   0\leq a-b=(\xi-\zeta)^2\qquad0\leq a+b=(\xi+\zeta)^2
   \end{split}
\end{equation}
The next definition will also prove useful.
\begin{define}\label{app-c-def-3}
For \(\xi,\zeta\geq 0\) and \(0\leq\varphi\leq\pi\),
\begin{align}\label{app-cdef-4}
   \Phi(\xi,\zeta,\varphi)  &:=\phi(a-b\cos\varphi)\phi(a+b\cos\varphi)\\
                            &=\phi(|\bs{q}-\bs{y}|^2)\phi(|\bs{q}+\bs{y}|^2)
\end{align}
\end{define}

\begin{lem}\label{lem-c-1}
The functions \(\xi\mapsto\alpha(\xi^2),\alpha_{\perp}(\xi^2),\beta(\xi^2)\) are given through the integrals:
\begin{align}\label{lem-c-1a}
\alpha(\xi^2)   &=2\omega_{d-2}\int_0^{\infty}\int_0^{\pi/2}\bigg(\zeta^2\Phi(\xi,\zeta,\varphi)\cos^2\varphi\bigg) \zeta^{d-1}\sin^{d-2}\varphi\,d\varphi\,d\zeta,\\
\alpha_{\perp}(\xi^2)   &=2\omega_{d-2}\int_0^{\infty}\int_0^{\pi/2}\bigg(\zeta^2\Phi(\xi,\zeta,\varphi)\frac{\sin^2\varphi}{d-1}\bigg) \zeta^{d-1}\sin^{d-2}\varphi\,d\varphi\,d\zeta,\\
\beta(\xi^2)   &=2\omega_{d-2}\int_0^{\infty}\int_0^{\pi/2}\Phi(\xi,\zeta,\varphi)\, \zeta^{d-1}\sin^{d-2}\varphi\,d\varphi\,d\zeta.
\end{align}
In paticular, we have
\begin{align}\label{lem-c-10}
\alpha(0)=\alpha_{\perp}(0)&=\frac{\omega_{d-1}}{d}\int_0^{\infty}\zeta^2\phi^2(\zeta^2)\zeta^{d-1}\,d\zeta\\
\beta(0)&=\omega_{d-1}\int_0^{\infty}\phi^2(\zeta^2)\zeta^{d-1}\,d\zeta
\end{align}
\end{lem}

\noindent\begin{proof}[Proof of Theorem~\ref{d10a6}] Recall that the eigenvalues \(\sigma_{\perp}(\xi^2),\sigma(\xi^2)\) are given through Equation~(\ref{d10a}) from which it follows that
\(\sigma_{\perp}(\xi^2)\leq\sigma(\xi^2)\) if, and only if, \(\alpha(\xi^2)-\alpha_{\perp}(\xi^2)\leq \xi^2\beta(\xi^2) \). Since \(\beta(\xi^2)\geq0\), \(\sigma_{\perp}(\xi^2)\leq\sigma(\xi^2)\)
will certainly hold whenever \(\alpha(\xi^2)-\alpha_{\perp}(\xi^2)\leq 0 \). Thus, \(\tilde{\bs{D}}\) is radially dominance whenever \(\alpha(\xi^2)-\alpha_{\perp}(\xi^2)\leq 0 \). The proof is
complete upon proving Lemma~\ref{app-c-lem2} below, which asserts that the logarithmic convexity of \(\phi\) implies \(\alpha(\xi^2)-\alpha_{\perp}(\xi^2)\leq 0 \).
\end{proof}
\begin{lem}\label{app-c-lem2}
If  \(\phi(\xi^2)\phi''(\xi^2)\leq (\phi'(\xi^2))^2,\,\forall\xi>0\), then \(\alpha(\xi^2)-\alpha_{\perp}(\xi^2)\leq 0,\,\forall\xi>0 \). The former condition is equivalent to
\((\ln\phi(z))''\leq0,\,\forall z\geq 0\); that is, \(\phi\) is logarithmically concave.\footnote{Recall that we have always assumed that \(\phi>0\) and \(\phi'\leq0\).}
\end{lem}
The following computational result will prove useful in the proofs of Lemmas~\ref{lem-c-1} and~\ref{app-c-lem2}. We omit its elementary proof.
\begin{lem}\label{fiform}
Let \((\bs{y},\bs{q})\mapsto \gamma(|\bs{y}|^2,|\bs{q}|^2,\bs{y}\bullet\bs{q})\); so \(\gamma\) is an \hyperlink{isotropic}{isotropic} scalar-valued function of \bs{y} and \bs{q}. For fixed \bs{y},
suppose \(\bs{q}\mapsto |\gamma(|\bs{y}|^2,|\bs{q}|^2,\bs{y}\bullet\bs{q})|\) is integrable over \(\mathbb{R}^d\). Then, for \(d\geq 2\),\footnote{Recall that \(\omega_{d-2}\) denotes the surface
area of \(\mathbb{S}^{d-2}\).}
\begin{equation}\label{lem-c-1-1}
\int_{\mathbb{R}^d}\gamma(|\bs{y}|^2,|\bs{q}|^2,\bs{y}\bullet\bs{q})\,d\bs{q}=\omega_{d-2}\int_0^{\infty}\int_0^{\pi}\gamma(\xi^2,\zeta^2,\xi\zeta\cos\varphi)\zeta^{d-1}\sin^{d-2}\varphi\,d\varphi\,d\zeta.
\end{equation}
\end{lem}

\noindent\begin{proof}[Proof of Lemma~\ref{lem-c-1}]
 From Equations~(\ref{d5}) and~(\ref{d9}), together with Equation~(\ref{lem-c-1-1}) of Lemma~\ref{fiform}, we see at once that
\begin{equation}\label{lem-c-1-2}
    \begin{split}
\beta(\xi^2)    &=\int_{\mathbb{R}^d}\phi(|\bs{y}|^2+|\bs{q}|^2+2|\bs{y}||\bs{q}|\cos\varphi)\phi(|\bs{y}|^2+|\bs{q}|^2-2|\bs{y}||\bs{q}|\cos\varphi)\,d\bs{q}\\
                &=\omega_{d-2}\int_0^{\infty}\int_0^{\pi}\Phi(\xi,\zeta,\varphi)\, \zeta^{d-1}\sin^{d-2}\varphi\,d\varphi\,d\zeta\\
                &=2\omega_{d-2}\int_0^{\infty}\int_0^{\pi/2}\Phi(\xi,\zeta,\varphi)\, \zeta^{d-1}\sin^{d-2}\varphi\,d\varphi\,d\zeta.
    \end{split}
\end{equation}
where
\begin{equation}\label{lem-c-1-3}
\Phi(\xi,\zeta,\varphi)=\phi(\xi^2+\zeta^2+2\xi\zeta\cos\varphi)\phi(\xi^2+\zeta^2-2\xi\zeta\cos\varphi).
\end{equation}
The last step follows from the fact that \(\Phi(\xi,\zeta,\varphi)\) is even about \(\varphi=\pi/2\).

Next consider the decomposition
\begin{equation}\label{lem-c-1-4}
    \begin{split}
    \bs{A}(\bs{y})  &=\bs{P}^{\perp}(\bs{y})\bs{A}(\bs{y})\bs{P}^{\perp}(\bs{y})+\bs{P}(\bs{y})\bs{A}(\bs{y})\bs{P}(\bs{y})\\
                    &=\alpha_{\perp}(|\bs{y}|^2)\bs{P}^{\perp}(\bs{y})+\alpha(|\bs{y}|^2)\bs{P}(\bs{y}).
    \end{split}
\end{equation}
This implies immediately that
\begin{equation}\label{lem-c-1-5}
    \begin{split}
    \alpha(|\bs{y}|^2)          &=trace\big(\bs{P}(\bs{y})\bs{A}(\bs{y})\bs{P}(\bs{y})\big)=\bs{A}(\bs{y})\bullet\bs{P}(\bs{y})\\
    \alpha_{\perp}(|\bs{y}|^2)  &=\frac{1}{d-1}trace\big(\bs{P}^{\perp}(\bs{y})\bs{A}(\bs{y})\bs{P}^{\perp}(\bs{y})\big)=\frac{1}{d-1}\bs{A}(\bs{y})\bullet\bs{P}^{\perp}(\bs{y}).
    \end{split}
\end{equation}
Now use the formulas
\begin{equation}\label{lem-c-1-6}
    \begin{split}
    |\bs{P}(\bs{y})\bs{q}|^2 &=trace\big(\bs{P}(\bs{y})\bs{q}\bs{q}^T \bs{P}(\bs{y})\big) =\bigg(\frac{\bs{q}\bullet\bs{y}}{|\bs{y}|}\bigg)^2=\zeta^2\cos^2\varphi\\
    |\bs{P}^{\perp}(\bs{y})\bs{q}|^2 &=trace\big(\bs{P}^{\perp}(\bs{y})\bs{q}\bs{q}^T \bs{P}^{\perp}(\bs{y})\big) =|\bs{q}|^2-\bigg(\frac{\bs{q}\bullet\bs{y}}{|\bs{y}|}\bigg)^2=\zeta^2\sin^2\varphi,
    \end{split}
\end{equation}
together with the definition of \(\bs{A}(\bs{y})\) in Equation~(\ref{d4}) and Equation~(\ref{lem-c-1-1}) of Lemma~\ref{fiform}, to conclude the first two formulas in Lemma~\ref{lem-c-1}
\end{proof}

It is easy to see that the inequality \(\sigma_{\perp}(\xi^2)<\sigma(\xi^2)\), for the eigenvalues of the \hyperlink{pairmatrix}{diffusion matrix} \(\tilde{\bs{D}}\), is equivalent to the inequality
\(\alpha(\xi^2)-\alpha_{\perp}(\xi^2)< \xi^2\beta(\xi^2)\). Hence, the result of Lemma~\ref{lem-c-1} indirectly supplies a criterion that \(\sigma_{\perp}(\xi^2)<\sigma(\xi^2)\) be satisfied. As
this is too unwieldy,  we provide a simple sufficient condition that \(\alpha(\xi^2)-\alpha_{\perp}(\xi^2)\leq 0\). Since \(\beta(\xi^2)>0\), this will in turn guarantee that
\(\sigma_{\perp}(\xi^2)<\sigma(\xi^2)\) holds; that is, the \hyperlink{pairmatrix}{diffusion matrix} is radially dominant, which is the content of Theorem~\ref{d10a6}.

\begin{proof}[Proof of Lemma~\ref{app-c-lem2}]
Use the first two formulas of Lemma~\ref{lem-c-1} and an integration by parts to get
\begin{multline}\label{lem-c-3-1}
\alpha(\xi^2)-\alpha_{\perp}(\xi^2)\\
=2\omega_{d-2}\int_0^{\pi/2}\Bigg(\int_0^{\infty}\zeta^{d+1}\Phi(\xi,\zeta,\varphi)\,d\zeta\Bigg)\bigg(\cos^2\varphi -\frac{1}{d-1}\sin^2\varphi\bigg)\sin^{d-2}\varphi\,d\varphi\\
=-2\omega_{d-2}\int_0^{\pi/2}\Bigg(\int_0^{\infty}\zeta^{d+1}\frac{\partial}{\partial\varphi}\Phi(\xi,\zeta,\varphi)\,d\zeta\Bigg)\bigg(\frac{1}{d-1}\cos\varphi\sin^{d-1}\varphi\bigg)\,d\varphi.
\end{multline}
Thus, \(\alpha(\xi^2)-\alpha_{\perp}(\xi^2)\leq 0\) whenever \(\frac{\partial}{\partial\varphi}\Phi(\xi,\zeta,\varphi)\geq 0\), For \(\xi,\zeta\geq 0\) and \(0\leq\varphi\leq\pi/2\). But
\begin{equation}\label{lem-c-3-2}
\frac{\partial}{\partial\varphi}\Phi(\xi,\zeta,\varphi)=-2\xi\zeta\sin\varphi\Phi(\xi,\zeta,\varphi)(\ln\phi(z))''\big|_{z=\xi^2+\zeta^2+\eta 2\xi\zeta\cos\varphi}
\end{equation}
for some \(\eta\in[-1,1]\). Hence, \((\ln\phi)''\leq 0\) guarantees that  \(\frac{\partial}{\partial\varphi}\Phi(\xi,\zeta,\varphi)\geq 0\)
\end{proof}\\

\section{van~Kampen's $1-$Dimensional Flux Rate}\label{app-b2}
van~Kampen~\cite{KA}(Ch.VIII.1) considers one-dimensional diffusions with drift \(a(\xi)\) and quadratic variation \(b(\xi)\geq0\), where both coefficients are bounded with continuous derivatives.
The Fokker-Planck equation for such diffusions is
\begin{equation}\label{b2-1}
\frac{\partial}{\partial t}X(\xi,t)=-\frac{\partial}{\partial \xi}\big(a(\xi)X(\xi,t)  \big)+\frac{1}{2}\frac{\partial^2}{\partial \xi^2}\big( b(\xi)X(\xi,t) \big),
\end{equation}
where \(X(\xi,t)\) is the probability density at the point \(\xi\) and time \(t\). van~Kampen then defines the \textit{probability flux} of the one-dimensional diffusion to be
\begin{equation}\label{b2-2}
\dot{J}(\xi,t):=a(\xi)X(\xi,t)-\frac{1}{2}\frac{\partial}{\partial \xi}\big( b(\xi)X(\xi,t)\big).
\end{equation}
If the probability density, \(X(\xi,t)\), is locally spatially constant, this reduces to
\begin{equation}\label{b2-3}
\dot{J}(\xi,t):=\bigg(a(\xi)-\frac{1}{2}\frac{\partial}{\partial\xi}\big( b(\xi)\big)\bigg)X(\xi,t).
\end{equation}
From the one-dimensional generator \(\dot{A}\) in Equation~(\ref{mag5a}), we see that \(a(\xi)=\frac{1}{2}\frac{d-1}{\xi}\sigma_{\perp}(\xi^2)\) and \(b(\xi)=\sigma(\xi^2)\).\footnote{The
differential operator in the right-hand-side of the Fokker-Planck Equation~(\ref{b2-1}) is the adjoint of the generator \(\dot{A}\).} If we apply Equation~(\ref{b2-3}) above we recover the
\(1-\)dimensional van-Kampen flux rate of Equation~(\ref{dir0}).



\newpage
\bibliographystyle{plain}
\bibliography{depletion}

\end{document}